\documentclass{lmcs}
\pdfoutput=1

\usepackage{lastpage}
\lmcsdoi{20}{3}{26}
\lmcsheading{}{\pageref{LastPage}}{}{}%
{Jul.~27,~2023}{Sep.~23,~2024}{}


\usepackage[T1]{fontenc}
\usepackage[utf8]{inputenc}
\usepackage{amsmath}
\usepackage{amssymb}
\usepackage{amsthm}
\usepackage{bbold} 
\usepackage{graphicx}
\usepackage{hyperref}
\usepackage{lineno}
\usepackage{listings}
\usepackage{mathpartir} 
\usepackage{stmaryrd}
\usepackage{mathtools} 
\usepackage{thmtools}
\usepackage{upquote}
\usepackage{xcolor}
\usepackage[all]{xy}
\usepackage{fontawesome5} 
\usepackage[most]{tcolorbox}

\begin{document}




\newcommand{\defeq}{\mathrel{\overset{\text{\tiny def}}{=}}} 

\newcommand{\pl}[1]{\textsc{#1}} 

\newcommand{\lambdaAEff}{$\lambda_{\text{\ae}}$} 

\newcommand{\bnfis}{\mathrel{\;{:}{:}{=}\ }}
\newcommand{\bnfor}{\mathrel{\,\;\big|\ \ \!}}



\newcommand{\One}{\mathbb{1}} 
\newcommand{\one}{\star} 
\newcommand{\Zero}{\mathbb{0}} 

\newcommand{\Bool}{\mathbb{B}} 
\newcommand{\true}{\mathbf{true}} 
\newcommand{\false}{\mathbf{false}} 

\newcommand{\expto}{\Rightarrow} 
\newcommand{\lam}[1]{\lambda #1 \,.\,} 
\newcommand{\pair}[2]{\langle #1 , #2 \rangle} 

\newcommand{\lifted}[1]{#1_\bot} 
\newcommand{\idte}[4]{\mathbf{ifdef}~#1~\mathbf{then}~#2 \mapsto #3~\mathbf{else}~#4} 

\newcommand{\ite}[3]{\mathbf{if}~#1~\mathbf{then}~#2~\mathbf{else}~#3} 


\newcommand{\Tree}[2]{\mathrm{Tree}_{#1}\left(#2\right)} 
\newcommand{\retTree}[1]{\mathsf{return}\,#1} 

\newcommand{\opsym}[1]{\mathsf{#1}} 
\newcommand{\op}{\opsym{op}} 

\newcommand{\sig}{\Sigma} 

\renewcommand{\o}{o} 
\renewcommand{\i}{\iota} 
\newcommand{\s}{\varsigma} 

\newcommand{\opincomp}[2]{{\mathsf{#1}}\,{\tmkw{\downarrow}}\,#2} 
\newcommand{\opincompp}[2]{{\mathsf{#1}}\,{\tmkw{\downarrow\downarrow}}\,#2} 

\newcommand{\eq}{\mathrm{Eq}} 

\newcommand{\FreeAlg}[2]{\mathrm{Free}_{#1}\left(#2\right)} 
\newcommand{\lift}[1]{#1^\dagger} 
\newcommand{\freelift}[1]{#1^\ddagger} 

\newcommand{\M}{\mathcal{M}} 
\newcommand{\Mcarrier}{\vert \mathcal{M} \vert} 

\newcommand{\T}{T} 


\newcommand{\sigget}{\mathsf{get}}
\newcommand{\sigset}{\mathsf{set}}


\newcommand{\ctxlock}{\text{\scriptsize\faUnlock}} 


\newcommand{\at}{\mathbin{!}} 
\newcommand{\att}{\mathbin{!!}} 


\newcommand{\tysym}[1]{\mathsf{#1}}
\newcommand{\tybase}{\tysym{b}} 
\newcommand{\tyunit}{\tysym{1}} 
\newcommand{\tyint}{\tysym{int}} 
\newcommand{\tystring}{\tysym{string}} 
\newcommand{\tylist}[1]{\tysym{list}~\tysym{#1}} 
\newcommand{\tyempty}{\tysym{0}} 
\newcommand{\typrod}[2]{#1 \times #2} 
\newcommand{\tysum}[2]{#1 + #2} 
\newcommand{\tyfun}[2]{#1 \to #2} 
\newcommand{\typromise}[1]{\langle #1 \rangle} 
\newcommand{\tybox}[1]{[#1]} 


\newcommand{\tycomp}[2]{#1 \at #2} 


\newcommand{\tyrun}[3]{#1 \att (#2,#3)} 
\newcommand{\typar}[2]{#1 \mathbin{\tmkw{\vert\vert}}  #2} 
\newcommand{\tyC}{C} 
\newcommand{\tyD}{D} 


\newcommand{\tm}[1]{\mathsf{#1}} 
\newcommand{\tmkw}[1]{\tm{\color{keywordColor}#1}} 

\newcommand{\tmpromise}[1]{\langle #1 \rangle} 

\newcommand{\tmconst}[1]{\tm{#1}}
\newcommand{\tmunit}{()} 
\newcommand{\tmpair}[2]{( #1 , #2 )} 
\newcommand{\tminl}[2][]{\tmkw{inl}_{#1}\,#2} 
\newcommand{\tminr}[2][]{\tmkw{inr}_{#1}\,#2} 
\newcommand{\tmfun}[2]{{\mathop{\tmkw{fun}}}\; (#1) \mapsto #2} 
\newcommand{\tmfunano}[2]{{\mathop{\tmkw{fun}}}\; #1 \mapsto #2} 
\newcommand{\tmapp}[2]{#1\,#2} 

\newcommand{\tmbox}[1]{[#1]} 
\newcommand{\tmunbox}[3]{\tmkw{unbox}\; #1\; \tmkw{as}\; \tmbox{#2}\; \tmkw{in}\; #3} 
\newcommand{\tmunboxgen}[1]{\tmkw{unbox}\; #1} 

\newcommand{\tmreturn}[2][]{\tmkw{return}_{#1}\, #2} 
\newcommand{\tmlet}[3]{\tmkw{let}\; #1 = #2 \;\tmkw{in}\; #3} 
\newcommand{\tmletrec}[5][]{\tmkw{let}\;\tmkw{rec}\; #2\; #3 #1 = #4 \;\tmkw{in}\; #5} 

\newcommand{\tmop}[4]{\tm{#1}\;(#2, #3. #4)} 
\newcommand{\tmopin}[3]{\tmkw{\downarrow}\, \tm{#1}\,(#2, #3)} 
\newcommand{\tmopout}[3]{\tmkw{\uparrow}\,\tm{#1}\, (#2, #3)} 
\newcommand{\tmopoutbig}[3]{\tmkw{\uparrow}\,\tm{#1}\, \big(#2, #3\big)} 
\newcommand{\tmopoutgen}[2]{\tmkw{\uparrow}\,\tm{#1}\, #2} 

\newcommand{\tmmatch}[3][]{\tmkw{match}\;#2\;\tmkw{with}\;\{#3\}_{#1}} 

\newcommand{\tmawait}[3]{\tmkw{await}\;#1\;\tmkw{until}\;\tmpromise{#2}\;\tmkw{in}\;#3} 
\newcommand{\tmawaitgen}[1]{\tmkw{await}\;#1} 

\newcommand{\tmwith}[5]{\tmkw{promise}\; (\tm{#1}\; #2 \mapsto #3)\; \tmkw{as}\; #4\; \tmkw{in}\; #5} 
\newcommand{\tmwithre}[6]{\tmkw{promise}\; (\tm{#1}\; #2 \, #3 \mapsto #4)\; \tmkw{as}\; #5\; \tmkw{in}\; #6} 
\newcommand{\tmwithrest}[9][]{\tmkw{promise}\; (\tm{#2}\; #3 \, #4 \, #5 \mapsto #6)\; \tmkw{@}_{#1}\; #7 \; \tmkw{as}\; #8\; \tmkw{in}\; #9} 

\newcommand{\tmwithgen}[3]{\tmkw{promise}\; (\tm{#1}\; #2 \mapsto #3)} 
\newcommand{\tmwithregen}[4]{\tmkw{promise}\; (\tm{#1}\; #2 \, #3 \mapsto #4)} 
\newcommand{\tmwithregenst}[7][]{\tmkw{promise}\; (\tm{#2}\; #3 \, #4 \, #5 \mapsto #6)\; \tmkw{@}_{#1}\; #7} 

\newcommand{\tmrun}[1]{\tmkw{run}\; #1} 
\newcommand{\tmpar}[2]{#1 \mathbin{\tmkw{\vert\vert}} #2} 

\newcommand{\tmspawn}[2]{\tmkw{spawn}\;(#1,#2)} 
\newcommand{\tmspawngen}[1]{\tmkw{spawn}\;#1} 


\newcommand{\reduces}{\leadsto} 
\newcommand{\tyreduces}{\rightsquigarrow} 

\newcommand{\E}{\mathcal{E}} 
\renewcommand{\H}{\mathcal{H}} 
\newcommand{\F}{\mathcal{F}} 


\newcommand{\types}{\vdash} 
\newcommand{\of}{\mathinner{:}} 

\newcommand{\sub}{\sqsubseteq} 

\newcommand{\coopinfer}[3]{\inferrule*[Lab={\color{rulenameColor}#1}]{#2}{#3}}


\newcommand{\CompResult}[2]{\mathsf{CompRes}\langle#1 \,\vert\, #2\rangle} 
\newcommand{\RunResult}[2]{\mathsf{RunRes}\langle#1 \,\vert\, #2\rangle} 

\newcommand{\Result}[2]{\mathsf{Res}\langle#1 \,\vert\, #2\rangle} 

\newcommand{\ProcResult}[1]{\mathsf{ProcRes}\langle #1 \rangle} 
\newcommand{\ParResult}[1]{\mathsf{ParRes}\langle #1 \rangle} 


\newcommand{\cond}[3]{\mathsf{if}\;#1\;\mathsf{then}\;#2\;\mathsf{else}\;#3} 

\newcommand{\carrier}[1]{\vert #1 \vert} 
\newcommand{\order}[1]{\sqsubseteq_{#1}} 
\newcommand{\lub}[1]{\bigsqcup_n \langle #1 \rangle} 

\newcommand{\Pow}[1]{\mathcal{P}(#1)} 
\newcommand{\sem}[1]{[\![#1]\!]} 

\makeatletter
\providecommand*{\cupdot}{
  \mathbin{%
    \mathpalette\@cupdot{}%
  }%
}
\newcommand*{\@cupdot}[2]{%
  \ooalign{%
    $\m@th#1\cup$\cr
    \hidewidth$\m@th#1\cdot$\hidewidth
  }%
}
\makeatother


\definecolor{redexColor}{rgb}{0.83, 0.83, 0.83} 
\newcommand{\highlightgray}[1]{\tikz[baseline=(X.base)] \node[rectangle, fill=redexColor, rounded corners, inner sep=0.3mm] (X) {$#1$};} 
\newcommand{\highlightwhite}[1]{\tikz[baseline=(X.base)] \node[rectangle, fill=white, rounded corners, inner sep=0.3mm] (X) {$#1$};} 

\definecolor{codegreen}{rgb}{0,0.6,0}
\definecolor{codegray}{rgb}{0.5,0.5,0.5}
\definecolor{codepurple}{rgb}{0.58,0,0.82}
\definecolor{backcolour}{rgb}{0.95,0.95,0.92}

\definecolor{keywordColor}{rgb}{0.0,0.0,0.5} 
\definecolor{rulenameColor}{rgb}{0.5,0.5,0.5} 

\newcommand{\sref}[2]{\hyperref[#2]{#1~\ref{#2}}} 
\newcommand{\srefcase}[3]{\hyperref[#2]{#1~\ref{#2}~(#3)}} 

\def\sectionautorefname{Section}
\def\subsectionautorefname{Section}
\def\subsubsectionautorefname{Section}

\def\lstlanguagefiles{aeff}
\lstset{language=aeff,upquote=true}
\let\ls\lstinline

\newtcolorbox{aeffbox}{
  enhanced, 
  frame hidden, 
  borderline west = {2pt}{0pt}{lightgray}, 
  colback=white, 
  boxsep=0pt, 
  top=0pt, 
  bottom=0pt, 
  left=0pt, 
  right=0pt, 
  toprule=0pt, 
  bottomrule=0pt
}

\makeatletter
\tcbset{
  after app={%
    \ifx\tcb@drawcolorbox\tcb@drawcolorbox@breakable
    \else
      \@endparenv
    \fi
  }
}

\appto\tcb@use@after@lastbox{\@endparenv\@doendpe}
\makeatother

\bibliographystyle{alphaurl}

\title[Higher-Order Asynchronous Effects]{Higher-Order Asynchronous Effects\rsuper*}
\titlecomment{{\lsuper*}This paper is an extended version of our previous work~\cite{Ahman:POPL}: 
it simplifies the meta-theory, removes the reliance on general recursion for reinstalling interrupt handlers, 
adds state to reinstallable interrupt handlers, and extends the calculus with higher-order signal and interrupt 
payloads, and with dynamic process creation.}
\thanks{ 
  This project has received funding from the European Union's Horizon 2020 research and 
  innovation programme under the Marie Sk\l{}odowska-Curie grant agreement No 834146
  \raisebox{-0.05cm}{
    \hspace{-0.15cm}
    \includegraphics[width=0.5cm]{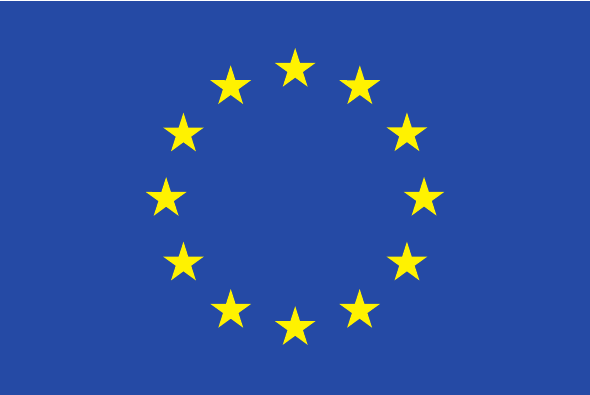}
    \hspace{-0.15cm}
  }.
  This material is based upon work supported by the Air Force Office of Scientific Research under 
  awards number FA9550-17-1-0326 and FA9550-21-1-0024.
}

\author[D.~Ahman]{Danel Ahman\lmcsorcid{0000-0001-6595-2756}}[a]
\author[M.~Pretnar]{Matija Pretnar\lmcsorcid{0000-0001-7755-2303}}[b,c]

\address{University of Tartu, Institute of Computer Science, Narva mnt 18, Tartu, Estonia}
\email{danel.ahman@ut.ee}

\address{University of Ljubljana, Faculty of Mathematics and Physics, Jadranska 19, Ljubljana, Slovenia}
\email{matija.pretnar@fmf.uni-lj.si}

\address{Institute of Mathematics, Physics and Mechanics,  Jadranska 19, Ljubljana, Slovenia}

\begin{abstract}
  \noindent
  We explore asynchronous programming with algebraic effects. We complement their conventional 
  synchronous treatment by showing how to naturally also accommodate asynchrony within them, 
  namely, by decoupling the execution of operation calls into signalling that an operation's implementation 
  needs to be executed, and interrupting a running computation with the operation's result, to which the 
  computation can react by installing interrupt handlers. We formalise these ideas in a small core calculus
  and demonstrate its flexibility using examples ranging from a multi-party web application, to pre-emptive
  multi-threading, to (cancellable) remote function calls, to a parallel variant of runners of algebraic effects.
  In addition, the paper is accompanied by a formalisation of the calculus's type safety proofs in \pl{Agda},
  and a prototype implementation in \pl{OCaml}.
  
\end{abstract}

\maketitle


\section{Introduction}

Effectful programming abstractions are at the heart of many modern general-purpose 
programming languages. They can increase expressiveness by giving programmers access to first-class 
(delimited) continuations, but often they simply help programmers to write cleaner code, e.g., by 
avoiding having to manage a program's memory explicitly in state-passing style, 
or getting lost in callback hell while programming asynchronous computations.

An increasing number of language designers and programmers are starting to 
embrace \emph{algebraic effects}, 
where one uses algebraic operations \cite{Plotkin:NotionsOfComputation} and 
effect handlers \cite{Plotkin:HandlingEffects} to uniformly, modularly, and 
user-definably express a wide range of effectful behaviour, 
ranging from basic examples such as state, rollbacks, exceptions, 
and nondeterminism \cite{Bauer:AlgebraicEffects}, to advanced applications 
in concurrency \cite{DBLP:conf/pldi/Sivaramakrishnan21} and statistical probabilistic programming 
\cite{Bingham:Pyro}, and even quantum computation \cite{Staton:AlgEffQuantum}.

While covering many examples, the conventional treatment of 
algebraic effects is \emph{synchronous} by nature. In it 
effects are invoked by placing operation calls in one's code, 
which then propagate outwards until they trigger the actual effect, finally yielding 
a result to the rest of the computation that has been \emph{waiting} in a blocked state 
the whole time. While blocking the computation is indeed sometimes necessary, e.g., 
in the presence of general effect handlers that can execute their continuation any 
number of times, it forces all uses of algebraic effects to be synchronous, even when this
is not necessary, e.g., when the effect involves executing 
a remote query to which a response is not needed (immediately).

Motivated by the recent interest in the combination of
asynchrony and algebraic effects \cite{Leijen:AsyncAwait,DBLP:conf/pldi/Sivaramakrishnan21}, 
in this paper we explore what it takes to accompany the
synchronous treatment of algebraic effects with
an \emph{asynchronous} one (in terms of
language design, safe programming abstractions, and a 
self-contained core calculus). At the heart of our approach is the 
decoupling of the execution of algebraic operation calls
into (i) \emph{signalling} that some implementation of an operation needs to be executed, and 
(ii) \emph{interrupting} a running computation with its result, to which the computation can 
react by (iii) \emph{installing interrupt handlers}. Importantly, we show that our 
approach is flexible enough that not all signals need to have a
corresponding interrupt, and vice versa, allowing us to also model 
\emph{spontaneous behaviour}, such as a
user clicking a button or the environment pre-empting a thread.

While we are not the first ones to work on asynchrony for algebraic effects, 
the prior work in this area (in the context of general effect handlers) has 
achieved it by simply \emph{delegating} the actual asynchrony to the respective language backends 
\cite{Leijen:AsyncAwait,DBLP:conf/pldi/Sivaramakrishnan21}. In contrast, in this paper 
we demonstrate how to capture the combination of 
asynchrony and algebraic effects in a \emph{self-contained} core calculus. 
It is important to emphasise that our aim is not to replace general effect handlers,  
but instead to \emph{complement} them with robust primitives 
tailored to asynchrony---as we highlight throughout the paper, our proposed approach is 
algebraic by design, so as to be ready for future extensions with general effect handlers.

\paragraph{Paper Structure}
In \autoref{sec:overview}, we give a high-level overview of our approach to 
asynchrony for algebraic effects. 
In Sections~\ref{sec:basic-calculus:computations} 
and \ref{sec:basic-calculus:processes}, we recap our previous work~\cite{Ahman:POPL} on asynchronous algebraic effects
using a core calculus, \lambdaAEff, equipped with a small-step operational semantics and a type-and-effect system. 
In \autoref{sec:higher-order-extensions}, we explore extensions of \lambdaAEff~necessary to accommodate reinstallable 
interrupt handlers, higher-order signal and interrupt payloads, and the dynamic creation of processes,
and prove their type safety.
In \autoref{sec:applications}, we show how these extensions can be used in examples such as pre-emptive multi-threading,
remote function calls, and a parallel variant of runners
of algebraic effects, simplifying the examples in our prior work~\cite{Ahman:POPL}.
We conclude by discussing related and future work in \autoref{sec:conclusion}.

\paragraph{Code}
The paper is accompanied by a \emph{formalisation} of \lambdaAEff's type safety proofs 
in \pl{Agda} \cite{ahman:AeffAgda}, and a \emph{prototype implementation} of \lambdaAEff~in 
\pl{OCaml}, called \pl{{\AE}ff} \cite{pretnar21:AEff}.

In \pl{Agda}, we consider only the well-typed syntax of a 
variant of \lambdaAEff~in which the subtyping rule manifests as an explicit coercion.
Working with such well-typed syntax is a standard approach for making
it easier to manage a de Bruijn indices-based representation of free and bound 
variables~\cite{Wadler:PLFA}. Meanwhile, the \pl{{\AE}ff} implementation provides an interpreter 
and a simple typechecker, but does not 
yet support inferring or checking effect annotations. \pl{{\AE}ff} also provides   
a web interface that allows users to interactively click through their programs' executions.
It also comes with implementations of all the examples we present in this paper.
Separately, Poulson~\cite{Poulson:AsyncEffectHandling} has shown how to implement \lambdaAEff~
in \pl{Frank} \cite{Convent:DooBeeDooBeeDoo}.


\section{Asynchronous Effects, by Example}
\label{sec:overview}

We begin with a high-level overview of how we model asynchrony within algebraic effects.

\subsection{Conventional Algebraic Effects Are Synchronous by Nature}
\label{sec:conventional-algebraic-effects}

We first recall the basic ideas of programming with algebraic effects, 
illustrating that their traditional treatment is synchronous by nature.
For an in-depth overview, we refer the reader to a tutorial on effect handlers~\cite{Pretnar:Tutorial}, and to the 
seminal papers of the field \cite{Plotkin:NotionsOfComputation,Plotkin:HandlingEffects}.

In this algebraic treatment, sources of computational effects are modelled using signatures 
of \emph{operation symbols} $\op : A_\op \to B_\op$. For instance, one models 
$S$-valued state using operations $\sigget : \tyunit \to S$ and $\sigset : S \to \tyunit$, 
and $E$-valued exceptions using a single operation $\opsym{raise} : E \to \tyempty$.

Programmers can then invoke the effect that an operation 
$\op : A_\op \to B_\op$ models by placing an \emph{operation call} $\tmop {op} V y M$ in their code. Here, 
the parameter value $V$ has type $A_\op$, and the variable $y$, which is bound in the continuation $M$, 
has type $B_\op$. For instance, for the $\sigset$ operation, the parameter value $V$ would be 
the new value of the store, and for the $\sigget$ operation, the variable $y$ 
would be bound to the current value of the store.

A program written in terms of operation calls is by itself just an inert piece of code. To 
execute it, programmers have to provide \emph{implementations} for the operation 
calls appearing in it. The idea is that an implementation of $\tmop {op} V y M$ takes $V$ as its input, 
and its output gets bound to $y$.
For instance, this could take the form of defining a suitable effect handler 
\cite{Plotkin:HandlingEffects}, but could also be given by calls  
to runners of algebraic effects \cite{Ahman:Runners}, or simply by invoking some  
(default) top-level (native) implementation.
What is important is that some pre-defined piece of code $M_\op[V/x]$
gets executed in place of every operation call $\tmop {op} V y M$.

Now, what makes the conventional treatment of algebraic effects \emph{synchronous} is 
that the execution of an operation call $\tmop {op} V y M$ \emph{blocks} until some implementation 
of $\op$ returns a value $W$ to be bound to $y$, so that 
the execution of the continuation $M[W/y]$ could proceed \cite{Kammar:Handlers,Bauer:EffectSystem}. 
Conceptually, this kind of blocking behaviour can be illustrated as
\begin{equation}
\begin{gathered}
\label{eq:syncopcall}
\xymatrix@C=1.25em@R=0.8em@M=0.5em{
& M_\op[V/x] \ar@{}[r]|{\mbox{\Large{$\leadsto^{\!*}$}}} & \tmreturn W \ar[d]
\\
\cdots \ar@{}[r]|>>>{\mbox{\Large{$\leadsto$}}} & \tmop {op} V y M \ar[u] & M[W/y] \ar@{}[r]|<<<{\mbox{\Large{$\leadsto$}}} & \cdots
}
\end{gathered}
\end{equation}
where $\tmreturn W$ is a computation that causes no effects and simply returns the value $W$.

While blocking the execution of the rest of the computation is needed in the presence of 
general effect handlers that can execute their continuation  any number 
of times, e.g., when simulating nondeterminism~\cite{Plotkin:HandlingEffects}, 
it forces all uses of algebraic effects to be synchronous, even 
when this is not necessary, e.g., when the effect in question involves 
executing a remote query to which a response is not needed immediately, 
or sometimes never at all.

In the rest of this section, we describe how we decouple the invocation of 
an operation call from the act of receiving its result, and how we give 
programmers a means to block execution only when it is necessary. 
While we end up surrendering some of effect handlers' generality, 
such as having access to the continuation that captures the rest of the 
computation to be handled, then in return we get a natural and robust formalism for 
asynchronous programming.

\subsection{Outgoing Signals and Incoming Interrupts}
\label{sec:overview:signals}

We begin by observing that the execution of an operation call $\tmop {op} V y M$, 
as depicted in (\ref{eq:syncopcall}), consists of \emph{three distinct phases}: (i) signalling that an 
implementation of $\op$ needs to be executed with parameter $V$ (the up-arrow), (ii) executing 
this implementation (the horizontal arrow), and (iii) interrupting the blocked computation $M$ with a value $W$ 
(the down-arrow). In order to overcome the unwanted side-effects of blocking execution at every operation call, 
we decouple these phases into separate programming concepts, allowing  
$M$ to proceed executing even if (ii) has not yet completed and (iii) taken place. In particular, we 
decouple an operation call into issuing an \emph{outgoing signal}, 
written $\tmopout{\op}{V}{M}$, and receiving an \emph{incoming interrupt}, written $\tmopin{\op}{W}{M}$.

It is important to note that while we have used the execution of operation calls  
to motivate the introduction of signals and interrupts as programming concepts, \emph{not all issued signals need to have a corresponding 
interrupt response}, and \emph{not all interrupts need to be responses to issued signals}, 
allowing us to also model spontaneous behaviour, such as a user clicking a button or the environment pre-empting a thread.

When \emph{issuing a signal} $\tmopout{\op}{V}{M}$, the value $V$ is called a \emph{payload}, such as a location to be looked up or a 
message to be displayed, aimed at whoever is listening for the given signal. We use the $\tmkw{\uparrow}$-notation to indicate that signals issued in sub-computations propagate outwards---in this sense signals behave just like conventional algebraic 
operation calls.

Since no additional variables are bound in the continuation $M$, it is naturally possible to
continue executing $M$ straight after the signal has been issued, as depicted below:
\vspace{-2ex}
\[
\xymatrix@C=1.25em@R=1.25em@M=0.5em{
& &
\\
\cdots \ar@{}[r]|<<<{\mbox{\Large{$\leadsto$}}} & \tmopout {op} V M \ar[u]^{\op\, V} \ar@{}[r]|<<<{\mbox{\Large{$\leadsto$}}} & M \ar@{}[r]|<<<{\mbox{\Large{$\leadsto$}}} & \cdots
}
\]
This crucially differs from the usual treatment of algebraic effects, which though being able to simulate our approach~\cite{Poulson:AsyncEffectHandling}, find asynchronous
evaluation of continuations undesirable. For example, even if in the (conventional) operation call
$\tmop {op} V y M$ the continuation $M$ does not depend on $y$, $M$ can cause further
effects, leading to unexpected behaviour if $M$ performs those effects before or after the handler
for $\op$ is evaluated.

\newcommand{\client}{M_{\text{feedClient}}}

As a \emph{running example}, let us consider a computation $\client$, which lets a user scroll 
through a seemingly infinite feed of data, e.g., by repeatedly clicking a ``next page'' button.
For efficiency, $\client$ does not initially cache all the data available on a server, but instead requests a 
new batch of data each time scrolling through the data is nearing the end of the cache. To communicate with 
the outside world, $\client$ can issue a signal
\[
  \tmopout{\opsym{request}}{\mathit{offset}}{\client}
\]
to request a new batch of data starting from the given $\mathit{offset}$, or a different signal
\[
  \tmopout{\opsym{display}}{\mathit{message}}{\client}
\]
to display a string $\mathit{message}$ to the user. In both cases, the continuation \emph{does not wait} 
for an acknowledgement that the signal was received, 
but instead continues to provide a seamless experience to the user.
It is however worth noting that these signals differ in what $\client$ expects of them: 
to the $\opsym{request}$ signal, it expects a response at some future point in 
its execution, while it does not expect any response to the $\opsym{display}$ signal, 
illustrating that not every issued signal needs an immediate response, and that 
some do not need one at all.

When the outside world wants to get the attention of a computation, be it in response to 
a signal or spontaneously, 
it happens by \emph{propagating an interrupt}~$\tmopin{\op}{W}{M}$ to the computation. 
Here, the value $W$ is again called a \emph{payload}, while $M$ is the computation receiving the interrupt.
It is important to note that unlike signals, interrupts are not triggered by the computation itself, 
but are instead issued by the \emph{outside world},
and can thus interrupt any sequence of evaluation steps,
e.g., as depicted in
\vspace{-2ex}
\[
\xymatrix@C=1.25em@R=1.25em@M=0.5em{
&  \ar[d]^-{\op\, W}  &
\\
\cdots \ar@{}[r]|<<<{\mbox{\Large{$\leadsto$}}} & M \ar@{}[r]|<<<{\mbox{\Large{$\leadsto$}}} & \tmopin {op} W M \ar@{}[r]|<<<{\mbox{\Large{$\leadsto$}}} & \cdots
}
\]
\noindent In our running example, there are two interrupts of interest that $\client$ might receive:
\[
\tmopin{\opsym{response}}{\mathit{newBatch}}{M}
\]
which delivers a batch of new data to replenish $\client$'s cache, and 
\[
\tmopin{\opsym{nextItem}}{\tmunit}{M}
\]
with which the user requests to see the next data item. In both cases, the continuation 
$M$ represents the state of $\client$ at the time of receiving the interrupt.

We use 
the $\tmkw{\downarrow}$-notation to indicate that interrupts propagate inwards into subcomputations, 
trying to reach anyone listening for them, and only get discarded when they reach a $\tmkw{return}$. 
Programmers are not expected to write interrupts explicitly in their programs---instead, 
interrupts are usually induced by signals issued by other parallel processes, as explained next. 

\subsection{A Signal for the Sender Is an Interrupt to the Receiver}
\label{sec:overview:processes}

As noted above, the computations we consider do not evolve in isolation, instead they also communicate with 
the outside world, by issuing outgoing signals and receiving incoming interrupts.

We model the outside world by composing individual computations into \emph{parallel processes} $P, Q, \ldots$.
To keep the presentation clean and focussed on the asynchrony of algebraic effects, we consider a very 
simple model of parallelism: a process is either one of the computations being run
in parallel, written $\tmrun M$, or the parallel composition of two processes, 
written $\tmpar P Q$. Later, in \autoref{sec:extensions:dynamic-process-creation}, we 
show how to also accommodate dynamic process creation.

\newcommand{\server}{M_{\text{feedServer}}}

To capture the signals and interrupts based interaction of processes,  
our operational semantics includes rules for \emph{propagating outgoing signals} from individual 
computations to processes, 
\emph{turning processes' outgoing signals into incoming interrupts} for their surrounding world, and
\emph{propagating incoming interrupts} from processes to individual computations.
For instance, in our running example, 
$\client$'s request for new data is executed as follows:
\[
\begin{array}{r l}
  & \tmpar{\highlightgray{\tmrun (\tmopout{request}{V}{\highlightwhite{\client}})}}{\tmrun \server} 
  \\[0.5ex]
  \reduces & \highlightgray{\tmpar{(\tmopout{request}{V}{\highlightwhite{\tmrun \client}})}{\highlightwhite{\tmrun \server}}}
  \\[0.5ex]
  \reduces & \tmopoutbig{request}{V}{\tmpar{\tmrun \client}{\highlightgray{\tmopin{request}{V}{\tmrun {\highlightwhite{\server}}}}}}
  \\[0.5ex]
  \reduces & \tmopoutbig{request}{V}{\tmpar{\tmrun \client}{\tmrun (\tmopin{request}{V}{\server})}}
\end{array}
\]
Here, the first and the last reduction step respectively propagate signals outwards and 
interrupts inwards. The middle reduction step corresponds to what we call a \emph{broadcast rule}---it 
turns an outward moving signal in one of the processes into an inward moving interrupt for the process 
parallel to it, while continuing to propagate the signal outwards to any further parallel processes.  
The active redexes in these rules are highlighted in grey.

\subsection{Promising To Handle Interrupts}
\label{sect:overview:promising}

So far, we have shown that our computations can issue outgoing signals and receive incoming interrupts, and how 
these evolve and get communicated when executing parallel processes, but we have not yet said 
anything about how computations can actually \emph{react} to incoming interrupts of interest. 

 In order to react to interrupts, our computations can install \emph{interrupt handlers}, written
\[
  \tmwith{op}{x}{M}{p}{N}
\]
that should be read as: ``we promise to handle a future incoming interrupt named $\op$ using the computation  
$M$ in the continuation $N$, with $x$ bound to the payload of the interrupt''. Fulfilling this promise consists of 
executing $M$ and binding its result to the promise variable $p$ in $N$ when a corresponding interrupt arrives, as 
captured by the following reduction rule:
\[
  \tmopin{op}{V}{\tmwith{op}{x}{M}{p}{N}} \reduces \tmlet{p}{M[V/x]}{\tmopin{op}{V}{N}}
\]
Interrupts that do not match a given interrupt handler ($\op \neq \op'$) simply move past it:
\[
  \tmopin{op'}{V}{\tmwith{op}{x}{M}{p}{N}} \reduces \tmwith{op}{x}{M}{p}{\tmopin{op'}{V}{N}}
\]
\noindent
It is worth noting that the interrupt itself \emph{keeps propagating inwards} into the sub-computation 
$N$, where it can trigger further interrupt handlers installed for the given interrupt.
Allowing the interrupts to always keep propagating inwards is a natural design choice, as it connects
the behaviour of our interrupts with the behaviour of deep effect handling~\cite{Plotkin:HandlingEffects} (see 
\autoref{sec:basic-calculus:semantics:computations}), and it is crucial for certain examples (see
\autoref{sec:applications:chaining}).

In order to skip certain interrupt handlers for some $\opsym{op}$, one can carry additional data 
in $\opsym{op}$'s payload (e.g., a thread ID) and then condition the (non-)triggering of those interrupt 
handlers on this data, e.g., as we demonstrate in \autoref{sec:applications:guarded-handlers}.
This is analogous to how one controls which particular operation calls are handled with 
ordinary effect handlers~\cite{Kammar:Handlers}.

Interrupt handlers differ from conventional algebraic operation calls 
(see \autoref{sec:conventional-algebraic-effects}) in two important aspects.
First, they enable \emph{user-side post-processing} of received data, using $M$, 
while in operation calls the result is immediately bound in the continuation. Second, and more 
importantly, their semantics is \emph{non-blocking}. In particular, we have a congruence rule
\[
N \reduces N' \qquad \text{implies} \qquad \tmwith{op}{x}{M}{p}{N} \reduces \tmwith{op}{x}{M}{p}{N'}
\]
meaning that the continuation $N$, and thus the whole computation, can make progress 
even though no interrupt $\opsym{op}$ has been propagated to the computation 
from the outside world.

As the observant reader might have noticed, the non-blocking behaviour of interrupt handling 
means that our operational semantics has to work on \emph{open terms} because the variable 
$p$ can appear free in both $N$ and $N'$ in the congruence rule given above. However, it is important 
to note that $p$ is not an arbitrarily typed variable, 
but in fact gets assigned a distinguished \emph{promise type} $\typromise X$ for some value type 
$X$---we shall crucially make use of this typing of $p$ in the proof of type safety for our \lambdaAEff-calculus 
(see \autoref{thm:progress} and \ref{thm:progress:extended}).
Furthermore, since it is the computation $M$ that fulfils the promise (either by supplying a value or returning
another promise), it also needs to have the same return type $\typromise X$.

\subsection{Blocking on Interrupts Only When Necessary}
\label{sec:overview:await}

As noted earlier, installing an interrupt handler means making a promise to handle a given 
interrupt in the future. To check that an interrupt has been received and handled, 
we provide programmers a means to selectively \emph{block execution}
and \emph{await} a specific promise to be fulfilled, written 
$\tmawait{V}{x}{M}$, where if $V$ has a promise type $\typromise X$, the variable $x$ bound in $M$ has type $X$.
Importantly, the continuation $M$ is executed only 
when the $\tmkw{await}$ is handed a \emph{fulfilled promise} $\tmpromise V$, as
\[
\tmawait{\tmpromise V}{x}{M} \reduces M[V/x]
\]
\noindent
In our example of scrolling through a seemingly infinite feed,
$\client$ could use $\tmkw{await}$ to block until it has received an initial configuration, 
such as the batch size used by $\server$.

As the terminology suggests, this part of \lambdaAEff~is strongly influenced by existing work on 
\emph{futures and promises} \cite{Schwinghammer:Thesis} for structuring concurrent programs, and their use in modern languages, 
such as in \pl{Scala} \cite{Haller:Futures}. While prior work often models promises as writeable, single-assignment 
references, we instead use the substitution of values for ordinary immutable variables (of distinguished promise type) 
to model that a promise gets fulfilled exactly once. 
This way we achieve the standard reading of promises without needing
a stateful operational semantics and a non-trivial type system to
enforce the single-assignment behaviour~\cite{Ahman:RecallingWitness}.

\subsection{Reinstalling Interrupt Handlers}
\label{sec:overview:reinstallable-interrupt-handlers}

As seen in the reduction rule
\[
  \tmopin{op}{V}{\tmwith{op}{x}{M}{p}{N}} \reduces \tmlet{p}{M[V/x]}{\tmopin{op}{V}{N}}
\]
the interrupt handler is \emph{not reinstalled by default}. The programmers can selectively reinstall
interrupt handlers using general recursion~\cite{Ahman:POPL}, or use the extension of \lambdaAEff~with
\emph{reinstallable interrupt handlers} we propose in this paper (see 
\autoref{sec:extensions:reinstallable-interrupt-handlers} for details), which have the form
\[
  \tmwithre{op}{x}{r}{M}{p}{N}
\]
These behave similarly to ordinary interrupt handlers, except that the handling computation $M$
has access to an additional variable $r$ bound to a function that reinstalls the handler when called.
Specifically, triggering a reinstallable interrupt handler has the following form:
\begin{multline*}
  \tmopin{op}{V}{\tmwithre{op}{x}{r}{M}{p}{N}} \\
  \reduces \tmlet{p}{M\big[V/x, \big(\tmfunano{\_}{\tmwithre{op}{x}{r}{M}{p}{\tmreturn{p}}}\big)/r\big]}{\tmopin{op}{V}{N}}
\end{multline*}
\noindent
Further, in examples we often find it useful to also pass data between 
subsequent reinstalls of an interrupt handler. Programmers can achieve this by working 
with an additionally assumed primitive notion of memory references~\cite{Ahman:POPL}, 
or by using a \emph{stateful variant of reinstallable interrupt handlers} that we propose in 
this paper. The latter have the form
\[
  \tmwithrest[S]{op}{x}{r}{s}{M}{W}{p}{N}
\]
where $S$ is the type of state associated with a particular interrupt handler, 
$W$ is the interrupt handler's state at the time of its next triggering, 
the variable $s$ gives the interrupt handler code $M$ access to the state, 
and the state can be updated by reinstalling the handler with an updated value
using $r$. Specifically, the interrupt handler triggering rule now has the form
\[
\tmopin{op}{V}{\tmwithrest[S\!]{op}{x}{r}{s}{M}{W}{p}{N}} \reduces \tmlet{p}{M\big[V/x , R/r , W/s\big]}{\tmopin{op}{V}{N}}
\]
where $R$ denotes a function that reinstalls the interrupt handler
with an updated state value:
\[
R ~\defeq~ \tmfun{s' \of S}{\tmwithrest[S]{op}{x}{r}{s}{M}{s'}{p}{\tmreturn p}}
\]
For brevity, we often omit the $S$-annotation in examples when it is clear from
the context.

\subsection{Putting It All Together}
\label{sec:overview:runningexample}

We conclude this overview by showing how to implement the example of a user scrolling through a seemingly 
infinite feed of data in our \lambdaAEff-calculus.

For a simpler exposition, we allow ourselves access to mutable references, with which 
we communicate data between different interrupt handlers, though the same can be  
achieved by rolling one's own state. For passing data between subsequent reinstalls of the same 
interrupt handler, we use the state-passing features of interrupt handlers 
introduced above.

While having explicit continuations in operation calls, signals, interrupt handlers, and when awaiting promises
to be fulfilled makes the meta-theory of the underlying calculus cleaner (see Section~\ref{sec:basic-calculus:semantics:computations}), in programming we prefer to use 
\emph{generic} versions of them, i.e., ones with trivial continuations~\cite{Plotkin:GenericEffects}.
In particular, we define and use the syntactic sugar:
\begin{align*}
  \tmopoutgen {op} V &~\defeq~ \tmopout {op} V {\tmreturn \tmunit} \\
  \tmwithregen {op}{x}{r}{M} &~\defeq~ \tmwithre{op}{x}{r}{M}{p}{\tmreturn{p}} \\
  \tmwithregenst{op}{x}{r}{s}{M}{W} &~\defeq~ \tmwithrest{op}{x}{r}{s}{M}{W}{p}{\tmreturn p} \\
  \tmawaitgen{V} &~\defeq~ \tmawait{V}{x}{\tmreturn x}
\end{align*}

\subsubsection{Client}
\label{sec:overview:runningexample:client}

We implement the client computation $\client$ as the function \ls$feedClient$ defined below. 
For presentation purposes, we split its definition between multiple code blocks.

First, the client initialises some auxiliary references, 
issues a signal to the server to ask for the data batch size that it uses, and then installs a corresponding
interrupt handler:
\begin{aeffbox}
\begin{lstlisting}
let feedClient () =
    let cachedData = ref [] in
    let requestInProgress = ref false in
    send batchSizeRequest ();
    let batchSizePromise = promise (batchSizeResponse batchSize |-> return <<batchSize>>) in
    ...
\end{lstlisting}
\end{aeffbox}
While the server is asynchronously responding to the batch size request, the client
sets up an auxiliary function \ls$requestNewData$, with which it can request new data from the server:
\begin{aeffbox}
\begin{lstlisting}
    ...
    let requestNewData offset =
        requestInProgress := true;
        send request offset;
        promise (response newBatch |->
            cachedData := !cachedData @ newBatch;
            requestInProgress := false;
            return <<()>>
        )
    in
    ...
\end{lstlisting}
\end{aeffbox}
Here, we first set a flag indicating that a new data request is in process, 
then issue a \ls$request$ signal to the server, and finally install an 
interrupt handler that updates the cache 
once a corresponding \ls$response$ interrupt arrives.
We note that the client computation does not block while awaiting new data
from the server---instead, it continues executing, notifying the user to wait and try again once 
the cache temporarily becomes empty (see below).

As a last step of setting itself up, the client blocks until the server has responded 
with the batch size it uses by awaiting \ls$batchSizePromise$ to be fulfilled, 
after which the client starts its main loop, which we implement as a simple reinstallable 
interrupt handler:
\begin{aeffbox}
\begin{lstlisting}
    ...
    let batchSize = await batchSizePromise in
    promise (nextItem _ r currentItem |->
        let cachedSize = length !cachedData in
        (if (currentItem > cachedSize - batchSize / 2) && (not !requestInProgress) then
             requestNewData (cachedSize + 1)
         else
             return ());
        if currentItem < cachedSize then
            send display (toString (nth !cachedData currentItem));
            r (currentItem + 1)
        else  
            send display "please wait a bit and try again";
            r currentItem
    ) @ 0
\end{lstlisting}
\end{aeffbox}
In it, the client listens for \ls$nextItem$ interrupts from the user to display more data.
Once the interrupt arrives, the client checks if its cache is becoming empty, i.e., if
the index of the currently viewed item is less than half of the batch size away from the
last cached item and if no request for new data has been issued yet. If that happens, 
the client uses the \ls$requestNewData$ function to request more data from the server,
starting with offset \ls$cachedSize + 1$, which is the index of the first item that is
outside of the data cached by the client.

Next, if there is still some data in the cache, the client issues a \ls$display$ signal to show
the next data item to the user. If however the cache is empty, the client issues a \ls$display$ signal 
to show a message to the user asking them to wait and try again. The client then simply reinvokes itself
by reinstalling the interrupt handler for \ls$nextItem$ interrupts (by calling \ls$r$). 

Observe that the 
\ls$currentItem$ counter is initially set to $0$ and then passed between subsequent 
interrupt handler reinstalls using the state-passing features introduced earlier.

\subsubsection{Server}
\label{sec:overview:runningexample:server}

We implement the server computation $\server$ as the following function:
\begin{aeffbox}
\begin{lstlisting}
let feedServer batchSize =
    promise (batchSizeRequest () r |->
        send batchSizeResponse batchSize;
        r ()
    );
    promise (request offset r |->
        let payload = map (fun x |-> 10 * x) (range offset (offset + batchSize - 1)) in
        send response payload;
        r ()
    )
\end{lstlisting}
\end{aeffbox}
where the computation \lstinline{range i j} returns a list of integers ranging from \lstinline{i} to \lstinline{j} (both inclusive).

The server simply installs two reinstallable interrupt handlers: the first 
one listens for and responds to client's requests about the batch size it uses; 
and the second one responds to client's requests for new data. Both interrupt 
handlers then simply reinstall themselves.

The two interrupt handlers share a common pattern of handling the interrupt by
issuing a signal and then immediately reinstalling the handler, and it is tempting to avoid the
repetition. A dual shared pattern can be found in
\autoref{sec:overview:runningexample:client}, where issuing a request signal is
immediately followed by installing an interrupt handler for its response. However, proper user-defined
abstractions capturing these patterns would require operation names to be first-class values, which is not only orthogonal to the issue of asynchrony we are focussing on, but leads to a dependently typed calculus in combination with an effect system.

\subsubsection{User}
\label{sec:overview:runningexample:user}

We can also simulate the user as a computation. For the sake of simplicity,
we allow ourselves general recursion to implement the user behaviour as an infinite loop
that every now and then issues a request to the client to display the next data item.
\begin{aeffbox}
\begin{lstlisting}
let rec user () =
    let rec wait n = 
        if n = 0 then return () else wait (n - 1)
    in
    send nextItem (); wait 10; user ()
\end{lstlisting}
\end{aeffbox}
Alternatively, without assuming general recursion, we could have implemented the user 
instead as two parallel processes that indefinitely ping-pong each other, and occasionally 
issue \ls$nextItem$ signals to the client (see~\autoref{sec:conclusion} for an 
example of such non-terminating behaviour).
It is also straightforward to extend the user program with a reinstallable handler for 
\ls$display$ interrupts that simulates displaying the data items received from the 
client (omitted here).

\subsubsection{Running the Server, Client, and User in Parallel}
\label{sec:overview:runningexample:parallel}

Finally, we can simulate our running example in full by running all 
three computations as parallel processes, as follows: 
\begin{aeffbox}
\begin{lstlisting}
run (feedServer 42) || run (feedClient ()) || run (user ())
\end{lstlisting}
\end{aeffbox}


\section{A Calculus for Asynchronous Effects: Values and Computations}
\label{sec:basic-calculus:computations}

Before we  focus on extensions necessary for higher-order asynchronous effects in \autoref{sec:higher-order-extensions},
we first recap \lambdaAEff, our existing core calculus for programming with first-order asynchronous effects~\cite{Ahman:POPL}.
The version we present here differs from the original one in two aspects: we drop the reliance on general recursion, as
reinstallable interrupt handlers that we introduce in \autoref{sec:extensions:reinstallable-interrupt-handlers} are sufficient to express
all the existing examples, and we slightly modify the behaviour of
the $\tmkw{await}$ construct in order to make the meta-theory slightly simpler.

To better explain the different features of the calculus and its semantics, we split the recap of
\lambdaAEff~into a \emph{sequential} part (discussed below) and a
\emph{parallel} part (discussed in \autoref{sec:basic-calculus:processes}).

\subsection{Values and Computations}
\label{sec:basic-calculus:values-and-computations}

We base \lambdaAEff~on the fine-grain
call-by-value $\lambda$-calculus (FGCBV)~\cite{Levy:FGCBV}, and as such, it is a low-level
intermediate language to which a corresponding high-level user-facing programming language
could be compiled to---this is what happens in our prototype implementation~\cite{pretnar21:AEff}.

The syntax of terms is given in \autoref{fig:terms}, stratified into \emph{values} and \emph{computations},
as in FGCBV.
While we do not study effect inference in this paper, we equip certain terms with type annotations that
in our experience should make it possible to fully infer types.

\begin{figure}[tp]
  \parbox{\textwidth}{
  \centering
  \small
  \begin{align*}
  \intertext{\textbf{Values}}
  V, W
  \bnfis& x                                       & &\text{variable} \\
  \bnfor& \tmunit \bnfor\! \tmpair{V}{W}                                & &\text{unit and pairing} \\
  \bnfor& \tminl[Y]{V} \bnfor\! \tminr[X]{V}    & &\text{left and right injections} \\
  \bnfor& \tmfun{x : X}{M}                        & &\text{function abstraction} \\
  \bnfor& \tmpromise V                            & &\text{fulfilled promise}
  \\[1ex]
  \intertext{\textbf{Computations}}
  M, N
  \bnfis& \tmreturn{V}                            & &\text{returning a value} \\
  \bnfor& \tmlet{x}{M}{N}          & &\text{sequencing} \\
  \bnfor& V\,W                                    & &\text{function application} \\
  \bnfor& \tmmatch{V}{\tmpair{x}{y} \mapsto M}    & &\text{product elimination} \\
  \bnfor& \tmmatch[\tycomp{Z}{(\o,\i)}]{V}{}                        & &\text{empty elimination} \\
  \bnfor& \tmmatch{V}{\tminl{x} \mapsto M, \tminr{y} \mapsto N}
                                                  & &\text{sum elimination} \\
  \bnfor& \tmopout{op}{V}{M}       & &\text{outgoing signal} \\
  \bnfor& \tmopin{op}{V}{M}          & &\text{incoming interrupt} \\
  \bnfor& \tmwith{op}{x}{M}{p}{N}      & &\text{interrupt handler} \\
  \bnfor& \tmawait{V}{x}{M}             & &\text{awaiting a promise to be fulfilled}
  \end{align*}
  }
  \caption{Values and Computations.}
  \label{fig:terms}
\end{figure}

\paragraph{Values}

The values $V,W,\ldots$ are mostly standard. They include
variables, introduction forms for
sums and products, and functions. The only \lambdaAEff-specific value
is $\tmpromise V$, which denotes a \emph{fulfilled promise}, indicating that the promise of
handling some interrupt has been fulfilled with the value $V$.

\paragraph{Computations} The computations $M,N,\ldots$ also include all
standard terms from FGCBV:
returning values, sequencing, function
application, and elimination forms.

The first two \lambdaAEff-specific computations are \emph{signals} $\tmopout{op}{V}{M}$ and
\emph{interrupts} $\tmopin{op}{V}{M}$, where
$\opsym{op}$ is drawn from a set $\sig$ of names, $V$ is a data
payload, and $M$ is a continuation.

The next \lambdaAEff-specific computation is the \emph{interrupt handler} $\tmwith{op}{x}{M}{p}{N}$,
where $x$ is bound in $M$ and $p$ in $N$.
As discussed in the previous section, one should understand this computation as making a promise
to handle a future incoming interrupt $\opsym{op}$ by executing the computation $M$. Sub-computations of the continuation
$N$ can then explicitly await, when necessary, for this promise to be fulfilled by blocking on the \emph{promise-typed variable} $p$
using the final \lambdaAEff-specific computation term, the \emph{awaiting} construct $\tmawait{V}{x}{M}$.
It is useful to note that the $p$ used above is an ordinary variable---it just gets assigned the distinguished promise type
$\typromise X$ by the interrupt handler (as discussed in \autoref{sec:basic-calculus:type-system:computations}).

\subsection{Small-Step Operational Semantics}
\label{sec:basic-calculus:semantics:computations}

We equip \lambdaAEff~with an evaluation contexts based
small-step operational semantics,
defined using a reduction relation $M \reduces N$.
The \emph{reduction rules} and \emph{evaluation contexts} are given
in \autoref{fig:small-step-semantics-of-computations}. We discuss
the rules in detail below. Note that since we have chosen to equip effectful constructs with explicit continuations,
the evaluation contexts are used to compress four congruence rules into a single one. If instead
we took generic versions (like seen in Section~\ref{sec:overview:runningexample})
as primitives, almost all the rules in \autoref{fig:small-step-semantics-of-computations},
apart from the ones for standard monadic computations, would need to be phrased
in terms of sequential composition (i.e., $\tmkw{let}$), leading to a notably less clear presentation.
\begin{figure}[tp]
  \small
  \begin{align*}
    \intertext{\textbf{Standard computation rules}}
    \tmapp{(\tmfun{x \of X}{M})}{V} &\reduces M[V/x]
    \\
    \tmlet{x}{(\tmreturn V)}{N} &\reduces N[V/x]
    \\
    \tmmatch{\tmpair{V}{W}}{\tmpair{x}{y} \mapsto M} &\reduces M[V/x, W/y]
    \\
    \mathllap{
      \tmmatch{(\tminl[Y]{V})}{\tminl{x} \mapsto M, \tminr{y} \mapsto N}
    } &\reduces
    M[V/x]
    \\
    \mathllap{
      \tmmatch{(\tminr[X]{W})}{\tminl{x} \mapsto M, \tminr{y} \mapsto N}
    } &\reduces
    N[W/y]
    \\[1ex]
    \intertext{\textbf{Algebraicity of signals, interrupt handlers, and awaiting}}
    \tmlet{x}{(\tmopout{op}{V}{M})}{N} &\reduces \tmopout{op}{V}{\tmlet{x}{M}{N}}
    \\
    \tmlet{x}{(\tmwith{op}{y}{M}{p}{N_1})}{N_2} &\reduces \tmwith{op}{y}{M}{p}{(\tmlet{x}{N_1}{N_2})}
    \\
    \tmlet{x}{(\tmawait{V}{y}{M})}{N} &\reduces \tmawait{V}{y}{(\tmlet{x}{M}{N})}
    \\[1ex]
    \intertext{\textbf{Commutativity of signals with interrupt handlers}}
    \tmwith{op}{x}{M}{p}{\tmopout{op'}{V}{N}} &\reduces \tmopout{op'}{V}{\tmwith{op}{x}{M}{p}{N}}
    \\[1ex]
    \intertext{\textbf{Interrupt propagation}}
    \tmopin{op}{V}{\tmreturn W} &\reduces \tmreturn W
    \\
    \tmopin{op}{V}{\tmopout{op'}{W}{M}} &\reduces \tmopout{op'}{W}{\tmopin{op}{V}{M}}
    \\
    \tmopin{op}{V}{\tmwith{op}{x}{M}{p}{N}} &\reduces \tmlet{p}{M[V/x]}{\tmopin{op}{V}{N}}
    \\
    \tmopin{op'}{V}{\tmwith{op}{x}{M}{p}{N}} &\reduces \tmwith{op}{x}{M}{p}{\tmopin{op'}{V}{N}}
    \;\;\; {\color{rulenameColor}(\op \neq \op')}
    \\
    \tmopin{op}{V}{\tmawait{W}{x}{M}} &\reduces \tmawait{W}{x}{\tmopin{op}{V}{M}}
    \\[1ex]
    \intertext{\textbf{Awaiting a promise to be fulfilled}}
    \tmawait{\tmpromise V}{x}{M} &\reduces M[V/x]
  \end{align*}
  \begin{gather*}
    \intertext{\textbf{Evaluation context rule}}
    \coopinfer{}{
      M \reduces N
    }{
      \E[M] \reduces \E[N]
    }
    \intertext{\textbf{where}}
    \text{$\E$}
    \bnfis [~]
    \bnfor \tmlet{x}{\E}{N}
    \bnfor \tmopout{op}{V}{\E}
    \bnfor \tmopin{op}{V}{\E}
    \bnfor \tmwith{op}{x}{M}{p}{\E}
  \end{gather*}
  \caption{Small-step Operational Semantics of Computations.}
  \label{fig:small-step-semantics-of-computations}
\end{figure}

\paragraph{Computation Rules}
The first group includes \emph{standard reduction rules} from FGCBV, such as $\beta$-reducing function applications, sequential composition, and the standard elimination forms.
These rules involve standard \emph{capture avoiding substitutions} $V[W/x]$ and $M[W/x]$,
defined by straightforward mutual structural recursion on $V$ and $M$.

\paragraph{Algebraicity}
This group of reduction rules \emph{propagates outwards} the signals that have been issued, interrupt handlers that have been installed,
and computations awaiting fulfilled promises. While it is not surprising that outgoing signals
behave like algebraic \emph{operation calls}, getting propagated outwards as far as possible, then it is much more curious that
the natural operational behaviour of interrupt handlers turns out to be the same. As we shall explain in \autoref{sec:conclusion},
despite using the (operating systems inspired) ``handler'' terminology, mathematically interrupt handlers are in fact a form of scoped algebraic operations~\cite{Pirog:ScopedOperations}.

In contrast to our original calculus~\cite{Ahman:POPL}, the awaiting construct also propagates outwards. Before,
awaiting a promise in any subcomputation would block the evaluation immediately, whereas now, we can do the
additional outwards propagation steps. Importantly, this does not significantly change the computational behaviour, as after
the propagation, the evaluation still blocks as long as the promise is left unfulfilled. The main difference and benefit is that
all computations awaiting for a promise variable $p$ now show this explicitly at their top-level, as they are of the form
$\tmawait{p}{x}{M}$. This change significantly simplifies
the normal forms of computations (see \autoref{sec:basic-calculus:type-safety}) and the resulting meta-theory.

In the last two algebraicity rules, and other similar ones, we assume Barendregt's variable
convention to avoid accidentally capturing free variables when extending the scope of 
binders.

\paragraph{Commutativity of Signals With Interrupt Handlers}
This rule complements the algebraicity rule for signals, by further propagating
them outwards, past enveloping interrupt handlers. From the perspective of algebraic effects,
this rule is an example of two algebraic operations \emph{commuting}~\cite{Hyland:SumAndTensor}.
Since in this rule, the scope of $p$ contracts, the usual variable naming precautions are
not sufficient for type safety. Instead, the type system ensures (see Section~\ref{sec:basic-calculus:type-system:computations}) that the (promise-typed)
variable $p$ cannot appear in the payload value $V$.

\paragraph{Interrupt Propagation}
The handler-operation curiosity does not end with interrupt handlers. This group of reduction rules describes how
interrupts are \emph{propagated inwards} into sub-computations. While $\tmopin{op}{V}{M}$ might look like a conventional
operation call, then its operational behaviour instead mirrors that of \emph{deep effect handling}~\cite{Plotkin:HandlingEffects},
where one also recursively descends into the computation being handled.

When designing interrupt propagation, we must ensure that each interrupt handler receives a corresponding interrupt, no matter
how deep inside the computation we install it. The first reduction rule states that
we can safely discard an interrupt when it reaches a trivial, effect-free
computation $\tmreturn W$. The second rule states that we can propagate incoming interrupts past any outward moving signals. The next
two rules describe how interrupts interact with interrupt handlers, in particular, that the former behave like effect handling
(when understanding interrupt handlers as generalised algebraic operations). On the one hand, if the interrupt
matches the interrupt handler it encounters, the corresponding handler code $M$ is executed, and the interrupt is
propagated inwards into the continuation $N$. On the other hand, if the interrupt
does not match the interrupt handler, it is simply propagated past the interrupt handler into $N$.
Finally, to simplify normal forms, we propagate interrupts inside computations awaiting fulfilled promises as well. As with the algebraicity rule,
this lets the computation take a single additional step after which the $\tmkw{await}$ construct reaches the top and blocks
the evaluation.

We have given the interrupt propagation rules only for terms that are in normal form (see \autoref{lem:results-are-final}). For example, we do not push interrupts into the branches of sum elimination.
Instead, for terms that are still reducing, interrupts remain as parts of their evaluation contexts
and wait for inner interrupt handlers to propagate outwards and meet them.

An alternative design choice for interrupt propagation would be to take inspiration from
\emph{shallow interrupt handling}~\cite{Kammar:Handlers}, and instead of always propagating
the interrupts inwards into the continuations of interrupt handlers, the programmers themselves
would have to manually (recursively) reinvoke the interrupts that need to be propagated
inwards. In addition to giving an algebraically more natural semantics (due to the relationship
with deep effect handling), our choice of allowing interrupts to always propagate inwards provides a
more predictable programming model, in which an installed interrupt handler is guaranteed to
be executed whenever a corresponding interrupt is received, no matter what other installed interrupt 
handlers may do on the way. We leave exploring a variant of \lambdaAEff~based on shallow effect 
handling, and its formal relationship to this paper, for future work.

\paragraph{Awaiting a Promise To Be Fulfilled}
In addition to the two rules for outwards propagation, the semantics of the $\tmkw{await}$ construct
includes a $\beta$-rule allowing the blocked computation $M$ to resume executing as $M[V/x]$
when the $\tmkw{await}$ in question is given a fulfilled promise $\tmpromise V$.

\paragraph{Evaluation Contexts}
The semantics allows reductions under \emph{evaluation contexts} $\E$.
Observe that  as discussed earlier, the inclusion of interrupt handlers in the evaluation contexts means that reductions
involve potentially open terms.
Also, differently from the semantics of conventional operation calls \cite{Kammar:Handlers,Bauer:EffectSystem},
our evaluation contexts include outgoing signals. As such, the \emph{evaluation context rule} allows the execution of a computation
to proceed even if a signal has not yet been propagated to its receiver, or when an interrupt has
not yet arrived. Importantly, the evaluation contexts do not include $\tmkw{await}$, so as to model its blocking behaviour.
We write $\E[M]$ for the operation of filling the hole $[~]$ in $\E$ with $M$.

\paragraph{Non-Confluence}
It is worth noting that the asynchronous design means that the operational semantics
is \emph{nondeterministic}. More interestingly, the semantics is also \emph{not confluent}.

For one source of non-confluence, let us consider two reduction sequences of a same computation,
where for better readability, we highlight the active redex for each step:
 \[
\hspace{-0.15cm}
\begin{array}{r@{\,} l}
  & \tmopin{op}{V}{\tmwith{op}{x}{(\tmwith{op'}{y}{M}{q}{\tmawait{q}{z}{M'}})}{p}{\!\highlightgray{N}}}
  \\[1ex]
  \reduces & \highlightgray{\tmopin{op}{V}{\tmwith{op}{x}{(\tmwith{op'\!}{y}{\!M}{q}{\tmawait{q}{z}{M'}})}{p}{\!\highlightwhite{N'}}}}
  \\[1ex]
  \reduces & \highlightgray{\tmlet{p}{(\tmwith{op'}{y}{M[V/x]}{q}{\tmawait{q}{z}{M'}})}{\!\highlightwhite{\tmopin{op}{V}{N'}}}}
  \\[1ex]
  \reduces & \tmwith{op'}{y}{M[V/x]}{q}{\tmawait{q}{z}{(\tmlet{p}{M'}{\tmopin{op}{V}{N'}})}}
\end{array}
\]
and
\[
\hspace{-0.15cm}
\begin{array}{r@{\,} l}
  & \highlightgray{\tmopin{op}{V}{\tmwith{op}{x}{(\tmwith{op'}{y}{M}{q}{\tmawait{q}{z}{M'}})}{p}{\!\highlightwhite{N}}}}
  \\[1ex]
  \reduces & \highlightgray{\tmlet{p}{(\tmwith{op'}{y}{M[V/x]}{q}{\tmawait{q}{z}{M'}})}{\!\highlightwhite{\tmopin{op}{V}{N}}}}
  \\[1ex]
  \reduces & \tmwith{op'}{y}{M[V/x]}{q}{\tmawait{q}{z}{(\tmlet{p}{M'}{\tmopin{op}{V}{N}})}}
\end{array}
\]
Here, both final computations are \emph{temporarily} blocked until an incoming interrupt $\opsym{op'}$
is propagated to them and the variable $q$ gets bound to a fulfilled promise. Until this happens,
it is not possible for the blocked continuation $N$ to reduce to $N'$ in the latter final computation.

Another, distinct source of non-confluence concerns the commutativity of outgoing signals with enveloping interrupt
handlers. For instance, the following composite computation
\[
\tmopin{op}{V}{{\tmwith {op} x {\tmopout{op'}{W'}{M}} p {\tmopout{op''}{W''}{N}}}}
\]
can nondeterministically reduce to either
\[
\tmopout{op'}{W'}{\tmopout{op''}{W''}{{\tmlet{p}{M}{\tmopin{op}{V}{N}}}}}
\]
if we first propagate the interrupt $\op$ inwards, or to
\[
\tmopout{op''}{W''}{\tmopout{op'}{W'}{{\tmlet{p}{M}{\tmopin{op}{V}{N}}}}}
\]
if we first propagate the signal $\op''$ outwards. As a result, in the resulting two computations,
the signals $\op'$ and $\op''$ get issued, and received by other processes, in a different order.

\paragraph{A More Efficient Operational Semantics?}
Finally, it is worth emphasising that the operational semantics we present in
this paper is meant to serve as a declarative reference semantics of
\lambdaAEff, and as a means to relate the behaviour of the program constructs
specific to \lambdaAEff~to the behaviour of conventional algebraic effects and
their handlers. As such, the semantics is clearly not as efficient as one might
desire in a real-world implementation. For instance, in the current semantics,
signals are propagated out of computations one small step at a time. Instead,
one might consider an alternative semantics in which there would be a reduction
rule to pull signals out of computations from arbitrary depths. Dually, the
propagation of interrupts into computations also happens one small step at a
time. Here one might wonder whether it could be possible to use substitution in
\lambdaAEff~to make that propagation more efficient, akin to how we currently
use substitution to propagate fulfilled promises to sub-computations. Yet
another approach could be to model signal and interrupt propagation using shared
channels, as noted in \autoref{sec:conclusion}. However, as in this paper our
focus is not on the efficiency of the semantics, we leave all such explorations
for future work.

\subsection{Type-and-Effect System}
\label{sec:basic-calculus:type-system:computations}

We equip \lambdaAEff~with a type system in the tradition of type-and-effect systems for algebraic effects and
effect handlers \cite{Bauer:EffectSystem,Kammar:Handlers}, by extending the simple type system of FGCBV
with annotations about programs' possible effects (such as issued signals and installed interrupt handlers)
in function and computation types.

\subsubsection{Types}
\label{sec:basic-calculus:type-system:computations:types}

We define types in \autoref{fig:types}, separated into ground, value, and computation types.

\begin{figure}[t]
  \parbox{\textwidth}{
  \centering
  \small
  \begin{align*}
  \text{Ground type $A$, $B$}
  \bnfis& \tybase \bnfor \tyunit \bnfor \tyempty \bnfor \typrod{A}{B} \bnfor \tysum{A}{B}
  \\[1.5ex]
  \text{Signal or interrupt signature:}
  \phantom{\bnfis}& \op : A_\op
  \\[1.5ex]
  \text{Outgoing signal annotations:}
  \phantom{\bnfis}& \o \in O
  \\
  \text{Interrupt handler annotations:}
  \phantom{\bnfis}& \i \in I
  \\[1.5ex]
  \text{Value type $X$, $Y$}
  \bnfis& A \bnfor \typrod{X}{Y} \bnfor \tysum{X}{Y} \bnfor \tyfun{X}{\tycomp{Y}{(\o,\i)}} \bnfor \typromise{X}
  \\
  \text{Computation type:}
  \phantom{\bnfis}& \tycomp{X}{(\o,\i)}
  \\[1.5ex]
  \text{Typing context $\Gamma$}
  \bnfis& \cdot \bnfor \Gamma, x \of X
  \end{align*}
  }
  \caption{Value and Computation Types.}
  \label{fig:types}
\end{figure}

As noted in \autoref{sec:basic-calculus:values-and-computations}, \lambdaAEff~is parameterised over a set
$\sig$ of signal and interrupt \emph{names}. To each such name $\op \in \sig$, we assign a \emph{signature}
$\op : A_\op$ that specifies the payload type~$A_\op$ of the corresponding signal or interrupt.
Crucially, in order to be able to later prove that \lambdaAEff~is type-safe, we must put restrictions
on these signatures, as they classify values that may cross interrupt handler or process boundaries.
In \autoref{sec:extensions:fitch-style-modal-types}, we describe the exact reasons behind this restriction,
and propose a more flexible type system employing Fitch-style modal types~\cite{Clouston:FitchStyle}. But for the sake of exposition, we use here
the more limited approach from our original work~\cite{Ahman:POPL}, and restrict payload types to
\emph{ground types} $A, B, \ldots$, which include base, unit, empty, product, and sum types, but importantly
exclude promise and function types.

\emph{Value types} $X,Y,\ldots$ extend ground types with function and promise types.
The \emph{function type} $\tyfun{X}{\tycomp{Y}{(\o,\i)}}$ classifies functions that take $X$-typed arguments
to computations classified by the \emph{computation type} $\tycomp{Y}{(\o,\i)}$, i.e., ones that return $Y$-typed
values, while possibly issuing signals specified by $\o$ and handling interrupts specified by $\i$.
The \emph{effect annotations} $\o$ and $\i$ are drawn from sets $O$ and $I$ whose definitions we discuss
in \autoref{sec:basic-calculus:effect-annotations}. The \lambdaAEff-specific \emph{promise type}
$\typromise{X}$ classifies promises that can be fulfilled by supplying a value of type $X$.

\subsubsection{Effect Annotations}
\label{sec:basic-calculus:effect-annotations}

We now explain how we define the sets $O$ and $I$ from which we draw the
effect annotations we use for specifying functions and computations.
Traditionally, effect systems for algebraic effects simply use (flat) sets of
operation names for effect annotations \cite{Bauer:EffectSystem,Kammar:Handlers}.
In \lambdaAEff, however, we need to be
more careful, because triggering an interrupt handler executes a computation
that can issue potentially different signals and handle different interrupts from the main
program, and we would like to capture this in types.

\paragraph{Signal Annotations}
First, as outgoing signals do not carry any computational data, we follow
the tradition of type-and-effect systems for algebraic effects, and define
$O$ to be the \emph{power set} $\Pow \sig$. As such, each $\o \in O$ is a subset of
the signature $\Sigma$, specifying which signals a computation might issue (this is
an over-approximation of the actually issued signals).

\paragraph{Interrupt Handler Annotations}
As observed above, for specifying installed interrupt handlers, we cannot use (flat) sets
of interrupt names as the effect annotations $\i \in I$ if we want to track the nested
(and sometimes recursive) effectful structure of interrupt handlers.

Instead, intuitively each $\i \in I$ is a
\emph{possibly infinite} nesting of partial mappings of pairs of $O$- and $I$-annotations to names in
$\sig$---these pairs of annotations classify the possible effects of the corresponding interrupt handler code.
We use the
record notation
\[
\i = \{ \op_1 \mapsto (\o_1,\i_1) , \ldots , \op_n \mapsto (\o_n,\i_n) \}
\]
to mean that $\i$ maps $\op_1, \ldots, \op_n$ to the annotations $(\o_1,\i_1), \ldots, (\o_n,\i_n)$,
while any other names in $\sig$ are unannotated, corresponding to no interrupt handlers being installed for
these other names. We write $\i\, (\op_i) = (\o_i,\i_i)$ to mean that the annotation
$\i$ maps $\op_i$ to $(\o_i,\i_i)$.

Formally,
we define $I$ as the \emph{greatest fixed point}
of a set functor $\Phi$, given by
\[
\Phi (X) \defeq \sig \Rightarrow (O \times X)_\bot
\]
where $\Rightarrow$ is exponentiation, $\times$ is Cartesian product,
and $(-)_\bot$ is the lifting operation, which we use to represent unannotated names, and which is defined
using the disjoint union as $(-) \cupdot \{\bot\}$. Formally speaking, $I$ is given
by an isomorphism $I \cong \Phi(I)$, but for presentation purposes we leave it
implicit and work as if we had a strict equality $I = \Phi(I)$.

\paragraph{Subtyping and Recursive Effect Annotations}
Both $O$ and $I$ come equipped with natural \emph{partial orders}: for $O$, $\order O$ is given simply by
subset inclusion; and for $I$, the pointwise order~$\order I$ is characterised as follows:
\[
\begin{array}{l c l}
\i \order I \i'
&
\text{iff}
&
\forall\, (\op \in \sig) \, (\o'' \in O) \, (\i'' \in I) .\, \i\, (\op) = ({\o''} , {\i''}) \implies
\\[0.5ex]
&& \exists\, (\o''' \in O) \, (\i''' \in I) .\, \i'\, (\op) = ({\o'''} , {\i'''}) \wedge \o'' \order O \o''' \wedge \i'' \order I \i'''
\end{array}
\]
We also use the \emph{product order} $\order {O \times I}$, defined as
$(\o,\i) \order {O \times I} (\o',\i') \defeq \o \order O \o' \wedge \i \order I \i'$.
In particular, we use $\order {O \times I}$ to define the subtyping
relation for \lambdaAEff's computation types.

Furthermore, both $O$ and $I$ carry a \emph{join-semilattice} structure, where
$\o \sqcup \o' \in O$ is given simply by the union of sets $\o \cup \o'$, while
$\i \sqcup \i' \in I$ is given pointwise as follows:
\[
(\i \sqcup \i')(\op)
~\defeq~
\begin{cases}
(\o'' \sqcup \o''' , \i'' \sqcup \i''') & \mbox{if } \i\, (\op) = (\o'',\i'') \wedge \i'\, (\op) = (\o''',\i''') \\
(\o'' , \i'') & \mbox{if } \i\, (\op) = (\o'',\i'') \wedge \i'\, (\op) = \bot \\
(\o''' , \i''') & \mbox{if } \i\, (\op) = \bot \wedge \i'\, (\op) = (\o''',\i''') \\
\bot & \mbox{if } \i\, (\op) = \bot \wedge \i'\, (\op) = \bot \\
\end{cases}
\]

Importantly, the partial orders $(O,\order O)$ and $(I,\order I)$ are both \emph{$\omega$-complete} and \emph{pointed}, i.e.,
they form \emph{pointed $\omega$-cpos}, meaning that they have least upper bounds of all increasing $\omega$-chains, and
least elements (given by the empty set $\emptyset$ and the constant $\bot$-valued mapping, respectively).
As a consequence, and as is well-known, \emph{least fixed points} of continuous (endo)maps on them are then guaranteed
to exist~\cite{Amadio:Domains, Gierz:ContinuousLattices}.
For \lambdaAEff, we are particularly interested in the least fixed points of continuous maps $f : I \to I$,
so as to specify and typecheck code examples involving reinstallable interrupt handlers, as we illustrate in
\autoref{sec:extensions:reinstallable-interrupt-handlers}.

We also note that if we were only interested in the type safety of \lambdaAEff, and not
in typechecking reinstallable interrupt handler examples, then we would not need $(I,\order I)$ to be \emph{$\omega$-complete},
and could have instead chosen $I$ to be the
\emph{least fixed point} of the set functor $\Phi$ defined earlier, which is what we do for simplicity in our \pl{Agda}
formalisation. In this case, each interrupt handler annotation $\i \in I$ would be a \emph{finite} nesting of partial mappings.

Finally, we envisage that any future full-fledged high-level language based on \lambdaAEff~would
allow users to define their (recursive) effect annotations in a small domain-specific language, providing
a syntactic counterpart to the domain-theoretic development we use in this paper.

\paragraph{Interrupt Actions}

We mimic the act of triggering an interrupt handler for some interrupt $\op$ on an effect annotation $(\o, \i)$ through an
\emph{action} defined as follows:
\[
\opincomp {op} {(\o , \i)}
~\defeq~
  \begin{cases}
   \left(\o \sqcup \o' , \i[\op \mapsto \bot] \sqcup \i' \right) & \mbox{if } \i\, (\op) = (\o',\i')\\
   (\o,\i) & \mbox{otherwise}
  \end{cases}
\]
If $(\o, \i)$ lists any interrupt handlers installed for $\op$, then $\i\, (\op) = (\o',\i')$,
where $(\o',\i')$ specifies the effects of said handler code. Now, when the inward propagating
interrupt reaches those interrupt handlers, it triggers the execution of the corresponding handler code,
and thus the entire interrupted computation can also issue signals in $\o'$ and handle interrupts in $\i'$.

The notation $\i[\op \mapsto \bot]$ sets $\i$ to $\bot$ at $\op$,
and leaves it unchanged elsewhere.
Mapping $\op$ to $\bot$ in the definition of $\tmkw{\downarrow}$ captures that the interrupt $\op$ triggers all the corresponding interrupt handlers that are installed in the computation that it is propagated to.

\subsubsection{Typing Rules}
\label{sect:typing-rules}

We characterise \emph{well-typed values} using the judgement $\Gamma \types V : X$
and \emph{well-typed computations} using the judgement $\Gamma \types M : \tycomp{X}{(\o,\i)}$.
In both judgements, $\Gamma$ is a \emph{typing context}.
The rules defining these judgements are respectively given in \autoref{fig:value-typing-rules} and
\ref{fig:computation-typing-rules}.

\begin{figure}[h]
  \centering
  \small
  \begin{mathpar}
  \coopinfer{TyVal-Var}{
  }{
    \Gamma, x \of X, \Gamma' \types x : X
  }
  \qquad
  \coopinfer{TyVal-Unit}{
  }{
    \Gamma \types \tmunit : \tyunit
  }
  \qquad
  \coopinfer{TyVal-Pair}{
    \Gamma \types V : X \\
    \Gamma \types W : Y
  }{
    \Gamma \types \tmpair{V}{W} : \typrod{X}{Y}
  }
  \qquad
  \coopinfer{TyVal-Promise}{
    \Gamma \types V : X
  }{
    \Gamma \types \tmpromise V : \typromise X
  }
  \\
  \coopinfer{TyVal-Inl}{
    \Gamma \types V : X
  }{
    \Gamma \types \tminl[Y]{V} : X + Y
  }
  \qquad
  \coopinfer{TyVal-Inr}{
    \Gamma \types W : Y
  }{
    \Gamma \types \tminr[X]{W} : X + Y
  }
  \qquad
  \coopinfer{TyVal-Fun}{
    \Gamma, x \of X \types M : \tycomp{Y}{(\o,\i)}
  }{
    \Gamma \types \tmfun{x : X}{M} : \tyfun{X}{\tycomp{Y}{(\o,\i)}}
  }
  \end{mathpar}
  \caption{Value Typing Rules.}
  \label{fig:value-typing-rules}
\end{figure}

\begin{figure}
  \flushright
  \small
  \renewcommand{\arraystretch}{4}
  \addtolength{\tabcolsep}{-0.2em}
  \begin{tabular}{cc}
  $
  \coopinfer{TyComp-Return}{
    \Gamma \types V : X
  }{
    \Gamma \types \tmreturn{V} : \tycomp{X}{(\o,\i)}
  }
  $
  & $
  \coopinfer{TyComp-Let}{
    \Gamma \types M : \tycomp{X}{(\o,\i)}
    \\
    \Gamma, x \of X \types N : \tycomp{Y}{(\o,\i)}
  }{
    \Gamma \types
    \tmlet{x}{M}{N} : \tycomp{Y}{(\o,\i)}
  }
  $ \\

  $
  \coopinfer{TyComp-Apply}{
    \Gamma \types V : \tyfun{X}{\tycomp{Y}{(\o,\i)}} \\
    \Gamma \types W : X
  }{
    \Gamma \types \tmapp{V}{W} : \tycomp{Y}{(\o,\i)}
  }
  $ &
  $
  \coopinfer{TyComp-MatchPair}{
    \Gamma \types V : \typrod{X}{Y} \\
    \Gamma, x \of X, y \of Y \types M : \tycomp{Z}{(\o,\i)}
  }{
    \Gamma \types \tmmatch{V}{\tmpair{x}{y} \mapsto M} : \tycomp{Z}{(\o,\i)}
  }
  $\\
  $
  \coopinfer{TyComp-MatchEmpty}{
    \Gamma \types V : \tyempty
  }{
    \Gamma \types \tmmatch[\tycomp{Z}{(\o,\i)}]{V}{} : \tycomp{Z}{(\o,\i)}
  }
  $ &
  $
  \coopinfer{TyComp-MatchSum}{
    \Gamma \types V : X + Y \\\\
    \Gamma, x \of X \types M : \tycomp{Z}{(\o,\i)} \\
    \Gamma, y \of Y \types N : \tycomp{Z}{(\o,\i)} \\
  }{
    \Gamma \types \tmmatch{V}{\tminl{x} \mapsto M, \tminr{y} \mapsto N} : \tycomp{Z}{(\o,\i)}
  }
  $\\
  $
  \coopinfer{TyComp-Signal}{
    \op \in \o \\
    \Gamma \types V : A_\op \\
    \Gamma \types M : \tycomp{X}{(\o,\i)}
  }{
    \Gamma \types \tmopout{op}{V}{M} : \tycomp{X}{(\o,\i)}
  }
  $&
  $
  \coopinfer{TyComp-Interrupt}{
    \Gamma \types V : A_\op \\
    \Gamma \types M : \tycomp{X}{(\o,\i)}
  }{
    \Gamma \types \tmopin{op}{V}{M} : \tycomp{X}{\opincomp {op} (\o,\i)}
  }
  $\\
  \multicolumn{2}{c}{
  $
  \coopinfer{TyComp-Promise}{
    ({\o'} , {\i'}) = \i\, (\op)  \\
    \Gamma, x \of A_\op \types M : \tycomp{\typromise X}{(\o',\i')} \\
    \Gamma, p \of \typromise X \types N : \tycomp{Y}{(\o,\i)}
  }{
    \Gamma \types \tmwith{op}{x}{M}{p}{N} : \tycomp{Y}{(\o,\i)}
  }$
  }
  \\
  $
  \coopinfer{TyComp-Await}{
    \Gamma \types V : \typromise X \\
    \Gamma, x \of X \types M : \tycomp{Y}{(\o,\i)}
  }{
    \Gamma \types \tmawait{V}{x}{M} : \tycomp{Y}{(\o,\i)}
  }
  $&
  $
   \coopinfer{TyComp-Subsume}{
      \Gamma \types M : \tycomp{X}{(\o, \i)} \\
      (\o,\i) \order {O \times I} (\o',\i')
    }{
      \Gamma \types M : \tycomp{X}{(\o', \i')}
    }
  $
  \end{tabular}
  \addtolength{\tabcolsep}{0.2em}
  \caption{Computation Typing Rules.}
  \label{fig:computation-typing-rules}
\end{figure}

\paragraph{Values}

The rules for values are mostly standard.
The only \lambdaAEff-specific rule is \textsc{TyVal-Promise}, which states that in order to fulfil
a \emph{promise} of type $\typromise X$, one has to supply a value of type $X$. In the rule \textsc{TyVal-Var}, we emphasise the position of the variable in the context, as it will become important once we extend the calculus with modal types in \autoref{sec:extensions:fitch-style-modal-types}.

\paragraph{Computations}

Analogously to values, the typing rules are standard for computation terms that \lambdaAEff~inherits from FGCBV,
with the \lambdaAEff-rules additionally tracking effect information.

The first \lambdaAEff-specific typing rule \textsc{TyComp-Signal} states that in order
to issue a signal $\op$ in a computation that has type $\tycomp{X}{(\o,\i)}$, we must have $\op \in \o$ and the type of
the payload value has to match $\op$'s signature $\op : A_\op$.

The rule \textsc{TyComp-Interrupt} is used to type incoming interrupts.
In particular, when the outside world propagates an interrupt $\op$ to a computation
$M$ of type $\tycomp{X}{(\o,\i)}$, the resulting
computation $\tmopin{op}{V}{M}$ gets assigned the type $\tycomp{X}{\opincomp {op} (\o,\i)}$,
where the action $\opincomp {op} (\o,\i)$ of the interrupt $\op$ on the annotation $(\o, \i)$ is given as
discussed in \autoref{sec:basic-calculus:effect-annotations}.

The rule \textsc{TyComp-Promise} states that
the interrupt handler code $M$ has to return a fulfilled promise of type $\typromise X$, for some type $X$,
while possibly issuing signals $\o'$ and handling interrupts $\i'$, both of which are
determined by the effect annotation $\i$ of the entire computation, as
$(\o',\i') = \i\, (\op)$. The variable $p$ bound in the continuation, which sub-computations can block on
to await $\op$ to arrive and be handled, also gets assigned
the promise type $\typromise X$.

It is worth noting that we could have had $M$ simply
return values of type $X$, but at the cost of not being able to implement some of the more interesting examples,
such as the guarded interrupt handlers defined in \autoref{sec:applications:guarded-handlers}.
At the same time, for \lambdaAEff's type safety, it is
crucial that $p$ would have remained assigned the distinguished promise type $\typromise X$.

The rule \textsc{TyComp-Await} simply states that after awaiting a promise of type $\typromise X$,
the continuation~$M$ can refer to the promised value using the variable $x$ of type $X$.

Finally, the rule \textsc{TyComp-Subsume} allows \emph{subtyping}, required to prove type preservation for rules where an interrupt encounters an interrupt handler.
To simplify the presentation, we consider a limited form of subtyping, in which we
shallowly relate only effect annotations.

\subsection{Type Safety}
\label{sec:basic-calculus:type-safety}

The sequential part of \lambdaAEff~satisfies the expected type safety properties
ensuring that ``well-typed programs do not go wrong''. We split these safety properties into the usual
\emph{progress} and \emph{preservation} theorems \cite{Wright:SynAppTypeSoundness}.
We omit their proofs~\cite{Ahman:POPL} from this summary, and revisit them
in \autoref{sec:type-safety} for the extended version of \lambdaAEff,
as the proofs for the extended calculus also apply to the version summarised in this section.

The progress result states that well-typed (and sufficiently) closed computations can either make another step of
reduction, or they are already in a well-defined result form (and thus have correctly stopped reducing).
As such, we first need to define when we consider \lambdaAEff-computations
to be in \emph{result form} (commonly also called a normal form). We do so using the
judgements $\CompResult {\Psi} {M}$,
which states that $M$ has reached its final form as an isolated computation term,
and $\RunResult {\Psi} {M}$, which states that $M$ has reached the final form of a
computation running inside a process with all its signals already having been propagated to
other parallel processes (described in more detail in \autoref{sec:basic-calculus:type-safety:processes}):
\begin{mathpar}
  \coopinfer{}{
    \CompResult {\Psi} {M}
  }{
    \CompResult {\Psi} {\tmopout {op} V M}
  }
  \and
  \coopinfer{}{
    \RunResult {\Psi} {M}
  }{
    \CompResult {\Psi} {M}
  }
  \and
  \coopinfer{}{
  }{
    \RunResult {\Psi} {\tmreturn V}
  }
  \and
  \coopinfer{}{
    \RunResult {\Psi \cup \{p\}} {N}
  }{
    \RunResult {\Psi} {\tmwith {op} x M p N}
  }
  \and
  \coopinfer{}{
    p \in \Psi
  }{
    \RunResult {\Psi} {\tmawait p x M}
  }
\end{mathpar}
In these judgements, $\Psi$ is a set
of (promise-typed) variables $p$ that have been bound by interrupt handlers enveloping the given computation.
Intuitively, these judgements express that a computation $M$ is in a (top-level)
result form $\CompResult {\Psi} {M}$ when, considered as a tree, it has a shape in which \emph{all}
signals are towards the root, interrupt handlers are in the intermediate nodes, and
the leaves contain return values and computations that are temporarily blocked
while awaiting one of the promise-typed variables $p$ in $\Psi$ to be fulfilled.

The new reduction rules that propagate the awaiting construct out of sequencing and interrupts into the awaiting construct ensure the explicit form of all blocking computations and considerably simplify the definition of $\RunResult {\Psi} {M}$ compared to the previous version of our work~\cite{Ahman:POPL}.
The finality of these result forms is captured by the next lemma.

\begin{lem}
\label{lem:results-are-final}
Given $\Psi$ and $M$, such that $\CompResult {\Psi} {M}$, then there is no $N$ with $M \reduces N$.
\end{lem}

Using the result forms, the progress theorem for the sequential part of \lambdaAEff\ is as follows:
\noindent
\begin{thm}[Progress for computations]
\label{thm:progress}
Given a well-typed computation
\[
  p_1 \of \typromise {X_1}, \ldots, p_n \of \typromise {X_n} \types M : \tycomp{Y}{(\o,\i)}
\]
then either
\begin{enumerate}[(a)]
  \item there exists a computation $N$, such that $M \reduces N$, or
  \item the computation $M$ is in a result form, i.e., we have $\CompResult {\{p_1, \ldots, p_n\}} {M}$.
\end{enumerate}
\end{thm}
\noindent
In particular, with the empty context, we get the usual progress statement, which states that
$\types M : \tycomp{X}{(\o, \i)}$ implies that either $M \reduces N$ for some $N$
or that $\CompResult {\emptyset} {M}$ holds. This implies that any promise variable
which we are awaiting to be fulfilled must correspond to one of the installed interrupt handlers.
Additionally, the type system ensures that all outgoing signals are listed in $\o$ and all
installed interrupt handlers are specified in $\i$.

The type preservation result is standard and says that reduction preserves well-typedness.

\begin{thm}[Preservation for computations]
\label{thm:preservation}
Given a computation $\Gamma \types M : \tycomp{X}{(\o,\i)}$, such that $M$
can reduce as $M \reduces N$, then we have $\Gamma \types N : \tycomp{X}{(\o,\i)}$.
\end{thm}


\section{A Calculus for Asynchronous Effects: Parallel Processes}
\label{sec:basic-calculus:processes}

We now describe the parallel part of \lambdaAEff. Similarly to the sequential part, we 
present the corresponding syntax, small-step semantics, 
type-and-effect system, and type safety results.

\subsection{Parallel Processes}

To keep the presentation focussed on the asynchronous use of algebraic effects, we 
consider a very simple model of parallelism: a process is either an \emph{individual computation} 
or the \emph{parallel composition} of two processes. To facilitate interactions, processes also  
contain outward propagating \emph{signals} and inward propagating \emph{interrupts}. 

In detail, the syntax of \emph{parallel processes} is given by the following grammar:
\[
  P, Q
  \bnfis \tmrun M
  \,\bnfor\! \tmpar P Q
  \,\bnfor\! \tmopout{op}{V}{P}
  \,\bnfor\! \tmopin{op}{V}{P}
\]
Note that processes do not include interrupt handlers---these are local to computations.

Here the number and hierarchy of processes running in parallel is fixed---a limitation that we address
in \autoref{sec:extensions:dynamic-process-creation} by introducing a means to dynamically create new processes.

\subsection{Small-Step Operational Semantics}

We equip the parallel part of \lambdaAEff~with a small-step operational semantics that  
naturally extends the semantics of \lambdaAEff's sequential part from \autoref{sec:basic-calculus:semantics:computations}.
The semantics is defined using a reduction relation $P \reduces Q$, as given in \autoref{fig:processes}.

\begin{figure}[h]
  \parbox{\textwidth}{
  \centering
  \small
  \begin{minipage}[t]{0.45\textwidth}
  \centering
  \begin{align*}
  \intertext{\textbf{Individual computations}}
    \coopinfer{}{
      M \reduces N
    }{
      \tmrun M \reduces \tmrun N
    }
  \end{align*}
  \begin{align*}
    \intertext{\textbf{Signal hoisting}}
    \tmrun {(\tmopout{op}{V}{M})}  &\reduces \tmopout{op}{V}{\tmrun M}
    \\[1ex]
    \intertext{\textbf{Broadcasting}}
    \tmpar{\tmopout{op}{V}{P}}{Q} &\reduces \tmopout{op}{V}{\tmpar{P}{\tmopin{op}{V}{Q}}}
    \\
    \tmpar{P}{\tmopout{op}{V}{Q}} &\reduces \tmopout{op}{V}{\tmpar{\tmopin{op}{V}{P}}{Q}}
  \end{align*}
  \end{minipage}
  \qquad
  \begin{minipage}[t]{0.47\textwidth}
  \centering
  \begin{align*}
    \intertext{\textbf{Interrupt propagation}}
    \tmopin{op}{V}{\tmrun M} &\reduces \tmrun {(\tmopin{op}{V}{M})}
    \\
    \tmopin{op}{V}{\tmpar P Q} &\reduces \tmpar {\tmopin{op}{V}{P}} {\tmopin{op}{V}{Q}}
    \\
    \tmopin{op}{V}{\tmopout{op'}{W}{P}} &\reduces \tmopout{op'}{W}{\tmopin{op}{V}{P}}
  \end{align*}
  \begin{align*}
    \intertext{\quad\textbf{Evaluation context rule}}
    \quad
    \coopinfer{}{
      P \reduces Q
    }{
      \F[P] \reduces \F[Q]
    }
  \end{align*}
  \end{minipage}
  \begin{align*}
  \intertext{\textbf{where}}
  \text{$\F$}
  \bnfis& [~]
  \bnfor \tmpar \F Q \bnfor\! \tmpar P \F
  \bnfor \tmopout{op}{V}{\F}
  \bnfor \tmopin{op}{V}{\F}
  \end{align*}
  } 
  \caption{Small-Step Operational Semantics of Processes.}
  \label{fig:processes}
\end{figure}

\paragraph{Individual Computations}
This rule states that, as processes, individual computations evolve according to the small-step
operational semantics $M \reduces N$ we defined in \autoref{sec:basic-calculus:semantics:computations}.

\paragraph{Signal Hoisting}
This rule propagates signals out of individual computations.
Note that we only hoist those signals that have propagated to the outer boundary
of a computation.

\paragraph{Broadcasting}
These rules turn outward moving signals in one process into inward moving interrupts 
for the process parallel to it, while continuing to propagate the signals outwards to any 
further parallel processes. The latter ensures that the semantics is compositional.

\paragraph{Interrupt Propagation}
These three rules simply propagate interrupts inwards into individual computations, 
into all branches of parallel compositions, and past any issued signals.

\paragraph{Evaluation Contexts}
Analogously to the semantics of computations, the semantics of processes presented here also 
includes an evaluation context rule, which allows reductions under \emph{evaluation contexts} 
$\F$. Observe that compared to the evaluation contexts for computations, those for processes
are more standard, in the sense that they do not bind variables. 

\subsection{Type-and-Effect System}
\label{sec:basic-calculus:processes:type-and-effect-system}

Analogously to its sequential part, we also equip \lambdaAEff's parallel part with a type-and-effect system.

\paragraph{Types} The \emph{process types} are designed to match their parallel structure, and 
are given by
\[
  \text{$\tyC$, $\tyD$}
  \bnfis \tyrun X \o \i
  \bnfor \typar \tyC \tyD
\]
Namely, $\tyrun X \o \i$ is a type of an individual computation of type $\tycomp{X}{(\o,\i)}$, and $\typar \tyC \tyD$
is the type of the parallel composition of two processes that respectively have types $\tyC$ and $\tyD$.

\paragraph{Typing Judgements}
\emph{Well-typed processes} are characterised using the judgement
$\Gamma \vdash P : \tyC$. The typing rules are given in \autoref{fig:process-typing-rules}.
While our processes are not currently higher-order, we allow 
non-empty contexts $\Gamma$ to model using libraries and top-level function definitions.

\begin{figure}[h]
  \centering
  \small
  \begin{mathpar}
  \coopinfer{TyProc-Run}{
    \Gamma \types M : \tycomp{X}{(\o,\i)}
  }{
    \Gamma \types \tmrun{M} : \tyrun{X}{\o}{\i}
  }
  \and
  \coopinfer{TyProc-Par}{
    \Gamma \types P : \tyC \\
    \Gamma \types Q : \tyD
  }{
    \Gamma \types \tmpar{P}{Q} : \typar{\tyC}{\tyD}
  }
  \\
  \coopinfer{TyProc-Signal}{
    \op \in \mathsf{signals\text{-}of}{(\tyC)} \\\\
    \Gamma \types V : A_\op \\
    \Gamma \types P : \tyC 
  }{
    \Gamma \types \tmopout{op}{V}{P} : \tyC
  }
  \and
  \coopinfer{TyProc-Interrupt}{
    \Gamma \types V : A_\op \\
    \Gamma \types P : \tyC 
  }{
    \Gamma \types \tmopin{op}{V}{P} : \opincomp{op}{\tyC}
  }  
  \end{mathpar}
  \caption{Process Typing Rules.}
  \label{fig:process-typing-rules}
\end{figure}

\noindent
The rules \textsc{TyProc-Run} and \textsc{TyProc-Par} capture the earlier 
intuition about the types of processes matching their parallel structure. The rules 
\textsc{TyProc-Signal} and \textsc{TyProc-Interrupt} are similar to the corresponding 
computation typing rules from \autoref{fig:computation-typing-rules}.

The \emph{signal annotations} of a process type used in \textsc{TyProc-Signal} are calculated as
\[
\mathsf{signals\text{-}of}(\tyrun{X}{\o}{\i}) ~\defeq~ \o
\qquad\qquad
\mathsf{signals\text{-}of}(\typar{\tyC}{\tyD}) ~\defeq~ \mathsf{signals\text{-}of}(\tyC) \sqcup \mathsf{signals\text{-}of}(\tyD)
\]
and the \emph{action of interrupts} on process types extends the action on effect annotations as
\[
\opincomp{op}(\tyrun{X}{\o}{\i}) 
~\defeq~
X \att (\opincomp {op} {(\o , \i)})
\qquad\qquad
\opincomp{op}(\typar{\tyC}{\tyD}) 
~\defeq~
\typar{(\opincomp{op}{\tyC})}{(\opincomp{op}{\tyD})}
\]
by propagating the interrupt towards the types of individual computations. 

It is worth noting that \autoref{fig:process-typing-rules} does not include an analogue  
of the computation subtyping rule \textsc{TyComp-Subsume}. This choice is 
deliberate because as we shall see below, \emph{process types reduce}
in conjunction with the processes they are assigned to, and the outcome   
of process type reduction is generally neither a sub- nor supertype of the original type.

\subsection{Type Safety}
\label{sec:basic-calculus:type-safety:processes}

We conclude summarising the meta-theory of \lambdaAEff~by stating the type safety 
of its parallel part. Analogously to \autoref{sec:basic-calculus:type-safety}, 
we once again split type safety into separate \emph{progress} 
and \emph{preservation} results, and relegate their proofs to \autoref{sec:type-safety}.

We characterise the \emph{result forms} of processes 
by defining two judgements, $\ProcResult P$ and $\ParResult P$, 
and by using the judgement $\RunResult {\Psi} {M}$ from 
\autoref{sec:basic-calculus:type-safety}, as follows:
\begin{mathpar}
  \coopinfer{}{
    \ProcResult {P}
  }{
    \ProcResult {\tmopout {op} V P}
  }
  \qquad
  \coopinfer{}{
    \ParResult {P}
  }{
    \ProcResult {P}
  }
  \qquad
  \coopinfer{}{
    \RunResult {\emptyset} {M}
  }{
    \ParResult {\tmrun M}
  }
  \qquad
  \coopinfer{}{
    \ParResult P \\
    \ParResult Q
  }{
    \ParResult {\tmpar P Q}
  }
\end{mathpar}
These judgements express that a process $P$ is in a (top-level) 
result form $\ProcResult {P}$ when, considered as a tree, it has a shape in which 
\emph{all} signals are towards the root, parallel compositions are in 
the intermediate nodes, and individual computation results are at the leaves. 
Importantly, the computation results $\RunResult {\emptyset} {M}$ we use in this definition 
are those from which all signals have been propagated out of 
(as discussed in \autoref{sec:basic-calculus:type-safety}). 

Again, these result forms are operationally final, as captured by the next lemma.

\begin{lem}
\label{lem:results-are-final:processes}
Given a process $P$, such that $\ProcResult {P}$, then there is no $Q$ such that $P \reduces Q$.
\end{lem}

We are now ready to state the progress theorem for the parallel part of \lambdaAEff,
which applies to closed processes and takes the expected form:

\begin{thm}[Progress for processes]
  \label{thm:procprogress}
  Given a well-typed process $\types P : \tyC$, then either
  \begin{enumerate}[(a)]
    \item there exists a process $Q$, such that $P \reduces Q$, or
    \item the process $P$ is already in a (top-level) result form, i.e., we have $\ProcResult {P}$.
  \end{enumerate}
\end{thm}
\noindent
The preservation theorem for processes that we state below is somewhat non-standard since term reductions 
also evolve effect annotations. In particular, the broadcast rule
\[
  \tmpar{\tmopout{op}{V}{P}}{Q} \reduces \tmopout{op}{V}{\tmpar{P}{\tmopin{op}{V}{Q}}}
\]
and its symmetric counterpart from \autoref{fig:processes} introduce new 
inward propagating interrupts in their right-hand sides that originally do not exist in their left-hand sides. As a result, 
compared to the types one assigns to the left-hand sides of these reduction rules, the types assigned to 
their right-hand sides will need to feature corresponding type-level actions of these interrupts.
We formalise this idea using a \emph{process type reduction} relation $\tyC \tyreduces \tyD$:
\[
  \coopinfer{}{
  }{
    \tyrun{X}{\o}{\i} \tyreduces \tyrun{X}{\o}{\i} 
  }
  \hspace{2pt}
  \coopinfer{}{
  }{
    X \att (\opincompp {ops} {(\o , \i)}) \tyreduces X \att (\opincompp {ops} {(\opincomp {op} {(\o , \i)})})
  }
  \hspace{2pt}
  \coopinfer{}{
    \tyC \tyreduces \tyC' \\
    \tyD \tyreduces \tyD'
  }{
    \typar{\tyC}{\tyD} \tyreduces \typar{\tyC'}{\tyD'}
  }
\]
where we write $\opincompp {ops} {(\o , \i)}$ for a recursively defined \emph{action of a list of interrupts} on $(\o , \i)$:
\[
\opincompp {[]} {(\o , \i)} ~\defeq~ (\o , \i)
\qquad
\opincompp {(\op :: \opsym{ops})} {(\o , \i)} ~\defeq~ \opincomp {op} {(\opincompp {ops} (\o , \i))}
\]

Intuitively, $\tyC \tyreduces \tyD$ describes how process types reduce by being acted upon by 
freshly arriving interrupts.
It is important that we introduce interrupts under an arbitrary enveloping sequence of interrupt actions, 
and not simply as $X \att {(\o , \i)} \tyreduces X \att (\opincomp {op} {(\o , \i)})$,
because we want to ensure that these actions preserve type reductions (see~\srefcase{Lemma}{lem:type-reduction}{3}),
which in turn ensures type preservation of reductions under arbitrary evaluation contexts~$\F$.

Using the process type reduction relation, we state the preservation theorem for the parallel part of 
\lambdaAEff~as follows:

\begin{thm}[Preservation for processes]
  \label{thm:procpreservation}
  Given a well-typed process $\Gamma \types P : \tyC$, such that $P$ can reduce as 
  $P \reduces Q$, then there exists a process type $\tyD$, such 
  that the process type $\tyC$ can reduce as $\tyC \tyreduces \tyD$, 
  and we can type the resulting process as $\Gamma \types Q : \tyD$.
\end{thm}


\section{Higher-Order Extensions}
\label{sec:higher-order-extensions}

While \lambdaAEff, as introduced in our original work \cite{Ahman:POPL} and summarised in the previous
two sections, can be used to naturally capture a wide range of asynchronous examples, 
it also has many notable \emph{limitations}: interrupt handlers disappear immediately after being triggered by a 
matching interrupt, payloads of signals and interrupts have to be ground values, and it is not possible to
dynamically create new parallel processes. In this section we introduce and discuss a number of 
\emph{higher-order extensions} of \lambdaAEff~that resolve these limitations. Below we discuss each of these 
extensions individually, with the full extended calculus given in~\autoref{sec:appendix}.
We highlight the parts of \lambdaAEff~that change in this section's extensions with a grey background.

\subsection{Reinstallable Interrupt Handlers}
\label{sec:extensions:reinstallable-interrupt-handlers}

We recall from the reduction rules in \autoref{fig:small-step-semantics-of-computations} that
once an interrupt reaches a matching interrupt handler, the handling computation is executed
and the handler is removed. However, the example from \autoref{sec:overview:runningexample}
shows that we often want to keep the handler around, e.g., to handle further interrupts
of the same kind. One option to achieve this is through general recursion~\cite{Ahman:POPL}.
Unfortunately, this results in programmers defining many auxiliary functions, obfuscating the resulting code. 
Furthermore, the heavy reliance on general recursion makes it difficult to justify leaving it out of the core 
calculus, despite it being an orthogonal concern to many programming abstractions and, in particular, to how
we model asynchrony in \lambdaAEff~based on algebraic effects---this is of course not to say that a 
higher-level language based on \lambdaAEff~could not include general recursion.

Instead, in this paper we propose extending \lambdaAEff's interrupt handlers with the ability to 
\emph{reinstall} themselves, by extending the syntax for interrupt handlers given in~\autoref{sec:basic-calculus:computations}
\[
    \tmwith{op}{x}{M}{p}{N}
\]
with an additional variable $r$ bound to a function through which $M$ can reinstall the handler:
\[
    \tmwithre{op}{x}{\highlightgray{\vphantom{/}r}}{M}{p}{N}
\]

In contrast to the continuation/resumption variables of ordinary effect handlers, here the variable $r$ does not refer 
to the continuation of the interrupt at the time triggering, but instead to the act of reinstalling the given interrupt handler.
Concretely, the triggering of reinstallable interrupt handlers is captured operationally with the following 
reduction rule:
\begin{multline*}
  \tmopin{op}{V}{\tmwithre{op}{x}{r}{M}{p}{N}} \\
  \reduces \tmlet{p}{M[V/x,\highlightgray{(\tmfun{}{\tmwithre{op}{x}{r}{M}{p}{\tmreturn p}})/r}]}{\tmopin{op}{V}{N}}
\end{multline*}
All other reduction rules remain the same, except that interrupt handlers are extended with the additional variables $r$.
Server-like processes can then be written more concisely, as
\[
\tmwithre{request}{x}{r}{\text{handle the 
$\opsym{request}$}; \text{issue a $\opsym{response}$ signal}; \highlightgray{r\, ()}}{p}{\tmreturn ()}
\]

In light of the similarity between interrupt propagation and deep effect handling, as discussed in 
\autoref{sec:basic-calculus:semantics:computations}, this reinstalling behaviour can be 
understood as an effect handler re-calling (in its corresponding operation case) the algebraic 
operation that it is handling, such as, an exception handler handling an exception and then re-raising it at the end for other, external exception handlers.

The typing rule for reinstallable interrupt handlers is also quite interesting:
\[
    \coopinfer{TyComp-RePromise}{
      ({\o'} , {\i'}) \mathrel{\highlightgray{\order{O \times I}}} \i\, (\op) \\
        \Gamma, x \of A_\op, \highlightgray{r \of \tyfun{\tyunit}{\tycomp{\typromise X}{\big(\emptyset, \{ \op \mapsto ({\o'} , {\i'}) \}\big)}}} \types M : \tycomp{\typromise X}{(\o',\i')} \\
        \Gamma, p \of \typromise X \types N : \tycomp{Y}{(\o,\i)} 
      }{
        \Gamma \types \tmwithre{op}{x}{r}{M}{p}{N} : \tycomp{Y}{(\o,\i)}
      }        
\]
First, observe that the context in which we type the interrupt handler code $M$ is now extended 
with the variable $r$, which denotes a function triggered by application to the unit value 
$\tmunit \of \tyunit$. The function does not emit any signals nor install any handlers apart from the one
in question for $\opsym{op}$, 
therefore its effect annotation is $\big(\emptyset, \{ \op \mapsto ({\o'} , {\i'}) \}\big)$, as expected.

Second, we have relaxed the requirement $({\o'} , {\i'}) = \i\, (\op)$. We now only require the effect 
annotation $({\o'} , {\i'})$ of the handler code $M$ to be contained in what the effect annotation $(\o, \i)$ of the continuation 
$N$, and thus of the entire composite computation $\tmwithre{op}{x}{r}{M}{p}{N}$, assigns to $\op$, i.e.,  $({\o'} , {\i'}) \mathrel{\order{O \times I}} \i\, (\op)$. The reason lies in the 
proof of type preservation (see \autoref{thm:preservation:extended}) when propagating unhandled 
interrupts past handlers:
\[
  \tmopin{op'}{V}{\tmwithre{op}{x}{r}{M}{p}{N}}
  \reduces \tmwithre{op}{x}{r}{M}{p}{\tmopin{op'}{V}{N}}
\]
On the left-hand side of this reduction rule, the effect annotation of the continuation of $\tmkw{promise}$ is $(\o, \i)$,
while on the right-hand side it is $\opincomp {op'} ({\o} , {\i})$. This mismatch did not pose a problem 
earlier~\cite{Ahman:POPL} as subtyping allowed us to increase the effect
annotation of $M$ to $\pi_2\, (\opincomp {op'} ({\o} , {\i})) (\op)$. Now on the other hand, as $M$'s effect annotation 
is also present in the type of $r$, it appears both co- and contravariantly, and is thus not safe to increase.
However, the tight coupling of the effect annotations is not really essential, as for safety 
it is enough that the annotation of the continuation simply encompasses any effects that $M$ may trigger.

As noted in \autoref{sec:basic-calculus:effect-annotations}, assigning types to reinstallable handlers
requires us to consider least fixed points of continuous maps on the $\omega$-cpo $(I,\order I)$ of interrupt handler annotations.
As an example, we recall the following fragment of the server code from \autoref{sec:overview:runningexample:server}:
\begin{aeffbox}
\begin{lstlisting}
...
promise (batchSizeReq () r |->
    send batchSizeResp batchSize;
    r ()
)
...
\end{lstlisting}
\end{aeffbox}
Here, the interrupt handler for $\opsym{batchSizeReq}$ reinstalls itself immediately 
after issuing a $\opsym{batchSizeResp}$ signal. 
Due to its recursive definition, 
it should not be surprising that this handler's effect annotation is given recursively, in 
particular, if we want to give it a more precise type-level specification than one which simply states that any effect is possible.

To that end, we assign this interrupt handler the effect annotation $(\emptyset, \i_{\text{b}})$, where 
\[
\i_{\text{b}} = \big\{~ \opsym{batchSizeReq} \mapsto (\{\opsym{batchSizeResp}\} , \{~ \opsym{batchSizeReq} \mapsto (\{\opsym{batchSizeResp}\} , ~\ldots~) ~\}) ~\big\}
\]
More precisely, $\i_{\text{b}}$ is the least fixed point of the following continuous map on $I$:
\[
    \i \mapsto \{~ \opsym{batchSizeReq} \mapsto (\{\opsym{batchSizeResp}\} , \i) ~\} : I \to I
\]
This least fixed point exists because $I$ is an $\omega$-cpo and the map is continuous
(see \autoref{sec:basic-calculus:effect-annotations}).

Returning to the example above, the effect annotation $(\emptyset, \i_{\text{b}})$  
specifies that the interrupt handler does not issue any signals at the top level,
and that every $\opsym{batchSizeReq}$
interrupt causes a $\opsym{batchSizeResp}$ signal to be issued and the interrupt handler 
to be reinstalled.

The examples of reinstallable interrupt handlers that we discuss in \autoref{sec:applications}
have their effect annotations assigned analogously, also as least fixed points of 
continuous maps on $I$.

\subsection{Stateful Reinstallable Interrupt Handlers}
\label{sec:extensions:stateful-reinstallable-interrupt-handlers}

When working with reinstallable interrupt handlers, it is often useful, and sometimes even 
necessary, to be able to pass data between subsequent reinstalls of a handler. For 
example, in \autoref{sec:applications:runners} we use reinstallable interrupt handlers to 
implement a pseudorandom number generator in which it is crucial to be able
to pass and update a seed value between reinstalls of an interrupt handler. As
another example, consider wanting to react to only the first $n$ interrupts of
a particular kind---here it is useful if we could pass and decrease a counter between 
handler reinstalls.
 
In our original work~\cite{Ahman:POPL}, such state-passing behaviour was achieved by
passing the relevant state values as arguments to the general-recursive functions that implemented
the reinstalling of interrupt handlers. However, with reinstallability of interrupt
handlers being now a primitive feature of \lambdaAEff, we want a similarly primitive approach to 
managing state.

To this end, we extend the reinstallable interrupt handlers of last section with \emph{state}:
\[
    \tmwithrest[\highlightgray{S}]{op}{x}{r}{\highlightgray{\vphantom{/}s}}{M}{\highlightgray{V}}{p}{N}
\]
Here $S$ denotes the type of state associated with a particular interrupt handler 
($S$ can be an arbitrary value type), $s$ is a variable bound in the interrupt handler code $M$, 
giving it access to the handler's state at the time of triggering, and $V$ is the value of state
that is used at the next triggering of the interrupt handler. The state can be updated between 
subsequent reinstalls of the interrupt handler by calling the reinstallation function $r$ with the 
new state value---$r$'s domain is now $S$ instead of $\tyunit$. This behaviour is summarised by the reduction rule
\[
\tmopin{op}{V}{\tmwithrest[S\!]{op}{x}{r}{s}{\!M}{W}{p}{N}} \reduces \tmlet{p}{M\big[V/x , R/r , \highlightgray{W/s}\big]}{\tmopin{op}{V}{N}}
\]
where $R$ denotes a function that reinstalls the interrupt handler
with an updated state value:
\[
R ~\defeq~ \tmfun{\highlightgray{\vphantom{/}s' \of S}}{\tmwithrest[S]{op}{x}{r}{s}{M}{\highlightgray{\vphantom{/}{s'}}}{p}{\tmreturn p}}
\]
Needing to know $S$ for the function abstraction in $R$ necessitates the type annotation on 
this variant of interrupt handlers.
All other reduction rules remain unchanged, except that interrupt handlers now include additional variables, type annotations, 
and values for states.

The typing rule for stateful reinstallable interrupt handlers is a straightforward
extension of the typing rule for reinstallable interrupt handlers we presented in the previous section:
\[
  \coopinfer{TyComp-ReStPromise}{
    ({\o'} , {\i'}) \mathrel{\order{O \times I}} \i\, (\op) \\
    \Gamma, x \of A_\op, r \of \tyfun{\highlightgray{\vphantom{/}S}}{\tycomp{\typromise X}{\big(\emptyset, \{ \op \mapsto ({\o'} , {\i'}) \}\big)}}, \highlightgray{\vphantom{/}s \of S} \types M : \tycomp{\typromise X}{(\o',\i')} \\
    \highlightgray{\vphantom{/}\Gamma \types V : S} \\
    \Gamma, p \of \typromise X \types N : \tycomp{Y}{(\o,\i)} 
  }{
    \Gamma \types \tmwithrest[S]{op}{x}{r}{s}{M}{V}{p}{N} : \tycomp{Y}{(\o,\i)}
  }        
\]
Observe that as noted above, the domain of $r$ is no longer fixed to the unit type $\tyunit$ 
but it can now be any value type $S$. If we pick $S \defeq \tyunit$, we recover the 
stateless reinstallable interrupt handlers of the previous section, with all the highlighted
parts trivialising. Therefore, as a convention, when working with 
reinstallable interrupt handlers with trivial state, we use the syntax introduced in the 
previous section, i.e., $\tmwithre{op}{x}{r}{M}{p}{N}$. Further, in examples 
we often omit the state type annotation $S$ when it is clear from the context.

We now illustrate the use of stateful reinstallable interrupt handlers via the example
mentioned earlier, of a program reacting to only the first $n$ interrupts of a particular kind:
\begin{aeffbox}
\begin{lstlisting}
promise (op x r m |-> 
    if (m > 0) then 
        comp;
        r (m - 1)
    else
        return <<()>>
) @ n
\end{lstlisting}
\end{aeffbox}
This interrupt handler carries a natural number counter as its state, which it  
uses to determine whether the handler computation \ls{comp} should be run.
The counter is originally set to the value \ls{n} and then decremented each time
the interrupt handler is reinstalled (using \ls{r (m - 1)}). When the counter reaches $0$, 
\ls{comp} is not run and the interrupt handler is no longer reinstalled.
More examples of stateful reinstallable interrupt handlers can be 
found in \autoref{sec:overview:runningexample} and \ref{sec:applications}.

Finally, it is also worth noting that while the syntax of our stateful reinstallable 
interrupt handlers is somewhat similar to parameterised effect handlers~\cite{Plotkin:HandlingEffects}, 
there is a subtle but important difference. Namely, as discussed in \autoref{sec:basic-calculus:semantics:computations}, 
interrupt handlers behave like algebraic operation calls and it is instead the interrupts that behave like
effect handling. Thus, in light of the discussion in \autoref{sec:extensions:reinstallable-interrupt-handlers}
about what reinstalling means, the stateful nature of our reinstallable interrupt handlers corresponds to changing a
(state) parameter of an algebraic operation when it is re-called by the 
corresponding effect handler, and not to including and passing state values in effect handlers. In particular,
any interrupt and its payload is passed to the continuation of an interrupt
handler unchanged irrespectively of any state changes that happen when 
this interrupt handler is triggered and (possibly) reinstalled.

\subsection{Fitch-Style Modal Types}
\label{sec:extensions:fitch-style-modal-types}

The next limitation of \lambdaAEff~we address is the restriction of signal and interrupt 
payloads to ground types, i.e., finite sums and products of base types. The reason 
behind this restriction lies in the propagation of signals past interrupt handlers:
\[
  \tmwith{op}{x}{M}{p}{\tmopout{op'}{V}{N}}
  \reduces \tmopout{op'}{V}{\tmwith{op}{x}{M}{p}{N}}
\]
Here, we want to ensure that the value $V$ on the left-hand side does not refer to the
promise-typed variable $p$, otherwise the right-hand side would be ill-scoped.
Note that the issue remains exactly the same when considering reinstallable or stateful interrupt handlers.

For example, consider a signal/interrupt $\op : \tyint$ carrying an
integer payload, and $\op' : \typromise{\tyint}$ carrying an integer-typed promise as a payload.
Then, the computation
\[
  \tmwith{op}{x}{\tmreturn{\tmpromise{1}}}{p}{\tmopout{op'}{p}{\tmreturn{2}}}
\]
which simply sends the promise-typed variable $p : \typromise{\tyint}$ back in a signal payload,
would be well-typed if no restrictions were put on signal and interrupt signatures. 
However, if this were allowed, then by the above reduction rule, this computation 
would step to 
\[
  \tmopout{op'}{p}{\tmwith{op}{x}{\tmreturn{\tmpromise{1}}}{p}{\tmreturn{2}}}
\]
where now the payload $p$ has escaped the binding scope of the interrupt handler, 
violating scope and type safety.

We run into similar problems when considering examples where payloads are higher-order, e.g., 
when wanting to send functions in payloads for remote execution in other processes. For example,
take $\op$ as before and consider a function-carrying signal/interrupt $\op' : (\tyunit \to \tyint)$,
where for brevity, we omit the effect annotation in the function type.
Then, the computation
\begin{align*}
  &\tmwith{op}{x}{\tmreturn{\tmpromise{1}}}{p}{\\&\tmlet{f}{\tmreturn{(\tmfunano{\tmunit}{\tmawaitgen{p}})}}\\&\tmopout{op'}{f}{\tmreturn{2}}}
\end{align*}
would again be well-typed if no restrictions were put on signal and interrupt
signatures. At the same, it would first $\beta$-reduce the sequential composition to
\[
  \tmwith{op}{x}{\tmreturn{\tmpromise{1}}}{p}{\tmopout{op'}{(\tmfunano{\tmunit}{\tmawaitgen{p}})}{\tmreturn{2}}}
\]
and then step to the computation
\[
  \tmopout{op'}{(\tmfunano{\tmunit}{\tmawaitgen{p}})}{\tmwith{op}{x}{\tmreturn{\tmpromise{1}}}{p}{\tmreturn{2}}}
\]
which is again ill-typed due to $p$ escaping the interrupt handler's binding scope.

Restricting $V$ to ground values is a simple way that ensures
type-safety~\cite{Ahman:POPL}, but as a result, e.g., one can only send the
arguments needed for the execution of remote function calls but not the
functions themselves. When relaxing the payload restrictions, the type-system
needs to track not only the use of promise-typed variables bound by interrupt
handlers, but as $f$ in the above example shows, also the use of any other variables
that may depend on them. An elegant way of achieving this is a \emph{Fitch-style
modal type system}~\cite{Clouston:FitchStyle}, where the typing context $\Gamma$
can contain (lock) tokens $\ctxlock$, which delimit the extent to which
variables are allowed to be used in terms. In particular, terms can refer only
to variables introduced after the last $\ctxlock$, and to a restricted subset of variables
introduced before it.

Specifically, we extend the grammar of typing contexts to
\begin{align*}
\Gamma
\bnfis& \cdot \bnfor \Gamma, x \of X \bnfor \highlightgray{\,\Gamma, \ctxlock\,}
\end{align*}
and change the typing rule for variables to
\[
   \coopinfer{TyVal-Var}{
    \highlightgray{X~\text{is mobile} \quad \text{or} \quad \ctxlock \not\in \Gamma'}
  }{
    \Gamma, x \of X, \Gamma' \types x : X
  }
\]
This means that we can refer only to variables introduced after the last $\ctxlock$,
or to variables with \emph{mobile types} $A$, defined as an extension of ground types with a \emph{modal (box) type}~$\tybox X$:
\[
A, B
\bnfis \tybase \,\bnfor\! \tyunit \,\bnfor\! \tyempty \,\bnfor\! \typrod{A}{B} \,\bnfor\! \tysum{A}{B} \,\bnfor\! \highlightgray{\tybox{X}}
\]
As with ground types, every mobile type is automatically also a value type, 
including $\tybox{X}$.

Note that $\tybox X$ is a mobile type even if $X$ is not. Equally importantly, neither promise nor function types are mobile 
on their own. When combined with how the context is delimited using $\ctxlock$ in the typing rule \textsc{TyVal-Box} 
given below, these properties of mobile types ensure 
that signal payloads, which are typed with mobile types, cannot use promise-typed variables bound by enveloping 
interrupt handlers. Consequently, it is safe to propagate signals with mobile payloads past any enveloping interrupt handlers 
and eventually to other processes. In its essence, this is similar to the use of modal types in distributed \cite{Murphy:PhDThesis} 
and reactive programming \cite{Krishnaswami:HOFRP,Bahr:RATT}
to classify values that can travel through space and time.

The type $\tybox X$ has a value constructor and a corresponding computation for elimination:
\begin{align*}
  V, W
  \bnfis& \cdots
  \bnfor \highlightgray{\tmbox V}
  \\
  M, N
  \bnfis& \cdots
  \bnfor \highlightgray{\tmunbox{V}{x}{M}}
\end{align*}
with the evident reduction rule given by
\[
  \tmunbox{\tmbox V}{x}{M} \reduces M[V/x]
\]
and with no associated evaluation contexts. 
More importantly, it is the typing rules that ensure one is allowed to box a value only when all the variables 
used in it have mobile types:
\[
    \coopinfer{TyVal-Box}{
    \Gamma, \ctxlock \types V : X
  }{
    \Gamma \types \tmbox V : \tybox X
  }
  \qquad
  \coopinfer{TyComp-Unbox}{
     \Gamma \types V : \tybox X \\
     \Gamma, x \of X \types M : \tycomp{Y}{(\o, \i)}
   }{
     \Gamma \types \tmunbox{V}{x}{M} : \tycomp{Y}{(\o, \i)}
   }
\]
Since this prevents us from constructing boxed values that would refer to a promise-typed variable,
we can safely extend payloads from ground to mobile types. Crucially however, when
constructing a (payload) value of type $\tybox{\tyfun X {\tycomp{Y}{(\o, \i)}}}$,
the boxed function can itself install additional interrupt handlers, it just cannot
refer to the results of any enveloping ones.

\subsection{Dynamic Process Creation}
\label{sec:extensions:dynamic-process-creation}

It turns out that the same Fitch-style modal typing mechanism can be reused to 
extend \lambdaAEff's computations also with \emph{dynamic process creation}:
\[
 M, N
  \bnfis \cdots
  \bnfor \highlightgray{\tmspawn{M}{N}}
\]
Here, $M$ is the new computation to be spawned and $N$ is the continuation of the existing program. 
Operationally, spawned computations propagate out of subcomputations as follows:
\begin{align*}
  \tmlet{x}{(\tmspawn{M_1}{M_2})}{N} &\reduces \tmspawn{M_1}{\tmlet{x}{M_2}{N}} \\[1ex]
  \tmwithrest[S]{op}{x}{r}{s}{M}{V}{p}{\tmspawn{N_1}{N_2}} &\reduces \\[-0.5ex]
      \tmspawn{N_1}{\tmwithrest[S]{op}{&x}{r}{s}{M}{V}{p}{N_2}} \\[1ex]
  \tmopin{op}{V}{\tmspawn{M}{N}} &\reduces \tmspawn{M}{\tmopin{op}{V}{N}}
\end{align*}
This allows the newly spawned computation to reach the top-level of the program, where it 
becomes a new parallel process, as expressed by the following reduction rule for processes:
\[
  \tmrun{(\tmspawn{M}{N})} \reduces \tmpar{\tmrun{M}}{\tmrun{N}}
\]

Analogously to signals, the natural semantics of $\tmkw{spawn}$ is non-blocking. Consequently, 
we also include it in the definition of evaluation contexts of the sequential part of \lambdaAEff:
\[
\E \bnfis \cdots \bnfor \highlightgray{\tmspawn{M}{\E}}
\]
The addition of $\tmkw{spawn}$ also requires us to extend the result forms of computations with
\[
  \coopinfer{}{
    \CompResult {\Psi} {N}
  }{
    \CompResult {\Psi} {\tmspawn{M}{N}}
  }
\]

Finally, the typing rule for $\tmkw{spawn}$ is defined as follows:
\[
   \coopinfer{TyComp-Spawn}{
     \Gamma, \ctxlock \types M : \tycomp{X}{(\o, \i)} \\
     \Gamma \types N : \tycomp{Y}{(\o', \i')}
   }{
     \Gamma \types \tmspawn{M}{N} : \tycomp{Y}{(\o', \i')}
   }
\]
Here, we first use Fitch-style modal typing to ensure that $M$ cannot refer to
promise-typed variables bound by any enveloping interrupt handlers, making it
safe to propagate it outwards, past them. 
This contrasts with other traditional concurrent/parallel languages, such as 
CML~\cite{Reppy:CML}, where no modal typing is needed because spawned processes
do not need to be (operationally) propagated past any binding constructs local to
individual processes. 

What is perhaps even more surprising is that the type of the whole computation depends only
on the type of the continuation $N$, and not on the type of the spawned computation $M$.
This is because spawning $M$ impacts only the execution of enveloping processes 
rather than of $N$ itself. In that sense, one can see the spawning of $M$ as a side-effect of $N$,
and one possibility would be to extend effect annotations to a form $(\o, \i, \s)$, where
\[
  \s = \{ \tycomp{X_1}{(\o_1, \i_1, \s_1)}, \dots, \tycomp{X_n}{(\o_n, \i_n, \s_n)} \}
\]
tracks the types and effects of spawned processes (including any additional processes $\s_i$ they may further spawn).
However, this makes the type system significantly more complicated and brings few additional assurances.
Indeed, at the process level, where spawned processes begin executing, the effect information is already very coarse since 
we need to account for actions of incoming interrupts.
For that reason, we opt for a simpler, yet still type-sound solution (see \autoref{thm:procpreservation:extended}), 
and instead extend the process type reduction relation with additional rules that allow spontaneously 
adding an arbitrary process type in parallel:
\[
  \coopinfer{}{
  }{
    X \att {(\o , \i)} \tyreduces
    \typar{(X \att {(\o , \i)})}{(Y \att {(\o' , \i')})}
  }
  \qquad
  \coopinfer{}{
  }{
    Y \att {(\o' , \i')} \tyreduces
    \typar{(X \att {(\o , \i)})}{(Y \att {(\o' , \i')})}
  }
\]

\subsection{Type Safety}
\label{sec:type-safety}

With all the higher-order extensions in place, we now prove type safety
for the full, final version of \lambdaAEff---first for computations and then for parallel 
processes. 

While for brevity we do not repeat them here, we note that the finality results 
we proved about the result forms in \autoref{lem:results-are-final} and 
\ref{lem:results-are-final:processes} also hold for this extended version of \lambdaAEff.

\subsubsection{Computations}

We recall that as standard, we split type safety into proofs of progress and preservation,
with the former stated as follows (see also the discussion in \autoref{sec:basic-calculus:type-safety}):

\begin{thm}[Progress for computations]
\label{thm:progress:extended}
Given a well-typed computation
\[
  p_1 \of \typromise {X_1}, \ldots, p_n \of \typromise {X_n} \types M : \tycomp{Y}{(\o,\i)}
\]
then either
\begin{enumerate}[(a)]
  \item there exists a computation $N$, such that $M \reduces N$, or
  \item the computation $M$ is in a result form, i.e., we have $\CompResult {\{p_1, \ldots, p_n\}} {M}$.
\end{enumerate}
\end{thm}

\begin{proof}
The proof is standard and proceeds by induction 
on the derivation of  $\Gamma \types M : \tycomp{Y}{(\o,\i)}$. For instance, 
if the derivation ends with a typing rule for function application or pattern-matching, 
we use an auxiliary canonical forms lemma to show that the value involved 
is either a function abstraction or in constructor form---thus $M$ can $\beta$-reduce and we prove (a).
Here we crucially rely on the context $\Gamma$ having the specific 
form $p_1 \of \typromise {X_1}, \ldots, p_n \of \typromise {X_n}$, 
with all the variables assigned promise types.
If the derivation ends with \textsc{TyComp-Await}, we 
use a canonical forms lemma to show that the promise value is either a variable in
$\Gamma$, in which case we prove (b), or in constructor form, in which case we prove (a).
If the derivation ends with a typing rule for any of the terms figuring in the 
evaluation contexts $\E$, we proceed based on the outcome of using the induction hypothesis 
on the corresponding continuation.
\end{proof}
    
The results that we present in this section (and that we summarised in \autoref{sec:basic-calculus:type-safety}) 
use standard \emph{substitution  
lemmas}. For instance, given $\Gamma, x \of X , \Gamma' \types M : \tycomp{Y}{(\o,\i)}$
and $\Gamma \types V : X$, then we can show that $\Gamma, \Gamma' \types M[V/x] : \tycomp{Y}{(\o,\i)}$.
In addition, we use standard \emph{typing inversion lemmas}. For example, given 
a computation $\Gamma \types \tmopin{op}{V}{M} : \tycomp{X}{(\o,\i)}$, then we can show that 
$\Gamma \types V : A_\op$ and $\Gamma \types M : \tycomp{X}{\opincomp {op} (\o',\i')}$, 
such that $\opincomp {op} (\o',\i') \order {O \times I} (\o,\i)$.
Furthermore, we use \emph{strengthening lemmas} for promise-typed variables, 
such as if we have $\Gamma, p : \typromise{X}, \Gamma' \types V : Y$, and if $\Gamma'$ 
contains $\ctxlock$ or if $Y$ is a mobile type, then also $\Gamma, \Gamma' \types V : Y$.

We also note that the action $\opincomp {op} {(-)}$ has various useful properties that we use below
(where we write $\pi_1$ and $\pi_2$ for the projections associated with the Cartesian product $O \times I$):

\begin{lem}
\label{lem:action}
\mbox{}
\begin{enumerate}
\item $\o \order O \pi_1\, (\opincomp {op} {(\o,\i)})$
\item If $\i\, (\op) = (\o',\i')$, then $(\o',\i') \order {O \times I} \opincomp {op} {(\o,\i)}$
\item If $\op \neq \op'$ and $(\o',\i') \order {O \times I} \i\, (\op')$, then $(\o',\i') \order {O \times I} (\pi_2\, (\opincomp {op} {(\o,\i)}))\, (\op')$
\end{enumerate}
\end{lem}
\noindent
Next, as the proof of type preservation proceeds by induction on reduction steps, 
we find it useful to define an auxiliary \emph{typing judgement for evaluation contexts}, 
written
\[
  \Gamma \types\!\![\, \Gamma' \,\vert\, \tycomp{X}{(\o,\i)} \,]~ \E : \tycomp{Y}{(\o',\i')}
\]
which we then use to prove the evaluation context rule case of the preservation proof.
In this judgement, $\Gamma'$ is the context of variables bound by the interrupt handlers in $\E$, and 
$\tycomp{X}{(\o,\i)}$ is the type of the hole $[~]$. This judgement is defined using rules 
similar to those for typing computations, including subtyping, e.g., for interrupt handlers we have
the following rule:
\begin{mathpar}
  \coopinfer{}{
    (\o'',\i'') \order{O \times I} \i'\, (\op) \\
    \begin{minipage}[c][5ex]{0.8\textwidth}
      \centering
      $\Gamma, x \of A_\op, r \of \tyfun{S}{\tycomp{\typromise Y}{\big(\emptyset, \{ \op \mapsto ({\o''} , {\i''}) \}\big)}}, s \of S \types M : \tycomp{\typromise Y}{(\o'',\i'')}$ 
    \end{minipage}\\
    \Gamma \types V : S \\
    \Gamma, p \of \langle Y \rangle \types\!\![\, \Gamma' \,\vert\, \tycomp{X}{(\o,\i)} \,]~ \E : \tycomp{Z}{(\o',\i')}
  }{
    \Gamma \types\!\![\, p \of \langle Y \rangle, \Gamma' \,\vert\, \tycomp{X}{(\o,\i)} \,]~ \tmwithrest[S]{op}{x}{r}{s}{M}{V}{p}{\E} : \tycomp{Z}{(\o',\i')}
  }
\end{mathpar}
The typing of evaluation contexts is straightforwardly related to that of computations:

\begin{lem}
\label{lem:eval-ctx-typing}
\begin{multline*}
  \big(\Gamma \types \E[M] : \tycomp{Y}{(\o',\i')}\big) 
  \,\Longleftrightarrow\,  \\
\exists\, \Gamma', X, \o, \i .~ 
\big(\Gamma \types\!\![\, \Gamma' \,\vert\, \tycomp{X}{(\o,\i)} \,]~ \E : \tycomp{Y}{(\o',\i')}\big)
~\wedge~
\big(\Gamma,\Gamma' \types M : \tycomp{X}{(\o,\i)}\big)
\end{multline*}
\end{lem}

We are now ready to prove the type preservation theorem for the sequential part of \lambdaAEff.

\begin{thm}[Preservation for computations]
\label{thm:preservation:extended}
Given a computation $\Gamma \types M : \tycomp{X}{(\o,\i)}$, such that $M$
can reduce as $M \reduces N$, then we have $\Gamma \types N : \tycomp{X}{(\o,\i)}$.
\end{thm}

\begin{proof}
The proof is standard and proceeds by induction on the derivation of  
$M \reduces N$, using typing inversion lemmas based on 
the structure forced on $M$ by the last rule used in $M \reduces N$.

There are four cases of interest in this proof. The first two concern
the interaction of interrupts and interrupt handlers.
On the one hand, if the derivation of $\reduces$ ends with 
\[
\tmopin{op}{V}{\tmwithrest{op}{x}{r}{s}{M}{W}{p}{N}} \reduces \tmlet{p}{M[V/x, R/r, W/s]}{\tmopin{op}{V}{N}}
\]
where $R$ is a function that reinstalls the interrupt handler,
then in order to type the right-hand side of this rule, we use subtyping with  
\srefcase{Lemma}{lem:action}{2} to show that $M$'s effect information is 
included in that of $\tmopin{op}{V}{N}$, i.e., in $\opincomp {op} {(\o , \i)}$. On the other hand, given the rule 
\[
\begin{array}{l}
\tmopin{op'}{V}{\tmwithrest[S]{op}{x}{r}{s}{M}{W}{p}{N}} \reduces 
\\[0.5ex]
\hfill \tmwithrest[S]{op}{x}{r}{s}{M}{W}{p}{\tmopin{op'}{V}{N}} 
\quad
\\[1ex]
\hspace{13cm} {\color{rulenameColor}(\op \neq \op')}
\end{array}
\]
then in order to type the right-hand side, we use subtyping with  
\srefcase{Lemma}{lem:action}{3}, so as to show that 
after acting on $(\o,\i)$ with $\op'$, $\op$ remains mapped to $M$'s effect information.

The third case of interest concerns the commutativity of signals with interrupt handlers:
\[
\begin{array}{l}
\tmwithrest[S]{op}{x}{r}{s}{M}{V}{p}{\tmopout{op'}{W}{N}} \reduces 
\\[0.5ex]
\hspace{6.5cm} \tmopout{op'}{W}{\tmwithrest[S]{op}{x}{r}{s}{M}{V}{p}{N}}
\end{array}
\]
where in order to type the signal's payload $W$ in the right-hand side of this rule, 
it is crucial that the promise-typed variable $p$ cannot appear in $W$---this is ensured by 
our modal type system that restricts the signatures $\op : A_\op$ to mobile types. As a result, 
we can strengthen the typing context of $W$ by removing the promise-typed variable $p$ 
from it. We also use an analogous context strengthening argument for $N_1$ when given the 
other commutativity rule 
\begin{align*}
  \tmwithrest[S]{op}{x}{r}{s}{M}{V}{p}{\tmspawn{N_1}{N_2}} &\reduces \\[-0.5ex]
      \tmspawn{N_1}{\tmwithrest[S]{op}{&x}{r}{s}{M}{V}{p}{N_2}}
\end{align*}

\noindent Finally, in the evaluation context case, we use the induction hypothesis with
\sref{Lemma}{lem:eval-ctx-typing}. \end{proof} 

Interestingly, the proof of \autoref{thm:preservation:extended} tells us that if 
one were to consider a variant of \lambdaAEff~in which the 
\textsc{TyComp-Subsume} rule appeared as an explicit coercion term $\tmkw{coerce}_{(\o,\i) \order {O \times I} (\o',\i')}\, M$,
which is the style we use in our Agda formalisation~\cite{ahman:AeffAgda}, then 
the right-hand sides of the two interrupt propagation rules highlighted in the above proof 
would also need to involve such coercions, corresponding to the two uses of \sref{Lemma}{lem:action}. 
This however means that other computations involved in these reduction rules would also 
need to be type-annotated accordingly, so as to determine the data to be used in these coercions.

\subsubsection{Processes}

For the parallel part of \lambdaAEff, we again first prove the progress theorem.

\begin{thm}[Progress for processes]
  \label{thm:procprogress:extended}
  Given a well-typed process $\types P : \tyC$, then either
  \begin{enumerate}[(a)]
    \item there exists a process $Q$, such that $P \reduces Q$, or
    \item the process $P$ is already in a (top-level) result form, i.e., we have $\ProcResult {P}$.
  \end{enumerate}
\end{thm}

\begin{proof}
The proof is unsurprising and proceeds by induction on the derivation of $\types P : \tyC$. 
In the base case, when the derivation ends with the \textsc{TyProc-Run} rule 
and $P = \tmrun {M}$, we use \sref{Theorem}{thm:progress:extended}. In the other cases, 
we simply use the induction hypothesis.
\end{proof}

To prove preservation, we first focus on properties of the process type reduction $\tyC \tyreduces \tyD$.

\begin{lem}
\label{lem:type-reduction} \mbox{}
\begin{enumerate}
\item Process types can remain unreduced, i.e., 
  $\tyC \tyreduces \tyC$, for any process type $\tyC$.
\item Process types can reduce by being acted upon, i.e., 
  $\tyC \tyreduces \opincomp {op} \tyC$, for any $\op$ and $\tyC$.
\item Process types can reduce under enveloping actions, i.e.,  
  $\tyC \tyreduces \tyD$ implies $\opincomp {op} \tyC \tyreduces \opincomp {op} \tyD$.
\item Process type reduction can introduce signals but does not erase them, i.e., 
  $\tyC \tyreduces \tyD$ implies $\mathsf{signals\text{-}of} (\tyC) \order O \mathsf{signals\text{-}of} (\tyD)$.
\end{enumerate}
\end{lem}
\noindent
The interesting case in the proof of \srefcase{Lemma}{lem:type-reduction}{3} is when the enveloped reduction
$\tyC \tyreduces \tyD$ introduces an interrupt $\op'$ under some sequence of interrupts $\opsym{ops}$, as follows:
\[
  X \att \opincompp {ops} {(\o , \i)} \tyreduces X \att \opincompp {ops} {(\opincomp {op'} {(\o , \i)})}
\]
To prove this case, we simply prepend $\op$ to the list $\opsym{ops}$ and reapply the same rule, as 
\[
\begin{array}{c c c}
\opincomp {op} {(X \att \opincompp {ops} {(\o , \i)})}
&&
\opincomp {op} {(X \att \opincompp {ops} {(\opincomp {op'} {(\o , \i)})})}
\\[1ex]
\rotatebox{90}{$=$}
&
&
\rotatebox{90}{$=$}
\\[0.3ex]
X \att \opincompp {(op :: ops)} {(\o , \i)}
&
\tyreduces
&
X \att \opincompp {(op :: ops)} {(\opincomp {op'} {(\o , \i)})}
\end{array}
\]
Observe that defining $\tyC \tyreduces \tyD$ using 
a simpler basic rule $X \att {(\o , \i)} \tyreduces X \att {(\opincomp {op'} {(\o , \i)})}$
would not have been sufficient to prove this case, i.e., 
$\opincomp {op} {(X \att {(\o , \i)})} \tyreduces \opincomp {op} {(X \att {(\opincomp {op'} {(\o , \i)})})}$.

For the proof of \srefcase{Lemma}{lem:type-reduction}{4}, we generalise 
\srefcase{Lemma}{lem:action}{1} to lists of actions.

\begin{lem}
\label{lem:signal-inclusion-lists-of-interrupts}
$\pi_1\, (\opincompp {ops} {(\o,\i)}) \order O \pi_1\, (\opincompp {ops} {(\opincomp {op} {(\o,\i)})})$
\end{lem}
\noindent
As with computations, it is useful 
to define a separate \emph{typing judgement for evaluation contexts}, this time written 
$\Gamma \types\!\![\, \tyC \,]~ \F : \tyD$, together with an 
analogue of \sref{Lemma}{lem:eval-ctx-typing}, which we omit here. Instead, we 
observe that this typing judgement preserves process type reduction.

\begin{lem}
\label{lem:hoisting-and-evaluation-context-types}
Given $\Gamma \types\!\![\, \tyC \,]~ \F : \tyD$ and $\tyC \tyreduces \tyC'$, then there exists $\tyD'$ with 
$\tyD \tyreduces \tyD'$, and we have $\Gamma \types\!\![\, \tyC' \,]~ \F : \tyD'$.
\end{lem}

Process types also satisfy an analogue of \srefcase{Lemma}{lem:action}{1}, which shows that 
the action $\opincomp {op} (-)$ of interrupts on process types does not erase any already specified outgoing signals.

\begin{lem}
\label{lem:signals-of-interrupt-action}
For any $\tyC$ and $\op$, we have $\mathsf{signals\text{-}of}(\tyC) \order O \mathsf{signals\text{-}of}(\opincomp{op}{\tyC})$.
\end{lem}

Finally, using the results above, we prove type preservation for the parallel part of \lambdaAEff.

\begin{thm}[Preservation for processes]
  \label{thm:procpreservation:extended}
  Given a well-typed process $\Gamma \types P : \tyC$, such that $P$ can reduce as 
  $P \reduces Q$, then there exists a process type $\tyD$, such 
  that the process type $\tyC$ can reduce as $\tyC \tyreduces \tyD$, 
  and we can type the resulting process as $\Gamma \types Q : \tyD$.
\end{thm}

\begin{proof}
The proof proceeds by induction on the derivation of  
$P \reduces Q$, using auxiliary typing inversion lemmas depending on 
the structure forced upon $P$ by the last rule used in $P \reduces Q$.

For most of the cases, we can pick $\tyD$ to be $\tyC$ and use
\srefcase{Lemma}{lem:type-reduction}{1}. For process creation, i.e., for the
interaction of $\tmkw{run}$ and $\tmkw{spawn}$, we define $D$ by composing $C$
in parallel with the spawned process's type, and build $\tyC \tyreduces \tyD$
using the new type reduction rule
\[
  \coopinfer{}{
  }{
    Y \att {(\o' , \i')} \tyreduces
    \typar{(X \att {(\o , \i)})}{(Y \att {(\o' , \i')})}
  }
\]
that we introduced in \autoref{sec:extensions:dynamic-process-creation}.

For the broadcast rules, we define $\tyD$ by introducing the corresponding 
interrupt, and build $\tyC \tyreduces \tyD$ using the parallel composition 
rule together with \srefcase{Lemma}{lem:type-reduction}{2}.

For the evaluation context rule, we use \sref{Lemma}{lem:hoisting-and-evaluation-context-types}
in combination with the induction hypothesis.
Finally, in order to discharge effect annotations-related side-conditions 
when commuting incoming interrupts with outgoing signals, 
we use \sref{Lemma}{lem:signals-of-interrupt-action}.
\end{proof}


\section{Asynchronous Effects in Action}
\label{sec:applications}

We now show examples of the kinds of programs one can write in \lambdaAEff.
Similarly to \autoref{sec:overview:runningexample}, we again allow ourselves access to 
mutable references as a matter of convenience. We use these references only for 
(function call) counters and for communicating data between different parts of a program---passing data between 
subsequent reinstalls of the same interrupt handler is dealt with using the stateful reinstallable 
interrupt handlers introduced in \autoref{sec:extensions:stateful-reinstallable-interrupt-handlers}. 

In addition to the generic versions of constructs defined in \autoref{sec:overview:runningexample}, 
we further use
\begin{align*}
    \tmspawngen{M} &~\defeq~ \tmspawn{M}{\tmreturn \tmunit} \\
    \tmunboxgen{V} &~\defeq~ \tmunbox{V}{x}{\tmreturn x}
\end{align*}

\subsection{Guarded Interrupt Handlers}
\label{sec:applications:guarded-handlers}

Before diving into the examples, we note that we often want the 
triggering of interrupt handlers to be 
conditioned on not only the names of interrupts, but also on the payloads that they carry.
In order to express such more fine-grained interrupt handler triggering behaviour, we shall use a
\emph{guarded interrupt handler}: 
\begin{aeffbox}
\begin{lstlisting}
promise (op x r s when guard |-> comp) @ v
\end{lstlisting}
\end{aeffbox}
which is simply a syntactic sugar for the following stateful interrupt handler that reinstalls 
itself until the boolean \ls$guard$ becomes true, in which case it executes the handler code \ls$comp$:
\begin{aeffbox}
\begin{lstlisting}
promise (op x r s |-> if guard then comp else r s) @ v
\end{lstlisting}
\end{aeffbox}
where \ls$x$ and \ls$s$ are bound both in \ls$guard$ and 
\ls$comp$. This means that the handler triggering can be conditioned both on the payload 
and state values. Meanwhile, \ls$r$ is bound only in \ls$comp$.
Also, note that regardless whether \ls$guard$ is true, every interrupt gets propagated into \ls$cont$.

As guarded interrupt handlers repeatedly reinstall themselves, they get assigned 
recursive effect annotations, as discussed in \autoref{sec:extensions:reinstallable-interrupt-handlers}.
For example, if \ls$comp$ has type $\tycomp{\typromise X}{(\o,\i)}$, then the corresponding
guarded interrupt handler gets assigned the type $\tycomp{\typromise X}{(\emptyset, \i_h)}$, where
\[
  \i_h = \{ \op \mapsto (\o, \i \sqcup \{ \op \mapsto (\o, \i \sqcup \{ \op \mapsto (\o, \cdots) \}) \}) \}
\]
is the least fixed point of the continuous map $\i' \mapsto \{ \op \mapsto (\o, \i \sqcup \i') \} : I \to I$.
As such, the type $\tycomp{\typromise X}{(\emptyset, \i_h)}$ specifies that the installation of the 
guarded interrupt handler does not issue any signals by itself, and that the arrival of any 
$\opsym{op}$ interrupt causes either the effects $(\o,\i)$ of \ls$comp$ to happen, 
or the interrupt handler to be reinstalled. Observe that as a consequence, some of the 
recursive encoding leaks via $\i_h$ into the type of guarded interrupt handlers.

Similarly to reinstallable interrupt handlers, we write \ls$promise (op x r when guard |-> comp)$ 
when the state associated with the guarded interrupt handler is trivial and can be omitted.

\subsection{Pre-Emptive Multi-Threading}
\label{sec:applications:multithreading}

Multi-threading remains one of the most exciting applications of algebraic effects, with the possibility of modularly
and user-definably expressing 
many evaluation strategies being the main reason for the extension of \pl{OCaml} with effect handlers~\cite{OCaml:ReleaseNotes}.
These evaluation strategies are however \emph{cooperative} in nature, where each thread needs to explicitly yield back 
control, stalling other threads until then. 

While it is possible to simulate \emph{pre-emptive multi-threading} within the usual treatment of algebraic effects, 
it requires a low-level access to the specific runtime environment, so as to inject 
yields into the currently running computation.
In contrast, implementing pre-emptive multi-threading in \lambdaAEff~is quite straightforward, and importantly, 
possible within the language itself---the injections into the running computation 
take the form of incoming interrupts.

For the purpose of modelling pre-emptive multi-threading, let us consider two interrupts, 
$\opsym{stop} : \tyunit$ and $\opsym{go} : \tyunit$, that communicate to a thread whether 
to \emph{pause} or \emph{resume} execution. For example, these interrupts might originate 
from a timer process being run in parallel.

At the core of our implementation of pre-emptive multi-threading is
the computation term \ls$waitForStop ()$ that is defined as the following reinstallable interrupt handler:
\begin{aeffbox}
\begin{lstlisting}
let waitForStop () = 
    promise (stop _ r |->
        let p = promise (go _ |-> return <<()>>) in
        await p;
        r ()
    )
\end{lstlisting}
\end{aeffbox}
which first installs an interrupt handler for $\opsym{stop}$, letting subsequent computations run their course. Once 
a $\opsym{stop}$ interrupt arrives, the interrupt handler for it is triggered and the next one for $\opsym{go}$ is 
installed. In contrast to the interrupt handler for $\opsym{stop}$, we now start awaiting 
the promise \ls$p$. This means that any subsequent computations are blocked until a $\opsym{go}$ interrupt 
is received, after which we reinstall the interrupt handler for $\opsym{stop}$ and repeat the cycle.

To \emph{initiate the pre-emptive behaviour} for some computation \ls{comp}, we run the program
\begin{aeffbox}
\begin{lstlisting}
waitForStop (); comp
\end{lstlisting}
\end{aeffbox}
The algebraicity reduction rules for interrupt handlers ensure that they propagate out of \ls{waitForStop} 
and eventually encompass the entire composite computation, including \ls{comp}.
It is important to note that in contrast to the usual effect handlers based encodings of multi-threading, \ls$waitForStop$ does 
not need any access to a thunk \lstinline{fun () |-> comp} representing the threaded computation. 
In particular, the computation \ls$comp$ that we want to pre-empt can be completely unaware of the multi-threaded behaviour, 
both in its definition and type.

This approach can be easily extended to multiple threads, by using interrupts' payloads to communicate thread IDs. To this end, 
we can consider interrupts $\opsym{stop} : \tyint$ and $\opsym{go} : \tyint$, and use guarded interrupt handlers to define  
a thread ID sensitive version of \ls$waitForStop$:
\begin{aeffbox}
\begin{lstlisting}
let waitForStop threadID =
    promise (stop threadID' r when threadID = threadID' |->
        let p = promise (go threadID' when threadID = threadID' |-> return <<()>>) in
        await p;
        r ()
    )
\end{lstlisting}
\end{aeffbox}
with the triggering of the interrupt handlers being conditional on the received thread IDs.

\subsection{Remote Function Calls}
\label{sec:applications:remotecall}

One of the main uses of asynchronous computation is to offload the execution of 
\emph{long-running functions} to remote processes. Below we show how to implement this
in \lambdaAEff~in a way that requires minimal cooperation from the remote process.

For a simpler exposition, we assume a fixed (mobile) result type $A$ shared by all functions 
that we may wish to execute remotely. For communicating a function to be executed to 
the remote process, we assume a signal $\opsym{call} : \tybox{\tyunit \to \tycomp{\tyunit}{(\o, \i)}}$. 
Finally, for communicating the remote function call's $A$-typed result back to the caller, we assume 
a signal $\opsym{result} : \typrod{A}{\tysym{int}}$.

The caller then calls functions \ls$f$ remotely through a wrapper function, \ls$remoteCall$,
which issues a \ls$call$ signal, installs a handler for a \ls$result$ 
interrupt, and returns a thunk that can be used to block the caller program's 
execution and await the remote function's result:
\begin{aeffbox}
\begin{lstlisting}
let remoteCall f =
    let callNo = !callCounter in callCounter := !callCounter + 1;
    let task = [| fun _ -> let g = unbox f in
                                let res = g () in 
                                send result (res, callNo) |] 
    in
    send call task;
    let resultPromise = promise (result (y, callNo') when callNo = callNo' |-> return <<y>>) in
    let awaitResult () = await resultPromise in
    return awaitResult
\end{lstlisting}
\end{aeffbox}

Observe that the function \ls$f$ is not sent directly in the payload of the \ls$call$ signal
to the remote process. Instead, \ls$call$'s payload combines the task of executing 
\ls$f$ with issuing a \ls$result$ signal with the function's result. This ensures 
that the result is always sent back to the caller, and the callee process can have a very 
simple implementation (see below). In addition, this combination explains why the signature
of \ls$call$ does not mention $A$. Further, we note that in order to ensure that the payload is a 
boxed value, as required by \ls$call$'s signature, the function \ls$f$ has to be 
passed to \ls$remoteCall$ in a boxed form (notice the use of \ls$unbox$ 
in \ls$task$).

To avoid the results of earlier remote function calls from fulfilling the promises of later ones, 
we assign to each call a unique identifier, which we implement using a counter local to the caller 
process. The identifier is passed together with the result and a guarded interrupt handler is used 
to ensure that only the result of the correct call is awaited. Note that this policy is again enforced 
by the caller and does not require any cooperation from the callee.

We also note that the effect annotation $(\o, \i)$ in \ls$call$'s signature can 
be used to  limit the effects the caller may trigger in the callee process---it also
influences the effects of functions \ls$f$ that one can call the 
\ls$remoteCall$ wrapper with. In order to be able to communicate 
the remote function's result back to the caller in \ls$task$, $\o$ should include at 
least the \ls$result$ signal.

For instance, one may then call remote functions in their code as follows:
\begin{aeffbox}
\begin{lstlisting}
let subtally = remoteCall [| fun () -> query "SELECT count(col) FROM table WHERE cond" |] in
let tally = remoteCall [| fun () -> query "SELECT count(col) FROM table" |] in
printf "Percentage: %d" (100 * subtally () / tally ())
\end{lstlisting}
\end{aeffbox}

In the \emph{callee process}, we simply install an interrupt handler that spawns a new process for 
executing the received function and then immediately recursively reinstalls itself, as follows:
\begin{aeffbox}
\begin{lstlisting}
promise (call boxedTask r |-> 
    spawn (let f = unbox boxedTask in 
              f ()); 
    r ()
)
\end{lstlisting}
\end{aeffbox}
Observe that as the payload of the \ls$call$ interrupt is received in a boxed form, 
it has to be unboxed before we are able to execute its underlying function. Here it is 
important that this unboxing happens inside the argument of \ls$spawn$ and not
before the call to \ls$spawn$. Namely, as the argument of 
\ls$spawn$ has to be mobile, its context is delimited by $\ctxlock$ (as discussed 
in \autoref{sec:extensions:dynamic-process-creation}) and therefore it can only refer 
to variables with mobile types bound outside of it, and whereas the type of 
\ls$boxedTask$ is mobile, the type of the underlying function is not.

This example can be naturally generalised to allow the remotely executed functions to take
non-unit arguments: on the one hand, simply by passing arguments to the callee using the \ls$remoteCall$ 
wrapper function, or on the other hand, by defining separate wrapper functions for communicating the 
function and a particular call's arguments to the callee one at a time. We omit this generalisation here, but 
refer the reader to our original work~\cite{Ahman:POPL} for an example of remote function 
calls being triggered by passing a particular call's arguments to the callee. 
However, it is important to highlight that whereas in our original work we were limited to only 
sending arguments to a fixed remote function, the modal boxed types adopted in this paper  
would enable the caller to dynamically also pass functions to the callee.

Unlike effect handlers, our interrupt handlers have very limited control over the execution of 
their continuation. Regardless, we can still simulate \emph{cancellations of asynchronous 
computations} using the ideas behind our implementation of pre-emptive multithreading that we 
described in \autoref{sec:applications:multithreading}. Specifically, we modify the \ls$remoteCall$ wrapper
function so that it returns an additional \emph{cancellation thunk}, which can be used to cancel the computation:
\begin{aeffbox}
\begin{lstlisting}
let remoteCancellableCall f =
    let callNo = !callCounter in callCounter := !callCounter + 1;
    let task = [| fun _ -> waitForCancel callNo; 
                                let g = unbox f in 
                                let res = g () in 
                                send result (res, callNo) |]
    in
    send call task;
    let resultPromise = promise (result (y, callNo') when callNo = callNo' |-> return <<y>>) in
    let awaitResult () = await resultPromise in
    let cancelCall () = send cancel callNo in
    return (awaitResult, cancelCall)
\end{lstlisting}
\end{aeffbox}
and where the function used to implement cancellations in the payload of \ls$call$ is defined as
\begin{aeffbox}
\begin{lstlisting}
let waitForCancel callNo =
  promise (cancel callNo' when callNo = callNo' ->
    let p = promise (impossible _ -> return <<()>>) in
    await p;
    return <<()>>
  )
\end{lstlisting}
\end{aeffbox}
for which we assume two additional signals: $\opsym{cancel} : \tyint$ and $\opsym{impossible} : \tyempty$.

The callee code remains unchanged. Running each remote call in a separate process ensures
that each \ls$cancel$ interrupt affects only one remote function call. In our original
work~\cite{Ahman:POPL}, where all remote calls were executed in a single process,
we additionally needed an auxiliary reinvoker process to continue executing the non-cancelled 
remote function calls.

Finally, we observe that the cancelled computation is only \emph{perpetually stalled} 
(indefinitely awaiting the \ls$impossible$ interrupt, which can never be propagated to the process due to 
its $\tyempty$-typed signature) but not discarded completely, leading to a memory leak. 
We conjecture that extending \lambdaAEff~with interrupts and interrupt handlers that have greater 
control over their continuations could lead to a more efficient, memory leak-free code for the callee site.

\subsection{Runners of Algebraic Effects}
\label{sec:applications:runners}

Next, we show how to use \lambdaAEff~to implement a parallel variant 
of \emph{runners of algebraic effects} \cite{Ahman:Runners}. These are a 
natural mathematical model and programming abstraction for resource management based on 
algebraic effects, and correspond to effect handlers that resume continuations (at most) 
once in a tail call position.

In a nutshell, for a signature of operation 
symbols $\op : A_\op \to B_\op$, a \emph{runner} $\mathcal{R}$ comprises a family of stateful functions 
$\overline{\op}_{\mathcal{R}} : A_\op \times R \to B_\op \times R$, 
called \emph{co-operations}, where $R$ is the type of \emph{resources} that the particular runner manipulates.
In the more general setting, the co-operations also model other, external 
effects, such as native calls to the operating system, and can furthermore raise  
exceptions---all of which we shall gloss over here.

Given a runner $\mathcal{R}$, the programmer is provided with a construct
\[
\tmkw{using}~\mathcal{R}~\tmkw{@}~V_{\text{init}}~\tmkw{run}~M~\tmkw{finally}~\{ \tmreturn x ~\tmkw{@}~ r_{\text{fin}} \mapsto N  \}
\]
which runs $M$ using $\mathcal{R}$, with resources initially set to $V_{\text{init}}$; and 
finalises the return value (bound to $x$) and final resources (bound to $r_{\text{fin}}$) using the computation $N$, e.g., 
ensuring that all file handles get closed. This is a form of effect handling: it executes $M$ by invoking 
co-operations in place of operation calls, while doing resource-passing under the hood. 
Below we show by means of examples how one can use \lambdaAEff~to naturally separate $\mathcal{R}$ and $M$ 
into different processes.
For simplicity, we omit the initialisation and finalisation phases.

For our first example, let us consider a runner that implements a \emph{pseudorandom number generator} 
by providing a co-operation for ${\opsym{random} : \tyunit \to \tyint}$, which we can implement as 
\begin{aeffbox}
\begin{lstlisting}
let linearCongruenceGeneratorRunner modulus a c initialSeed =
    promise (randomReq callNo r seed |->
        let seed' = (a * seed + c) mod modulus in
        send randomRes (seed, callNo); 
        r seed'
    ) @ initialSeed
\end{lstlisting}
\end{aeffbox}
It is given by a recursive interrupt handler, which listens for $\opsym{randomReq} : \tysym{int}$ requests 
issued by clients, and itself issues $\opsym{randomRes} : \typrod{\tysym{int}}{\tysym{int}}$ responses. 
The resource that this runner manages is the seed, which it passes between subsequent co-operation 
calls using the state-passing features provided by our reinstallable interrupt handlers. The seed is originally set 
to \ls$initialSeed$, recalculated during each execution of the interrupt handler, and passed to the next 
co-operation call by reinstalling the interrupt handler with the updated seed value.

In the client, we implement operation calls \ls$random ()$ as discussed in \autoref{sec:overview:signals}, 
by decoupling them into signals and interrupt handling. We use guarded interrupt handlers and call identifiers to 
avoid a response to one operation call fulfilling the promises of other ones. 

\begin{aeffbox}
\begin{lstlisting}
let random () =
    let callNo = !callCounter in callCounter := callNo + 1;
    send randomReq callNo;
    let p = promise (randomRes (n, callNo') when callNo = callNo' |-> return <<n mod 10>>) in
    await p
\end{lstlisting}
\end{aeffbox}

As a second example of runners, we show that this parallel approach to runners naturally 
extends to multiple co-operations. Specifically, we implement a \emph{runner for a heap}, 
which provides co-operations for the following three operation symbols:
\[
    \opsym{alloc} : \tysym{int} \to \tysym{loc} \qquad
    \opsym{lookup} : \tysym{loc} \to \tysym{int} \qquad
    \opsym{update} : \tysym{loc} \times \tysym{int} \to \tyunit
\]
We represent the co-operations using a signal/interrupt pair $(\opsym{opReq},\opsym{opRes})$ with 
respective payload types $\typrod{\tysym{payloadReq}}{\tysym{int}}$ and 
$\typrod{\tysym{payloadRes}}{\tysym{int}}$, tagged with call identifiers, and where
\begin{aeffbox}
\begin{lstlisting}
type payloadReq = AllocReq of int | LookupReq of loc | UpdateReq of loc * int
type payloadRes = AllocRes of loc | LookupRes of int | UpdateRes of unit
\end{lstlisting}
\end{aeffbox}

\noindent
The resulting runner is implemented by pattern-matching on the payload value as follows:
\begin{aeffbox}
\begin{lstlisting}
let heapRunner initialHeap =
    promise (opReq (payloadReq, callNo) r heap |->
        let heap', payloadRes =
            match payloadReq with
            | AllocReq v |-> 
                  let heap', l = allocHeap heap v in 
                  return (heap', AllocRes l)
            | LookupReq l |-> 
                  let v = lookupHeap heap l in 
                  return (heap, LookupRes v)
            | UpdateReq (l, v) |-> 
                  let heap' = updateHeap heap l v in 
                  return (heap', UpdateRes ())
        in
        send opRes (payloadRes, callNo); 
        r heap'
    ) @ initialHeap
\end{lstlisting}
\end{aeffbox}
The resource that this runner manages is the heap---it is initially set to \ls$initialHeap$, and then 
updated and passed between subsequent co-operation calls analogously to the seed in the previous example. 
On the client side, the operation calls for allocation, lookup, and update are also implemented similarly 
to how \ls$random ()$ was defined in the previous example.

Finally, we note that we could have instead used three signal/interrupt pairs and split \ls$heapRunner$ 
into three distinct reinstallable interrupt handlers, one for each of the three co-operations. However, then we would 
not have been able to use the state-passing provided by our interrupt handlers and we would 
have had to store the heap in the memory instead.

\subsection{Non-Blocking Post-Processing of Promised Values}
\label{sec:applications:chaining}

As discussed in \autoref{sect:overview:promising}, interrupt handlers differ
from ordinary operation calls by allowing user-side post-processing of received data in the handler code. 
In this example, we show that \lambdaAEff~is flexible enough to modularly perform \emph{further 
non-blocking post-processing} of this data anywhere in a program.

For instance, let us assume we are writing a program that contains an interrupt handler (for some $\op$)
that promises to return us a list of integers. Let us further assume that at some later point 
in the program we decide that we want to further process this list 
if and when it becomes available, 
e.g., by using some of its elements to issue an outgoing signal. 
Of course, we could do this by going back and changing the definition of the original interrupt handler, 
but this would not be very modular; 
nor do we want to block the entire program's execution (using \ls$await$) until the $\op$ interrupt
arrives and the concrete list becomes available.

Instead, we can define a generic combinator for \emph{non-blocking post-processing} of  
promises
\begin{aeffbox}
\begin{lstlisting}
process$_{\op}$ p with (<<x>> |-> comp)
\end{lstlisting}
\end{aeffbox}
that takes an earlier made promise \ls$p$ (which we assume originates 
from handling the specified interrupt $\op$), 
and makes a new promise to
execute the post-processing code \ls$comp[v/x]$ once \ls$p$ gets fulfilled with some value \ls$v$.
Under the hood, \ls{process$_{\op}$} is simply a syntactic sugar for
\begin{aeffbox}
\begin{lstlisting}
promise (op _ |-> let x = await p in let y = comp in return <<y>>)
\end{lstlisting}
\end{aeffbox}
While \ls{process$_{\op}$} involves an \ls$await$, it gets 
exposed only after  
\ls$op$ is received, but by that time \ls$p$ will have been fulfilled with some  
\ls$v$ by an earlier interrupt handler, and thus \ls$await$ can reduce.

Returning to post-processing a list of integers promised by some interrupt
handler, below is an example showing the use of the \ls{process$_{\op}$}
combinator and how to \emph{chain together multiple post-processing
computations} (filtering, folding, and issuing a signal), in the same spirit as
how one is taught to program compositionally with futures and
promises~\cite{Haller:Futures}:
\begin{aeffbox}
\begin{lstlisting}
let p = promise (op x |-> initialHandler) in 
...
let q = process$_{\op}$ p with (<<is>> |-> filter (fun i |-> i > 0) is) in 
let r = process$_{\op}$ q with (<<js>> |-> fold (fun j j' |-> j * j') 1 js) in 
process$_{\op}$ r with (<<k>> |-> send productOfPositiveElements k);
...
\end{lstlisting}
\end{aeffbox}
For this to work, it is crucial that incoming interrupts behave 
like (deep) effect handling and propagate into continuations (see \autoref{sec:basic-calculus:semantics:computations})
so that all three post-processing computations get executed, in their program order, 
when an interrupt \ls$op$ is propagated to the program.


\section{Conclusion}
\label{sec:conclusion}

We have shown how to incorporate asynchrony within 
algebraic effects, by decoupling 
the execution of operation calls into signalling that an operation's implementation 
needs to be executed, and interrupting a running computation with the operation's result, 
to which it can react by installing interrupt handlers.
We have shown that our approach is flexible enough that not all signals have to have a matching 
interrupt, and vice versa, allowing us to also model spontaneous behaviour, such as a user 
clicking a button or the environment pre-empting a thread. We have formalised these ideas in a small 
calculus, called \lambdaAEff, and demonstrated its flexibility on a number of examples.
We have also accompanied the paper with an \pl{Agda} formalisation of \lambdaAEff's type 
safety and a prototype implementation of \lambdaAEff.

Compared to our original work~\cite{Ahman:POPL}, in this extended version we 
have simplified the meta-theory of \lambdaAEff, removed the reliance on general recursion 
for reinstalling interrupt handlers, added a notion of state to reinstallable interrupt handlers, 
and extended \lambdaAEff~with higher-order signal and interrupt 
payloads, and with dynamic process creation.
However, various future work directions still remain. We discuss these and related work below.

\paragraph{Asynchronous Effects}
As asynchrony is desired in practice, it is no surprise that \pl{Koka} \cite{Leijen:AsyncAwait} 
and \pl{OCaml} \cite{Dolan:MulticoreOCaml,DBLP:conf/pldi/Sivaramakrishnan21}, the two 
largest implementations of algebraic effects and effect handlers, have been extended accordingly. 
In \pl{Koka}, algebraic operations  
reify their continuation into an explicit callback structure that is then dispatched to a primitive 
such as \lstinline{setTimeout} in its \pl{Node.JS} backend. In \pl{OCaml}, one writes effectful 
operations in a direct style, but then uses handlers to access the actual asynchronous I/O through 
calls to an external library such as \pl{libuv}. Both approaches thus \emph{delegate} the actual 
asynchrony to existing concepts in their backends. In contrast, using \lambdaAEff, we 
can express such backend features solely within the core calculus and the prototype implementation of it.

Further, in \lambdaAEff, we avoid having to manually use (un)masking to 
disable asynchronous effects in unwanted places, which can be a very 
tricky business to get right~\cite{Dolan:MulticoreOCaml}.
Instead, by design, interrupts in \lambdaAEff~\emph{never} 
influence running code unless the code has an explicit interrupt handler installed, 
and they \emph{always} wait for any potential handler to present itself during
execution (recall that they get discarded only when reaching a $\tmkw{return}$).

Finally, it is also worth discussing how signals and interrupts in
\lambdaAEff~compare to asynchronous exceptions, e.g., as found in
Haskell~\cite{Marlow:AsyncExceptions}. The two mechanisms are similar in that
both are issued outside of the running process. While asynchronous exceptions
are thrown to a specific thread, we can simulate this in our broadcast-based
semantics by carrying extra identifying information in signal and interrupt
payloads, as discussed in \autoref{sec:applications:guarded-handlers}. There is
however a crucial difference between the two approaches: while interrupts only
affect a given computation when a matching interrupt handler is installed, and
they get always discarded when they reach the program's $\tmkw{return}$ clause,
then asynchronous exceptions behave in the exact opposite way, causing the
program to stop with a thrown exception unless the asynchronous exception is
caught and handled away by the programmer.

\paragraph{Message-Passing}
While in this paper we have focussed on the foundations of asynchrony in the 
context of programming with algebraic effects, the ideas we propose have also many common 
traits with concurrency models based on \emph{message-passing}, 
such as the Actor model \cite{Hewitt:Actors}, the $\pi$-calculus \cite{Milner:PiCalculus}, 
and the join-calculus \cite{FournetGonthier:JoinCalculus}, just to name a few.
Namely, one can view the issuing of a signal $\tmopout{op}{V}{M}$ as sending a message, 
and handling an interrupt $\tmopin{op}{W}{M}$ as receiving a message, both along a channel
named $\op$. 
In fact, we believe that in our prototype implementation we could replace the semantics 
presented in the paper with an equivalent one based on shared channels
(one for each $\op$), to which the installed interrupt handlers could subscribe to.
Instead of propagating signals first out of and then back into processes, they would then 
be sent directly to channels where interrupt handlers immediately receive them, 
drastically reducing the cost of communication.

Comparing \lambdaAEff~to the Actor model, we see that 
the $\tmrun M$ processes evolve in their own bubbles, and only communicate with other
processes via signals and interrupts, similarly to actors.
However, in contrast to messages not being required to be ordered 
in the Actor model, in our parallel composition operation $\tmpar P Q$, the process $Q$ receives 
interrupts in the same order as the respective signals are issued by $P$ 
(and vice versa). This communication ordering could be relaxed by allowing 
signals to be hoisted out of computations from deeper than just the top level, or by
extending the operational semantics of \lambdaAEff~with commutativity rules for signals.
Another difference with actors is that by default \lambdaAEff-computations react to interrupts  
sequentially---this difference can be remedied by writing programs in a style in which 
interrupts are handled in parallel in dynamically spawned dedicated processes.

It is worth noting that our interrupt handlers are similar to the message receiving construct 
in the $\pi$-calculus, in that they both synchronise with matching incoming
interrupts or messages. However, the two constructs are also different, in that interrupt handlers allow
reductions to take place under them and non-matching interrupts to propagate past them.
Further, our interrupt handlers are also similar to join definitions in the join-calculus, describing
how to react when a corresponding interrupt arrives or join pattern appears, where in both cases
the reaction could involve effectful code. To this end, our interrupt handlers resemble join definitions 
with simple one-channel join patterns. However, where the two constructs differ is that join definitions 
additionally serve to define new (local) channels, similarly to the restriction operator in the $\pi$-calculus, 
whereas we assume a fixed global set of channels (i.e., signal and interrupt names $\op$). 
We expect that extending \lambdaAEff~with local algebraic effects 
\cite{Staton:Instances,Biernacki:AbstractingAlgEffects}
could help us fill this gap between the formalisms.

\paragraph{Scoped Operations}
As noted in \autoref{sec:basic-calculus:semantics:computations}, despite their name, interrupt handlers
behave like algebraic operations, not like effect handlers. However, one should also note 
that they are not conventional operations as they carry computational data that sequential 
composition does not interact with, and that executes only when a corresponding interrupt is received. 

Such generalised operations are known in the literature as \emph{scoped operations}~\cite{Pirog:ScopedOperations},  
a leading example of which is $\tmspawn{M}{N}$.
Further recalling \autoref{sec:basic-calculus:semantics:computations}, despite their appearance, 
incoming interrupts behave computationally like effect handling, not like algebraic operations. 
In fact, it turns out they correspond to effect handling 
induced by an instance of \emph{scoped effect handlers} \cite{Pirog:ScopedOperations}.
Compared to ordinary effect handlers, scoped effect handlers explain both 
how to interpret operations and their scopes. In our setting, this 
corresponds to triggering interrupt handlers and executing the corresponding handler code.

It would be interesting to extend \lambdaAEff~both with 
scoped operations having more general signatures, and with effect handlers 
for them, e.g., to allow
preventing the propagation of incoming interrupts into continuations, discarding the continuation 
of a cancelled remote call, and techniques such as masking or reordering interrupts
according to priority levels.

\paragraph{Denotational Semantics}
In this paper we study only the operational side of \lambdaAEff, 
and leave developing its denotational semantics for the future.
In light of how we have motivated the \lambdaAEff-specific programming 
constructs, and based on the above discussion, we expect the denotational semantics 
to take the form of an algebraically natural \emph{monadic semantics}, where the monad would 
be given by an instance of the one studied in the case of scoped operations~\cite{Pirog:ScopedOperations}
(quotiented by the commutativity of signals and interrupt handlers, 
and extended with nondeterminism to model different evaluation
outcomes). Incoming interrupts would be modelled as homomorphisms
induced by scoped algebras, while for parallel composition, we could
consider all nondeterministic interleavings of (the outgoing signals of) individual computations,
similarly to how it can be done in the context of general effect handlers~\cite{Plotkin:BinaryHandlers, Lindley:DoBeDoBeDo}.
Finally, we expect to be able to take inspiration for the denotational semantics of the 
promise type from that of modal logics and modal types.

\paragraph{Reasoning About Asynchronous Effects}
In addition to using \lambdaAEff's type-and-effect system only for specification purposes (such as specifying 
that $M : \tycomp{X}{(\emptyset,\{\})}$ raises no signals and installs no interrupt handlers), 
we wish to make further use of it for validating \emph{effect-dependent optimisations} \cite{Kammar:Optimisations}. 
For instance, whenever $M : \tycomp{X}{(\o,\i)}$ and $\i\, (\op) = \bot$, we would like to know  
that $\tmopin{\op}{V}{M} \reduces^* M$. One way to validate such optimisations 
is to develop an adequate denotational semantics, 
and then use a semantic \emph{computational induction} principle \cite{Bauer:EffectSystem,Plotkin:Logic}.
For \lambdaAEff, this would amount to only having to prove the optimisations for return values, signals, 
and interrupt handlers. Another way to validate effect-dependent optimisations would 
be to define a suitable logical 
relation for \lambdaAEff~\cite{Benton:AbstractEffects}.

In addition to optimisations based on \lambdaAEff's existing effect system, 
we plan to refine the current ``broadcast everything everywhere'' communication strategy, e.g., 
by extending process types with \emph{communication protocols} inspired by session types \cite{Honda:LangPrimitives}, or
adding \emph{restriction operations} like in CCS~\cite{DBLP:books/sp/Milner80} and suitably reflecting their use in the effect annotations.

\paragraph{Strong Normalisation for Computations}
In addition to getting an overall more principled core calculus, one of the motivations for introducing 
reinstallable interrupt handlers to \lambdaAEff~and removing general recursion~(compared 
to our original work~\cite{Ahman:POPL}) was that the sequential part of the resulting calculus 
ought to be strongly normalising, i.e., there should be no infinite reduction sequences for 
computations. Intuitively, strong normalisation should follow from interrupt handlers getting 
reinstalled only when a corresponding interrupt is propagated to the computation, and no 
single interrupt can reinstall a particular interrupt handler more than once. We leave making 
this argument formal for future work. We expect to be able to build on $\top\top$-lifting style
logical relation proofs of strong normalisation~\cite{Lindley:TopTopLifting}.

However, even after making the above-mentioned changes to \lambdaAEff, its parallel 
part of course remains non-terminating---simply consider two parallel processes built from 
reinstallable interrupt handlers that indefinitely exchange ping-pong signals with each other, e.g., as
\begin{aeffbox}
\begin{lstlisting}
run (send ping (); promise (pong _ r -> send ping (); r ()))
||
run (send pong (); promise (ping _ r -> send pong (); r ()))
\end{lstlisting}
\end{aeffbox}

\section*{Acknowledgements}

We thank the anonymous reviewers, Otterlo IFIP WG 2.1 meeting participants, 
and Andrej Bauer, Gavin Bierman, Žiga Lukšič, Janez Radešček, and Alex Simpson for their useful feedback.

\bibliography{references}

\newcommand{\etalchar}[1]{$^{#1}$}
\begin{thebibliography}{CLMM20}

\bibitem[AB20]{Ahman:Runners}
D.~Ahman and A.~Bauer.
\newblock Runners in action.
\newblock In {\em Proc. of 29th {E}uropean Symp.~on Programming, {ESOP} 2020},
  volume 12075 of {\em LNCS}, pages 29--55. Springer, 2020.
\newblock \href {https://doi.org/10.1007/978-3-030-44914-8\_2}
  {\path{doi:10.1007/978-3-030-44914-8\_2}}.

\bibitem[AC98]{Amadio:Domains}
R.~M. Amadio and P-L. Curien.
\newblock {\em Domains and {Lambda} {Calculi}}.
\newblock Cambridge Tracts in Theoretical Computer Science. Cambridge
  University Press, 1998.
\newblock \href {https://doi.org/10.1017/CBO9780511983504}
  {\path{doi:10.1017/CBO9780511983504}}.

\bibitem[AFH{\etalchar{+}}18]{Ahman:RecallingWitness}
D.~Ahman, C.~Fournet, C.~Hritcu, K.~Maillard, A.~Rastogi, and N.~Swamy.
\newblock Recalling a witness: foundations and applications of monotonic state.
\newblock {\em Proc. {ACM} Program. Lang.}, 2({POPL}):65:1--65:30, 2018.
\newblock \href {https://doi.org/10.1145/3158153} {\path{doi:10.1145/3158153}}.

\bibitem[Ahm24]{ahman:AeffAgda}
D.~Ahman.
\newblock Agda formalisation of the $\lambda_{\text{\ae}}$-calculus.
\newblock Available at
  \url{https://github.com/danelahman/higher-order-aeff-agda/releases/tag/lmcs},
  2024.

\bibitem[AP21]{Ahman:POPL}
D.~Ahman and M.~Pretnar.
\newblock Asynchronous effects.
\newblock {\em Proc. {ACM} Program. Lang.}, 5({POPL}):1--28, 2021.
\newblock \href {https://doi.org/10.1145/3434305} {\path{doi:10.1145/3434305}}.

\bibitem[BCJ{\etalchar{+}}19]{Bingham:Pyro}
E.~Bingham, J.~P. Chen, M.~Jankowiak, F.~Obermeyer, N.~Pradhan, T.~Karaletsos,
  R.~Singh, P.~Szerlip, P.~Horsfall, and N.~D. Goodman.
\newblock Pyro: Deep universal probabilistic programming.
\newblock {\em J. Mach. Learn. Res.}, 20(1):973--978, January 2019.

\bibitem[BGM19]{Bahr:RATT}
P.~Bahr, C.~Graulund, and R.~E. M{\o}gelberg.
\newblock Simply {RaTT}: a fitch-style modal calculus for reactive programming
  without space leaks.
\newblock {\em Proc. {ACM} Program. Lang.}, 3({ICFP}):109:1--109:27, 2019.
\newblock \href {https://doi.org/10.1145/3341713} {\path{doi:10.1145/3341713}}.

\bibitem[BHN14]{Benton:AbstractEffects}
N.~Benton, M.~Hofmann, and V.~Nigam.
\newblock Abstract effects and proof-relevant logical relations.
\newblock In {\em Proc. of 41st Ann. {ACM} {SIGPLAN-SIGACT} Symp. on Principles
  of Programming Languages, {POPL} 2014}, pages 619--632. {ACM}, 2014.
\newblock \href {https://doi.org/10.1145/2535838.2535869}
  {\path{doi:10.1145/2535838.2535869}}.

\bibitem[BP14]{Bauer:EffectSystem}
A.~Bauer and M.~Pretnar.
\newblock An effect system for algebraic effects and handlers.
\newblock {\em Logical Methods in Computer Science}, 10(4), 2014.
\newblock \href {https://doi.org/10.2168/LMCS-10(4:9)2014}
  {\path{doi:10.2168/LMCS-10(4:9)2014}}.

\bibitem[BP15]{Bauer:AlgebraicEffects}
A.~Bauer and M.~Pretnar.
\newblock Programming with algebraic effects and handlers.
\newblock {\em J. Log. Algebr. Meth. Program.}, 84(1):108--123, 2015.
\newblock \href {https://doi.org/10.1016/j.jlamp.2014.02.001}
  {\path{doi:10.1016/j.jlamp.2014.02.001}}.

\bibitem[BPPS19]{Biernacki:AbstractingAlgEffects}
D.~Biernacki, M.~Pir\'{o}g, P.~Polesiuk, and F.~Sieczkowski.
\newblock Abstracting algebraic effects.
\newblock {\em Proc. {ACM} Program. Lang.}, 3(POPL):6:1--6:28, 2019.
\newblock \href {https://doi.org/10.1145/3290319} {\path{doi:10.1145/3290319}}.

\bibitem[CLMM20]{Convent:DooBeeDooBeeDoo}
L.~Convent, S.~Lindley, C.~McBride, and C.~McLaughlin.
\newblock Doo bee doo bee doo.
\newblock {\em J. Funct. Program.}, 30:e9, 2020.
\newblock \href {https://doi.org/10.1017/S0956796820000039}
  {\path{doi:10.1017/S0956796820000039}}.

\bibitem[Clo18]{Clouston:FitchStyle}
R.~Clouston.
\newblock Fitch-style modal lambda calculi.
\newblock In {\em Proc. of 21st Int. Conf. on Foundations of Software Science
  and Computation Structures, {FOSSACS} 2018}, volume 10803 of {\em LNCS},
  pages 258--275. Springer, 2018.
\newblock \href {https://doi.org/10.1007/978-3-319-89366-2\_14}
  {\path{doi:10.1007/978-3-319-89366-2\_14}}.

\bibitem[DEH{\etalchar{+}}17]{Dolan:MulticoreOCaml}
S.~Dolan, S.~Eliopoulos, D.~Hillerstr{\"{o}}m, A.~Madhavapeddy, K.~C.
  Sivaramakrishnan, and L.~White.
\newblock Concurrent system programming with effect handlers.
\newblock In {\em Revised Selected Papers from 18th Int. Symp. on Trends in
  Functional Programming - 18th International Symposium, {TFP} 2017}, volume
  10788 of {\em Lecture Notes in Computer Science}, pages 98--117. Springer,
  2017.
\newblock \href {https://doi.org/10.1007/978-3-319-89719-6\_6}
  {\path{doi:10.1007/978-3-319-89719-6\_6}}.

\bibitem[FG96]{FournetGonthier:JoinCalculus}
C.~Fournet and G.~Gonthier.
\newblock The reflexive {CHAM} and the join-calculus.
\newblock In {\em Proc. of 23rd {ACM} {SIGPLAN-SIGACT} Symp. on Principles of
  Programming Languages, {POPL}'96}, pages 372--385. ACM, 1996.
\newblock \href {https://doi.org/10.1145/237721.237805}
  {\path{doi:10.1145/237721.237805}}.

\bibitem[GHK{\etalchar{+}}03]{Gierz:ContinuousLattices}
G.~Gierz, K.~H. Hofmann, K.~Keimel, J.~D. Lawson, M.~Mislove, and D.~S. Scott.
\newblock {\em Continuous Lattices and Domains}.
\newblock Number~93 in Encyclopedia of Mathematics and its Applications.
  Cambridge University Press, 2003.
\newblock \href {https://doi.org/10.1017/CBO9780511542725}
  {\path{doi:10.1017/CBO9780511542725}}.

\bibitem[HBS73]{Hewitt:Actors}
C.~Hewitt, P.~Bishop, and R.~Steiger.
\newblock A universal modular {ACTOR} formalism for artificial intelligence.
\newblock In {\em Proc. of 3rd Int. Joint Conf. on Artificial Intelligence,
  {IJCAI}'73}, pages 235--245. Morgan Kaufmann Publishers Inc., 1973.

\bibitem[HPM{\etalchar{+}}20]{Haller:Futures}
P.~Haller, A.~Prokopec, H.~Miller, V.~Klang, R.~Kuhn, and V.~Jovanovic.
\newblock \pl{Scala} documentation: Futures and promises.
\newblock Available online at
  \url{https://docs.scala-lang.org/overviews/core/futures.html}, July 2020.

\bibitem[HPP06]{Hyland:SumAndTensor}
M.~Hyland, G.~D. Plotkin, and J.~Power.
\newblock Combining effects: Sum and tensor.
\newblock {\em Theor. Comput. Sci.}, 357(1-3):70--99, 2006.
\newblock \href {https://doi.org/10.1016/J.TCS.2006.03.013}
  {\path{doi:10.1016/J.TCS.2006.03.013}}.

\bibitem[HVK98]{Honda:LangPrimitives}
K.~Honda, V.~T. Vasconcelos, and M.~Kubo.
\newblock Language primitives and type discipline for structured
  communication-based programming.
\newblock In {\em Proc. of 7th {E}uropean Symp. on Programming, {ESOP} 1998},
  volume 1381 of {\em LNCS}, pages 122--138. Springer, 1998.
\newblock \href {https://doi.org/10.1007/BFb0053567}
  {\path{doi:10.1007/BFb0053567}}.

\bibitem[KLO13]{Kammar:Handlers}
O.~Kammar, S.~Lindley, and N.~Oury.
\newblock Handlers in action.
\newblock In {\em Proc. of 18th {ACM} {SIGPLAN} Int. Conf. on Functional
  Programming, {ICFP} 2013}, pages 145--158. ACM, 2013.
\newblock \href {https://doi.org/10.1145/2500365.2500590}
  {\path{doi:10.1145/2500365.2500590}}.

\bibitem[KP12]{Kammar:Optimisations}
O.~Kammar and G.~D. Plotkin.
\newblock Algebraic foundations for effect-dependent optimisations.
\newblock In {\em Proc. of 39th {ACM} {SIGPLAN-SIGACT} Symp. on Principles of
  Programming Languages, {POPL} 2012}, pages 349--360. {ACM}, 2012.
\newblock \href {https://doi.org/10.1145/2103656.2103698}
  {\path{doi:10.1145/2103656.2103698}}.

\bibitem[Kri13]{Krishnaswami:HOFRP}
N.~R. Krishnaswami.
\newblock Higher-order functional reactive programming without spacetime leaks.
\newblock In {\em Proc. of 18th {ACM} {SIGPLAN} Int. Conf. on Functional
  Programming, {ICFP} 2013}, pages 221--232. ACM, 2013.
\newblock \href {https://doi.org/10.1145/2500365.2500588}
  {\path{doi:10.1145/2500365.2500588}}.

\bibitem[Lei17]{Leijen:AsyncAwait}
D.~Leijen.
\newblock Structured asynchrony with algebraic effects.
\newblock In {\em Proc. of 2nd {ACM} {SIGPLAN} Int. Wksh. on Type-Driven
  Development, {TyDe@ICF} 2017}, pages 16--29. {ACM}, 2017.
\newblock \href {https://doi.org/10.1145/3122975.3122977}
  {\path{doi:10.1145/3122975.3122977}}.

\bibitem[LMM17]{Lindley:DoBeDoBeDo}
S.~Lindley, C.~McBride, and C.~McLaughlin.
\newblock Do be do be do.
\newblock In {\em Proc. of 44th {ACM} {SIGPLAN} Symp. on Principles of
  Programming Languages, {POPL} 2017}, pages 500--514. ACM, 2017.
\newblock \href {https://doi.org/10.1145/3009837.3009897}
  {\path{doi:10.1145/3009837.3009897}}.

\bibitem[LPT03]{Levy:FGCBV}
P.~B. Levy, J.~Power, and H.~Thielecke.
\newblock Modelling environments in call-by-value programming languages.
\newblock {\em Inf. Comput.}, 185(2):182--210, 2003.
\newblock \href {https://doi.org/10.1016/S0890-5401(03)00088-9}
  {\path{doi:10.1016/S0890-5401(03)00088-9}}.

\bibitem[LS05]{Lindley:TopTopLifting}
S.~Lindley and I.~Stark.
\newblock Reducibility and $\top\top$-lifting for computation types.
\newblock In {\em Proc. of 7th Int. Conf. of Typed Lambda Calculi and
  Applications, {TLCA} 2005}, volume 3461 of {\em Lecture Notes in Computer
  Science}, pages 262--277. Springer, 2005.
\newblock \href {https://doi.org/10.1007/11417170\_20}
  {\path{doi:10.1007/11417170\_20}}.

\bibitem[Mil80]{DBLP:books/sp/Milner80}
R.~Milner.
\newblock {\em A Calculus of Communicating Systems}, volume~92 of {\em Lecture
  Notes in Computer Science}.
\newblock Springer, 1980.
\newblock \href {https://doi.org/10.1007/3-540-10235-3}
  {\path{doi:10.1007/3-540-10235-3}}.

\bibitem[MJMR01]{Marlow:AsyncExceptions}
S.~Marlow, S.~L.~Peyton Jones, A.~Moran, and J.~H. Reppy.
\newblock Asynchronous exceptions in {Haskell}.
\newblock In {\em Proc. of of the 2001 {ACM} {SIGPLAN} Conference on
  Programming Language Design and Implementation ({PLDI})}, pages 274--285.
  {ACM}, 2001.
\newblock \href {https://doi.org/10.1145/378795.378858}
  {\path{doi:10.1145/378795.378858}}.

\bibitem[MPW92]{Milner:PiCalculus}
R.~Milner, J.~Parrow, and D.~Walker.
\newblock A calculus of mobile processes, {I}.
\newblock {\em Inf. Comput.}, 100(1):1--40, 1992.
\newblock \href {https://doi.org/10.1016/0890-5401(92)90008-4}
  {\path{doi:10.1016/0890-5401(92)90008-4}}.

\bibitem[{Mur}08]{Murphy:PhDThesis}
T.~{Murphy VII}.
\newblock {\em Modal Types for Mobile Code}.
\newblock PhD thesis, Carnegie Mellon University, 2008.

\bibitem[{OCa}]{OCaml:ReleaseNotes}
{OCaml Development Team}.
\newblock {OCaml} 5.0.0 release notes. Available online at
  \url{https://ocaml.org/releases/5.0.0}.

\bibitem[Plo12]{Plotkin:BinaryHandlers}
G.~D. Plotkin.
\newblock Concurrency and the algebraic theory of effects.
\newblock Invited talk at the 23rd Int. Conf. on Concurrency Theory, {CONCUR}
  2012, 2012.

\bibitem[Pou20]{Poulson:AsyncEffectHandling}
L.~Poulson.
\newblock Asynchronous effect handling.
\newblock Master's thesis, School of Informatics, University of Edinburgh,
  2020.

\bibitem[PP02]{Plotkin:NotionsOfComputation}
G.~D. Plotkin and J.~Power.
\newblock Notions of computation determine monads.
\newblock In {\em Proc. of 5th Int. Conf. on Foundations of Software Science
  and Computation Structures, {FOSSACS} 2002}, volume 2303 of {\em LNCS}, pages
  342--356. Springer, 2002.
\newblock \href {https://doi.org/10.1007/3-540-45931-6_24}
  {\path{doi:10.1007/3-540-45931-6_24}}.

\bibitem[PP03]{Plotkin:GenericEffects}
G.~D. Plotkin and J.~Power.
\newblock Algebraic operations and generic effects.
\newblock {\em Appl. Categorical Struct.}, 11(1):69--94, 2003.
\newblock \href {https://doi.org/10.1023/A:1023064908962}
  {\path{doi:10.1023/A:1023064908962}}.

\bibitem[PP08]{Plotkin:Logic}
G.~D. Plotkin and M.~Pretnar.
\newblock A logic for algebraic effects.
\newblock In {\em Proc. of 23th Ann. {IEEE} Symp. on Logic in Computer Science,
  {LICS} 2008}, pages 118--129. IEEE, 2008.
\newblock \href {https://doi.org/10.1109/LICS.2008.45}
  {\path{doi:10.1109/LICS.2008.45}}.

\bibitem[PP13]{Plotkin:HandlingEffects}
G.~D. Plotkin and M.~Pretnar.
\newblock Handling algebraic effects.
\newblock {\em Logical Methods in Computer Science}, 9(4:23), 2013.
\newblock \href {https://doi.org/10.2168/LMCS-9(4:23)2013}
  {\path{doi:10.2168/LMCS-9(4:23)2013}}.

\bibitem[Pre15]{Pretnar:Tutorial}
M.~Pretnar.
\newblock An introduction to algebraic effects and handlers. {Invited tutorial
  paper}.
\newblock {\em Electr. Notes Theor. Comput. Sci.}, 319:19--35, 2015.
\newblock \href {https://doi.org/10.1016/j.entcs.2015.12.003}
  {\path{doi:10.1016/j.entcs.2015.12.003}}.

\bibitem[Pre24]{pretnar21:AEff}
M.~Pretnar.
\newblock Programming language {\textsc{{\ae}ff}}.
\newblock Available at
  \url{https://github.com/matijapretnar/aeff/releases/tag/lmcs}, 2024.

\bibitem[PSWJ18]{Pirog:ScopedOperations}
M.~Pir{\'{o}}g, T.~Schrijvers, N.~Wu, and M.~Jaskelioff.
\newblock Syntax and semantics for operations with scopes.
\newblock In {\em Proc. of 33rd Annual {ACM/IEEE} Symp. on Logic in Computer
  Science, {LICS} 2018}, pages 809--818. {ACM}, 2018.
\newblock \href {https://doi.org/10.1145/3209108.3209166}
  {\path{doi:10.1145/3209108.3209166}}.

\bibitem[Rep93]{Reppy:CML}
J.~H. Reppy.
\newblock Concurrent {ML:} design, application and semantics.
\newblock In {\em Functional Programming, Concurrency, Simulation and Automated
  Reasoning: International Lecture Series 1991-1992, McMaster University,
  Hamilton, Ontario, Canada}, volume 693 of {\em Lecture Notes in Computer
  Science}, pages 165--198. Springer, 1993.
\newblock \href {https://doi.org/10.1007/3-540-56883-2\_10}
  {\path{doi:10.1007/3-540-56883-2\_10}}.

\bibitem[Sch02]{Schwinghammer:Thesis}
J.~Schwinghammer.
\newblock A concurrent lambda-calculus with promises and futures.
\newblock Master's thesis, Programming Systems Lab, Universit{\"a}t des
  Saarlandes, 2002.

\bibitem[SDW{\etalchar{+}}21]{DBLP:conf/pldi/Sivaramakrishnan21}
K.~C. Sivaramakrishnan, Stephen Dolan, Leo White, Tom Kelly, Sadiq Jaffer, and
  Anil Madhavapeddy.
\newblock Retrofitting effect handlers onto {OCaml}.
\newblock In {\em {PLDI}}, pages 206--221. {ACM}, 2021.
\newblock \href {https://doi.org/10.1145/3453483.3454039}
  {\path{doi:10.1145/3453483.3454039}}.

\bibitem[Sta13]{Staton:Instances}
S.~Staton.
\newblock Instances of computational effects: An algebraic perspective.
\newblock In {\em Proc. of 28th Ann. {ACM/IEEE} Symp. on Logic in Computer
  Science, {LICS} 2013}, pages 519--519. IEEE, 2013.
\newblock \href {https://doi.org/10.1109/LICS.2013.58}
  {\path{doi:10.1109/LICS.2013.58}}.

\bibitem[Sta15]{Staton:AlgEffQuantum}
S.~Staton.
\newblock Algebraic effects, linearity, and quantum programming languages.
\newblock In {\em Proc. of 42nd Annual {ACM} {SIGPLAN-SIGACT} Symp. on
  Principles of Programming Languages, {POPL} 2015}, pages 395--406. {ACM},
  2015.
\newblock \href {https://doi.org/10.1145/2676726.2676999}
  {\path{doi:10.1145/2676726.2676999}}.

\bibitem[WF94]{Wright:SynAppTypeSoundness}
A.~K. Wright and M.~Felleisen.
\newblock A syntactic approach to type soundness.
\newblock {\em Information and Computation}, 115(1):38--94, 1994.
\newblock \href {https://doi.org/10.1006/inco.1994.1093}
  {\path{doi:10.1006/inco.1994.1093}}.

\bibitem[WKS22]{Wadler:PLFA}
P.~Wadler, W.~Kokke, and J.~G. Siek.
\newblock {\em Programming Language Foundations in {A}gda}.
\newblock August 2022.
\newblock URL: \url{https://plfa.inf.ed.ac.uk/22.08/}.

\end{thebibliography}

\appendix


\begingroup
\allowdisplaybreaks

\section{The Full Calculus for Higher-Order Asynchronous Effects}
\label{sec:appendix}

In this appendix we present \lambdaAEff~with all the higher-order 
extensions discussed in \autoref{sec:higher-order-extensions}.


\subsection{Terms}

\begin{center}
  \small
  \begin{align*}
  \intertext{\textbf{Values}}
  V, W
  \bnfis& x                                       & &\text{variable} \\
  \bnfor& \tmunit \bnfor\! \tmpair{V}{W}                                & &\text{unit and pairing} \\
  \bnfor& \tminl[Y]{V} \bnfor\! \tminr[X]{V}    & &\text{left and right injections} \\
  \bnfor& \tmfun{x : X}{M}                        & &\text{function abstraction} \\
  \bnfor& \tmpromise V                            & &\text{fulfilled promise} \\
  \bnfor& \tmbox V                                   & &\text{boxed value}
  \\[1ex]
  \intertext{\textbf{Computations}}
  M, N
  \bnfis& \tmreturn{V}                            & &\text{returning a value} \\
  \bnfor& \tmlet{x}{M}{N}          & &\text{sequencing} \\
  \bnfor& V\,W                                    & &\text{function application} \\
  \bnfor& \tmmatch{V}{\tmpair{x}{y} \mapsto M}    & &\text{product elimination} \\
  \bnfor& \tmmatch[\tycomp{Z}{(\o,\i)}]{V}{}                        & &\text{empty elimination} \\
  \bnfor& \tmmatch{V}{\tminl{x} \mapsto M, \tminr{y} \mapsto N}
                                                  & &\text{sum elimination} \\
  \bnfor& \tmopout{op}{V}{M}       & &\text{outgoing signal} \\
  \bnfor& \tmopin{op}{V}{M}          & &\text{incoming interrupt} \\
  \bnfor& \tmwithrest[S]{op}{x}{r}{s}{M}{V}{p}{N}      & &\text{stateful reinstallable interrupt handler} \\
  \bnfor& \tmawait{V}{x}{M}             & &\text{awaiting a promise to be fulfilled} \\
  \bnfor& \tmunbox{V}{x}{M}            & &\text{unboxing a mobile value} \\
  \bnfor& \tmspawn{M}{N}                & &\text{dynamic process creation}
  \\[1ex]
  \intertext{\textbf{Processes}}
    P, Q
  \bnfis& \tmrun M                              & &\text{individual computation} \\
  \bnfor& \tmpar P Q                         & &\text{parallel composition} \\
  \bnfor& \tmopout{op}{V}{P}            & &\text{outgoing signal} \\
  \bnfor& \tmopin{op}{V}{P}              & &\text{incoming interrupt}
  \end{align*}
\end{center}


\subsection{Types}

\begin{center}
  \small
  \begin{align*}
  \text{Mobile type $A$, $B$}
  \bnfis& \tybase \bnfor \tyunit \bnfor \tyempty \bnfor \typrod{A}{B} \bnfor \tysum{A}{B} \bnfor \tybox X
  \\[1ex]
  \text{Signal or interrupt signature:}
  \phantom{\bnfis}& \op : A_\op
  \\[1ex]
  \text{Value type $X$, $Y$, $S$}
  \bnfis& A \bnfor \typrod{X}{Y} \bnfor \tysum{X}{Y} \bnfor \tyfun{X}{\tycomp{Y}{(\o,\i)}} \bnfor \typromise{X}
  \\[1ex]
  \text{Computation type:}
  \phantom{\bnfis}& \tycomp{X}{(\o,\i)}
  \\[1ex]
  \text{Process type $\tyC$, $\tyD$}
  \bnfis& \tyrun X \o \i \bnfor \typar \tyC \tyD
  \\[1ex]
  \text{Typing context $\Gamma$}
  \bnfis& \cdot \bnfor \Gamma, x \of X \bnfor \Gamma, \ctxlock
  \end{align*}
\end{center}


\subsection{Type System}

\begin{center}
  \small
  \begin{mathparpagebreakable}
  \textbf{Values}
  \\
  \coopinfer{TyVal-Var}{
    X~\text{is mobile} \quad \text{or} \quad \ctxlock \not\in \Gamma'
  }{
    \Gamma, x \of X, \Gamma' \types x : X
  }
  \qquad
  \coopinfer{TyVal-Unit}{
  }{
    \Gamma \types \tmunit : \tyunit
  }
  \qquad
  \coopinfer{TyVal-Pair}{
    \Gamma \types V : X \\
    \Gamma \types W : Y
  }{
    \Gamma \types \tmpair{V}{W} : \typrod{X}{Y}
  }
  \\
  \coopinfer{TyVal-Inl}{
    \Gamma \types V : X
  }{
    \Gamma \types \tminl[Y]{V} : X + Y
  }
  \qquad
  \coopinfer{TyVal-Inr}{
    \Gamma \types W : Y
  }{
    \Gamma \types \tminr[X]{W} : X + Y
  }
  \qquad
  \coopinfer{TyVal-Fun}{
    \Gamma, x \of X \types M : \tycomp{Y}{(\o,\i)}
  }{
    \Gamma \types \tmfun{x : X}{M} : \tyfun{X}{\tycomp{Y}{(\o,\i)}}
  }
  \\
  \coopinfer{TyVal-Promise}{
    \Gamma \types V : X
  }{
    \Gamma \types \tmpromise V : \typromise X
  }
  \qquad
  \coopinfer{TyVal-Box}{
    \Gamma, \ctxlock \types V : X
  }{
    \Gamma \types \tmbox V : \tybox X
  }
  \end{mathparpagebreakable}  
  \mbox{}
  
  \begin{mathparpagebreakable}
  \textbf{Computations}
  \\
  \coopinfer{TyComp-Return}{
    \Gamma \types V : X
  }{
    \Gamma \types \tmreturn{V} : \tycomp{X}{(\o,\i)}
  }
  \qquad
  \coopinfer{TyComp-Let}{
    \Gamma \types M : \tycomp{X}{(\o,\i)}
    \\
    \Gamma, x \of X \types N : \tycomp{Y}{(\o,\i)}
  }{
    \Gamma \types
    \tmlet{x}{M}{N} : \tycomp{Y}{(\o,\i)}
  }
  \\
  \coopinfer{TyComp-Apply}{
    \Gamma \types V : \tyfun{X}{\tycomp{Y}{(\o,\i)}} \\
    \Gamma \types W : X
  }{
    \Gamma \types \tmapp{V}{W} : \tycomp{Y}{(\o,\i)}
  }
  \qquad
  \coopinfer{TyComp-MatchPair}{
    \Gamma \types V : \typrod{X}{Y} \\
    \Gamma, x \of X, y \of Y \types M : \tycomp{Z}{(\o,\i)}
  }{
    \Gamma \types \tmmatch{V}{\tmpair{x}{y} \mapsto M} : \tycomp{Z}{(\o,\i)}
  }
  \\
  \coopinfer{TyComp-MatchEmpty}{
    \Gamma \types V : \tyempty
  }{
    \Gamma \types \tmmatch[\tycomp{Z}{(\o,\i)}]{V}{} : \tycomp{Z}{(\o,\i)}
  }
  \qquad
  \coopinfer{TyComp-MatchSum}{
    \Gamma \types V : X + Y \\\\
    \Gamma, x \of X \types M : \tycomp{Z}{(\o,\i)} \\
    \Gamma, y \of Y \types N : \tycomp{Z}{(\o,\i)} \\
  }{
    \Gamma \types \tmmatch{V}{\tminl{x} \mapsto M, \tminr{y} \mapsto N} : \tycomp{Z}{(\o,\i)}
  }
  \\
  \coopinfer{TyComp-Signal}{
    \op \in \o \\
    \Gamma \types V : A_\op \\
    \Gamma \types M : \tycomp{X}{(\o,\i)} 
  }{
    \Gamma \types \tmopout{op}{V}{M} : \tycomp{X}{(\o,\i)}
  }
  \qquad
  \coopinfer{TyComp-Interrupt}{
    \Gamma \types V : A_\op \\
    \Gamma \types M : \tycomp{X}{(\o,\i)} 
  }{
    \Gamma \types \tmopin{op}{V}{M} : \tycomp{X}{\opincomp {op} (\o,\i)}
  }
  \\
  \coopinfer{TyComp-ReStPromise}{
    ({\o'} , {\i'}) \mathrel{\order{O \times I}} \i\, (\op) \\
    \Gamma, x \of A_\op, r \of \tyfun{S}{\tycomp{\typromise X}{\big(\emptyset, \{ \op \mapsto ({\o'} , {\i'}) \}\big)}}, s \of S \types M : \tycomp{\typromise X}{(\o',\i')} \\
    \Gamma \types V : S \\
    \Gamma, p \of \typromise X \types N : \tycomp{Y}{(\o,\i)} 
  }{
    \Gamma \types \tmwithrest[S]{op}{x}{r}{s}{M}{V}{p}{N} : \tycomp{Y}{(\o,\i)}
  }        
  \\
  \coopinfer{TyComp-Await}{
    \Gamma \types V : \typromise X \\
    \Gamma, x \of X \types M : \tycomp{Y}{(\o,\i)} 
  }{
    \Gamma \types \tmawait{V}{x}{M} : \tycomp{Y}{(\o,\i)}
  }
  \qquad
  \coopinfer{TyComp-Unbox}{
     \Gamma \types V : \tybox X \\
     \Gamma, x \of X \types M : \tycomp{Y}{(\o, \i)}
   }{
     \Gamma \types \tmunbox{V}{x}{M} : \tycomp{Y}{(\o, \i)}
   }
  \\
  \coopinfer{TyComp-Spawn}{
     \Gamma, \ctxlock \types M : \tycomp{X}{(\o, \i)} \\
     \Gamma \types N : \tycomp{Y}{(\o', \i')}
   }{
     \Gamma \types \tmspawn{M}{N} : \tycomp{Y}{(\o', \i')}
   }
  \qquad
   \coopinfer{TyComp-Subsume}{
      \Gamma \types M : \tycomp{X}{(\o, \i)} \\
      (\o,\i) \order {O \times I} (\o',\i')
    }{
      \Gamma \types M : \tycomp{X}{(\o', \i')}
    }
  \end{mathparpagebreakable}  
  \mbox{}
  
  \begin{mathparpagebreakable}
  \textbf{Processes}
  \\
  \coopinfer{TyProc-Run}{
    \Gamma \types M : \tycomp{X}{(\o,\i)}
  }{
    \Gamma \types \tmrun{M} : \tyrun{X}{\o}{\i}
  }
  \qquad
  \coopinfer{TyProc-Par}{
    \Gamma \types P : \tyC \\
    \Gamma \types Q : \tyD
  }{
    \Gamma \types \tmpar{P}{Q} : \typar{\tyC}{\tyD}
  }
  \\
  \coopinfer{TyProc-Signal}{
    \op \in \mathsf{signals\text{-}of}{(\tyC)} \\
    \Gamma \types V : A_\op \\
    \Gamma \types P : \tyC 
  }{
    \Gamma \types \tmopout{op}{V}{P} : \tyC
  }
  \qquad
  \coopinfer{TyProc-Interrupt}{
    \Gamma \types V : A_\op \\
    \Gamma \types P : \tyC 
  }{
    \Gamma \types \tmopin{op}{V}{P} : \opincomp{op}{\tyC}
  }  
  \end{mathparpagebreakable}
\end{center}


\subsection{Small-Step Operational Semantics of Computations}

\begin{center}
  \small
  \begin{align*}
    \intertext{\textbf{Standard computation rules}}
    \tmapp{(\tmfun{x \of X}{M})}{V} &\reduces M[V/x]
    \\[0.5ex]
    \tmlet{x}{(\tmreturn V)}{N} &\reduces N[V/x]
    \\[0.5ex]
    \tmmatch{\tmpair{V}{W}}{\tmpair{x}{y} \mapsto M} &\reduces M[V/x, W/y]
    \\[0.5ex]
    \mathllap{
      \tmmatch{(\tminl[Y]{V})}{\tminl{x} \mapsto M, \tminr{y} \mapsto N} 
    } &\reduces
    M[V/x]
    \\[0.5ex]
    \mathllap{
      \tmmatch{(\tminr[X]{W})}{\tminl{x} \mapsto M, \tminr{y} \mapsto N}
    } &\reduces
    N[W/y]
    \\[1ex]
    \intertext{\textbf{Algebraicity of signals, interrupt handlers, awaiting, and process creation}}
    \tmlet{x}{(\tmopout{op}{V}{M})}{N} &\reduces \tmopout{op}{V}{\tmlet{x}{M}{N}}
    \\[0.5ex]
    \tmlet{x}{(\tmwithrest[S]{op}{y}{r}{s}{M}{V}{p}{N_1})}{N_2} &\reduces \\[-0.5ex]
        \tmwithrest[S]{op}{y}{&r}{s}{M}{V}{p}{(\tmlet{x}{N_1}{N_2})}
    \\[0.5ex]
    \tmlet{x}{(\tmawait{V}{y}{M})}{N} &\reduces \tmawait{V}{y}{(\tmlet{x}{M}{N})}
    \\[0.5ex]
    \tmlet{x}{(\tmspawn{M}{N_1})}{N_2} &\reduces \tmspawn{M}{\tmlet{x}{N_1}{N_2}}
    \\[1ex]
    \intertext{\textbf{Commutativity of signals and process creation with interrupt handlers}}
    \tmwithrest[S]{op}{x}{r}{s}{M}{V}{p}{\tmopout{op'}{W}{N}} &\reduces \\[-0.5ex]
        \tmopout{op'}{W}{&\, \tmwithrest[S]{op}{x}{r}{s}{M}{V}{p}{N}}
    \\[0.5ex]
    \tmwithrest[S]{op}{x}{r}{s}{M_1}{V}{p}{\tmspawn{M_2}{N}} &\reduces \\[-0.5ex]
        \tmspawn{&M_2}{\tmwithrest[S]{op}{x}{r}{s}{M_1}{V}{p}{N}}
    \\[1ex]
    \intertext{\textbf{Interrupt propagation}}
    \tmopin{op}{V}{\tmreturn W} &\reduces \tmreturn W
    \\[0.5ex]
    \tmopin{op}{V}{\tmopout{op'}{W}{M}} &\reduces \tmopout{op'}{W}{\tmopin{op}{V}{M}}
    \\[0.5ex]
    \tmopin{op}{V}{\tmwithrest[S]{op}{x}{r}{s}{M}{W}{p}{N}} &\reduces \tmlet{p}{M\big[V/x , R/r , W/s\big]}{\tmopin{op}{V}{N}}
    \\[-0.5ex]
    \intertext{\hfill\textbf{where} $R \defeq \tmfun{s' \of S}{\tmwithrest[S]{op}{x}{r}{s}{M}{s'}{p}{\tmreturn p}}$}
    \\[-3ex]
    \tmopin{op'}{V}{\tmwithrest[S]{op}{x}{r}{s}{M}{W}{p}{N}} &\reduces \\[-0.5ex]
        \tmwithrest[S]{op}{x}{r}{s}{&\,M}{W}{p}{\tmopin{op'}{V}{N}}
    \quad {\color{rulenameColor}(\op \neq \op')}
    \\[0.5ex]
    \tmopin{op}{V}{\tmawait{W}{x}{M}} &\reduces \tmawait{W}{x}{\tmopin{op}{V}{M}}
    \\[0.5ex]
    \tmopin{op}{V}{\tmspawn{M}{N}} &\reduces \tmspawn{M}{\tmopin{op}{V}{N}}
    \\[1ex]
    \intertext{\textbf{Awaiting a promise to be fulfilled}}
    \tmawait{\tmpromise V}{x}{M} &\reduces M[V/x]
    \\[1ex]
    \intertext{\textbf{Unboxing a mobile value}}
    \tmunbox{\tmbox V}{x}{M} &\reduces M[V/x]
  \end{align*}
  \begin{gather*}
    \intertext{\textbf{Evaluation context rule}}
    \coopinfer{}{
      M \reduces N
    }{
      \E[M] \reduces \E[N]
    }
  \\[1ex]
  \begin{align*}
    \intertext{\textbf{where}}
    \E
    \bnfis& \, [~] 
    \bnfor \tmlet{x}{\E}{N} 
    \bnfor \tmopout{op}{V}{\E}
    \bnfor \tmopin{op}{V}{\E} \\
    \bnfor& \, \tmwithrest[S]{op}{x}{r}{s}{M}{V}{p}{\E}
    \bnfor \tmspawn{M}{\E}  
  \end{align*}
  \end{gather*}    
\end{center}


\subsection{Small-Step Operational Semantics of Processes}

\begin{center}
\small
  \begin{align*}
  \intertext{\textbf{Individual computations}}
    \coopinfer{}{
      M \reduces N
    }{
      \tmrun M \reduces \tmrun N
    }
  \end{align*}
  \begin{align*}
  \intertext{\textbf{Signal hoisting}}
  \tmrun {(\tmopout{op}{V}{M})}  &\reduces \tmopout{op}{V}{\tmrun M}
  \\[1ex]
  \intertext{\textbf{Process creation}}
  \tmrun{(\tmspawn{M}{N})} &\reduces \tmpar{\tmrun{M}}{\tmrun{N}}
  \\[1ex]
  \intertext{\textbf{Broadcasting}}
  \tmpar{\tmopout{op}{V}{P}}{Q} &\reduces \tmopout{op}{V}{\tmpar{P}{\tmopin{op}{V}{Q}}}
  \\[0.5ex]
  \tmpar{P}{\tmopout{op}{V}{Q}} &\reduces \tmopout{op}{V}{\tmpar{\tmopin{op}{V}{P}}{Q}}
  \\[1ex]
  \intertext{\textbf{Interrupt propagation}}
  \tmopin{op}{V}{\tmrun M} &\reduces \tmrun {(\tmopin{op}{V}{M})}
  \\[0.5ex]
  \tmopin{op}{V}{\tmpar P Q} &\reduces \tmpar {\tmopin{op}{V}{P}} {\tmopin{op}{V}{Q}}
  \\[0.5ex]
  \tmopin{op}{V}{\tmopout{op'}{W}{P}} &\reduces \tmopout{op'}{W}{\tmopin{op}{V}{P}}
  \end{align*}
  \begin{gather*}
    \intertext{\quad\textbf{Evaluation context rule}}
    \quad
    \coopinfer{}{
      P \reduces Q
    }{
      \F[P] \reduces \F[Q]
    }
  \\[1ex]
  \begin{align*}
  \intertext{\textbf{where}}
  \text{$\F$}
  \bnfis& [~]
  \bnfor \tmpar \F Q \bnfor\! \tmpar P \F
  \bnfor \tmopout{op}{V}{\F}
  \bnfor \tmopin{op}{V}{\F}
  \end{align*}
  \end{gather*}
\end{center}


\subsection{Process Type Reduction}

\begin{center}
  \small
  \begin{mathparpagebreakable}
  \coopinfer{}{
  }{
    \tyrun{X}{\o}{\i} \tyreduces \tyrun{X}{\o}{\i} 
  }
  \quad
  \coopinfer{}{
  }{
    X \att (\opincompp {ops} {(\o , \i)}) \tyreduces X \att (\opincompp {ops} {(\opincomp {op} {(\o , \i)})})
  }
  \quad
  \coopinfer{}{
    \tyC \tyreduces \tyC' \\
    \tyD \tyreduces \tyD'
  }{
    \typar{\tyC}{\tyD} \tyreduces \typar{\tyC'}{\tyD'}
  }
  \\
  \coopinfer{}{
  }{
    X \att {(\o , \i)} \tyreduces
    \typar{(X \att {(\o , \i)})}{(Y \att {(\o' , \i')})}
  }
  \quad
  \coopinfer{}{
  }{
    Y \att {(\o' , \i')} \tyreduces
    \typar{(X \att {(\o , \i)})}{(Y \att {(\o' , \i')})}
  }
  \end{mathparpagebreakable}
\end{center}


\subsection{Result Forms}

\begin{center}
\small
\begin{mathparpagebreakable}
  \textbf{Computations}
  \\
  \coopinfer{}{
    \CompResult {\Psi} {M}
  }{
    \CompResult {\Psi} {\tmopout {op} V M}
  }
  \qquad
  \coopinfer{}{
    \CompResult {\Psi} {N}
  }{
    \CompResult {\Psi} {\tmspawn{M}{N}}
  }
  \qquad
  \coopinfer{}{
    \RunResult {\Psi} {M}
  }{
    \CompResult {\Psi} {M}
  }
  \\
  \coopinfer{}{
  }{
    \RunResult {\Psi} {\tmreturn V}
  }
  \qquad
  \coopinfer{}{
    p \in \Psi
  }{
    \RunResult {\Psi} {\tmawait p x M}
  }
  \\
  \coopinfer{}{
    \RunResult {\Psi \cup \{p\}} {N}
  }{
    \RunResult {\Psi} {\tmwithrest {op} x r s M V p N}
  }
  \\\\
  \textbf{Processes}
  \\
  \coopinfer{}{
    \ProcResult {P}
  }{
    \ProcResult {\tmopout {op} V P}
  }
  \qquad
  \coopinfer{}{
    \ParResult {P}
  }{
    \ProcResult {P}
  }
  \qquad
  \coopinfer{}{
    \RunResult {\emptyset} {M}
  }{
    \ParResult {\tmrun M}
  }
  \qquad
  \coopinfer{}{
    \ParResult P \\
    \ParResult Q
  }{
    \ParResult {\tmpar P Q}
  }
\end{mathparpagebreakable}
\end{center}

\endgroup

\end{document}